\documentclass[english,11pt]{article}
\usepackage[T1]{fontenc}

\usepackage[latin9]{inputenc}
\usepackage{geometry}
\usepackage{amsmath}
\usepackage{amsthm}
\usepackage{amssymb}
\usepackage{fancybox}
\usepackage{float}
\usepackage[normalem]{ulem}
\usepackage{comment}

\makeatletter

\numberwithin{equation}{section}
\numberwithin{figure}{section}
\theoremstyle{plain}
	\newtheorem{theorem}{Theorem}[section]
	\newtheorem{thm}[theorem]{Theorem}

	\newtheorem{lem}[theorem]{Lemma}
	\newtheorem{lemma}[theorem]{Lemma}

 	\newtheorem{cor}[theorem]{Corollary}

	\newtheorem{fact}[theorem]{Fact}

\theoremstyle{definition}
	\newtheorem{definition}[theorem]{Definition}
	\newtheorem{defn}[theorem]{Definition}
	\newtheorem*{remark*}{Remark}

\pdfoutput=1

\newif\ifusebb
\usebbtrue
\ifusebb
  \usepackage{bbm}
  \DeclareSymbolFont{bbold}{U}{bbold}{m}{n}
  \DeclareSymbolFontAlphabet{\mathbbold}{bbold}
  \newcommand{\allones}{\ensuremath{\mathbbm{1}}}
  \newcommand{\allzeros}{\ensuremath{\mathbbold{0}}}
\else
  \newcommand{\allones}{\vec{1}}
  \newcommand{\allzeros}{\vec{0}}
\fi
 
\newcommand{\poly}{\mathrm{poly}}
\usepackage[pdftex,bookmarks,colorlinks]{hyperref}
\usepackage{enumitem}
\setlist[description,1]{align=left,leftmargin=0.1in} 
\setlist[itemize,1]{leftmargin=*}
\usepackage{color}
\usepackage{tikz,wrapfig}

\usepackage[lined,boxed,ruled,norelsize,algo2e]{algorithm2e}

\hypersetup{colorlinks=true,citecolor=red}

\makeatother

\usepackage{babel}

\newenvironment{fminipage}%
  {\begin{Sbox}\begin{minipage}}%
  {\end{minipage}\end{Sbox}\fbox{\TheSbox}}

\newenvironment{algbox}[0]{\vskip 0.2in
\noindent 
\begin{fminipage}{6.3in}
}{
\end{fminipage}
\vskip 0.2in
}

\def\defeq{\stackrel{\mathrm{def}}{=}}

\def\ceil#1{\left\lceil #1 \right\rceil}

\def\abs#1{\left|#1  \right|}

\def\norm#1{\left\| #1 \right\|}

\newcommand{\tsolvelap}{\ensuremath{O\left(m+n2^{O(\sqrt{\log n \log \log n})}\right)\log^{O(1)}\left(n \kappa \epsilon^{-1}\right)}}
\newcommand{\tsolveprev}{\ensuremath{O\left(nm^{3/4}+n^{2/3}m\right)\log^{O(1)} \left(n \kappa \epsilon^{-1}\right)}}
\newcommand{\tsolvelapsimple}{\ensuremath{O\left(m+n^{1+o(1)}\right)\log^{O(1)}\left(n \kappa \epsilon^{-1}\right)}}

\renewcommand\AA{\mvar{A}}
\newcommand\DD{\mvar{D}}

\newcommand\II{\mvar{I}}
\newcommand\MM{\mvar{M}}
\newcommand\LL{\mvar{L}}
\newcommand\UU{\mvar{U}}
\newcommand\XX{\mvar{X}}
\newcommand\ZZ{\mvar{Z}}

\newcommand\MMtil{\mapp{\mvar{M}}}

\newcommand\UUtil{\mapp{\mvar{U}}}

\newcommand\AAcal{\boldsymbol{\mathcal{A}}}
\newcommand\WWhat{\boldsymbol{\mathcal{A}}}

\newcommand\ZZbar{\overline{\mvar{Z}}}
\newcommand\ZZtil{\widetilde{\mvar{Z}}}

\newcommand{\mapp}[1]{\widetilde{#1}}

\newcommand\Otil{\widetilde{O}}

\newcommand{\degrat}{R}

\newcommand{\rp}[1]{{\bf \color{green} Richard: #1}}

\newcommand{\av}[1]{{\bf \color{green} Adrian: #1}}

\newcommand{\sparsifysubgraph}{\textsc{SparsifySubgraph}}
\newcommand{\expanderpartition}{\textsc{FindDecomposition}}
\newcommand{\sparsifysquare}{\textsc{SparsifySquare}}
\newcommand{\sparsifyeulerian}{\textsc{SparsifyEulerian}}
\newcommand{\sparsifyproduct}{\textsc{SparsifyProduct}}

\global\long\def\R{\mathbb{{R}}}

\global\long\def\E{\mathbb{E}}

\global\long\def\boldVar#1{\mathbf{#1}}
\global\long\def\mvar#1{\boldVar{#1}}
\global\long\def\vvar#1{\vec{#1}}

\global\long\def\defeq{\stackrel{\mathrm{{\scriptscriptstyle def}}}{=}}
\newcommand{\otilde}{\widetilde{O}}

\global\long\def\norm#1{\|#1\|}
\global\long\def\normFull#1{\left\Vert #1\right\Vert }

\global\long\def\im#1{im(#1)}
\global\long\def\imFull#1{im\left(#1\right) }

\newcommand{\vb}{\vvar{b}}
\newcommand{\vc}{\vvar{c}}
\newcommand{\vd}{\vvar{d}}
\newcommand{\ve}{\vvar{e}}

\newcommand{\vr}{\vvar{r}}

\newcommand{\vv}{\vvar{v}}

\newcommand{\vx}{\vvar{x}}
\newcommand{\vy}{\vvar{y}}

\newcommand{\vzero}{\allzeros}
\newcommand{\vones}{\allones}
\newcommand{\vindic}[1]{\ve_{i}}

\newcommand{\ma}{\mvar A}
\newcommand{\mb}{\mvar B}
\newcommand{\mc}{\mvar C}
\newcommand{\md}{\mvar D}
\newcommand{\mE}{\mvar E}

\newcommand{\mh}{\mvar H}
\newcommand{\mI}{\mvar I}

\newcommand{\mLL}{\mvar L}
\newcommand{\mL}{\mvar L}
\newcommand{\mm}{\mvar M}
\newcommand{\mn}{\mvar N}

\newcommand{\mr}{\mvar R}
\newcommand{\ms}{\mvar S}

\newcommand{\mU}{\mvar U}
\newcommand{\mv}{\mvar V}

\newcommand{\mx}{\mvar X}
\newcommand{\my}{\mvar Y}
\newcommand{\mz}{\mvar Z}

\newcommand{\mzero}{\mvar 0}
\newcommand{\mlap}{\mathcal{L}}

\newcommand{\mlaphat}{\widehat{\mathcal{L}}}
\newcommand{\mSigma}{\mvar \Sigma}

\newcommand{\mdiag}{\mvar{diag}}
\newcommand{\diag}{\mathrm{diag}}

\newcommand{\sspan}{\mathrm{span}}

\global\long\def\tsolve{\mvar{\mathcal{T}_{\textnormal{solve}}}}

\global\long\def\abs#1{\left|#1\right|}

\global\long\def\tr{\mathrm{tr}}

\global\long\def\ceil#1{\left\lceil #1 \right\rceil }

\global\long\def\nnz{\mathrm{nnz}}

\global\long\def\supp{\mathrm{supp}}
\global\long\def\dist{\mathcal{D}}

\renewcommand{\dagger}{+}
\renewcommand{\intercal}{\top}

\global\long\def\empircalA{\widetilde{\ma}}

\newcommand{\indic}{\ensuremath{\vec{{1}}}}

\newcommand{\lambdanonzero}{\lambda_{*}}

\begin{document}

\title{Almost-Linear-Time Algorithms for Markov Chains \\
and New Spectral Primitives for Directed Graphs
}

\author{Michael B. Cohen\thanks{This material is based upon work supported by the National Science
Foundation under Grant No. 1111109.}\\
MIT\\
micohen@mit.edu\and Jonathan Kelner\footnotemark[1] \\
MIT\\
kelner@mit.edu\and  John Peebles\thanks{This material is based upon work supported by the National Science
Foundation Graduate Research Fellowship under Grant No. 1122374 and
by the National Science Foundation under Grant No. 1065125.}\\
MIT\\
jpeebles@mit.edu\and  Richard Peng\thanks{This material is based upon work supported by the National Science
	Foundation under Grant No. 1637566.}\\
Georgia Tech \\
rpeng@cc.gatech.edu\and  Anup B. Rao\\
Georgia Tech\\
anup.rao@gatech.edu\and  Aaron Sidford\\
Stanford University\\
sidford@stanford.edu\and  Adrian Vladu\footnotemark[1] \\
MIT\\
avladu@mit.edu}

\date{}
\maketitle
\begin{abstract}
In this paper we introduce a notion of spectral approximation for directed graphs. While there are many potential ways one might define approximation for directed graphs, most of them are too strong to allow sparse approximations in general. In contrast, we prove that for our notion of approximation, such sparsifiers do exist, and we show how to compute them in almost linear time. 

Using this notion of approximation, we provide a general framework for solving asymmetric linear systems that is broadly inspired by the work of [Peng-Spielman, STOC`14].
Applying this framework in conjunction with our sparsification algorithm, we obtain an almost-linear-time algorithm for solving directed Laplacian systems associated with Eulerian Graphs. Using this solver in the  recent framework of
[Cohen-Kelner-Peebles-Peng-Sidford-Vladu, FOCS`16],
we obtain almost linear time algorithms for
solving a directed Laplacian linear system,
computing the stationary distribution of a Markov chain,
computing expected commute times in a directed graph, and more. 

For each of these problems, our algorithms improves
the previous best running times of $O((nm^{3/4} + n^{2/3} m) \log^{O(1)} (n \kappa \epsilon^{-1}))$
to $O((m + n2^{O(\sqrt{\log{n}\log\log{n}})})
\log^{O(1)} (n \kappa \epsilon^{-1}))$
where $n$ is the number of vertices in the graph, $m$ is the number of edges, $\kappa$ is a natural condition number associated with the problem, and $\epsilon$ is the desired accuracy. We hope these results open the door for further studies into directed spectral graph theory,
and that they will serve as a stepping stone for designing a new generation of fast algorithms for directed graphs.

\end{abstract}

\vfill

\pagebreak{}

\section{Introduction}
\label{sec:intro}

 In the analysis of Markov chains, there has been a longstanding algorithmic gap between the general case, corresponding to random walks on directed graphs, and the special case of reversible chains, for which the corresponding graph can be taken to be undirected.
 This gap begins with the most basic computational task---computing the stationary distribution---and persists for many of the fundamental problems associated with random walks, such as computing hitting and commute times, escape probabilities, and personalized PageRank vectors.
 In the undirected case, there are algorithms for all of these problems that run in linear or nearly-linear time.
 In the directed case, however, the best algorithms have historically been much slower. Specifically, the best running times were given by a recent precursor to the present paper~\cite{cohen2016faster}, which showed that one could solve these problems on a graph with $n$ vertices and $m$ edges in time $\otilde(nm^{3/4}+n^{2/3}m)$.\footnote{
We use $\otilde$ notation to suppress terms that are polylogarithmic in $n$, the natural condition number of the problem $\kappa$, and the desired accuracy $\epsilon$.
We use the term ``nearly linear'' to refer to algorithms whose running time is $\otilde(m)=m\log^{O(1)}(n\kappa\epsilon^{-1}) $ and ``almost linear'' to refer to algorithms that are linear up to sub-polynomial (but possibly super-logarithmic) factors, i.e.,  whose running time is $O(m(n\kappa\epsilon^{-1})^{o(1)})$.}
 Prior to that work, it was unknown whether one could solve any of them faster than the time needed to solve an arbitrary linear system with the given size and sparsity, i.e. $\Theta(\min(mn,n^\omega))$ time, where $\omega<2.3729$ is the exponent for matrix multiplication.
 
 This gap has its origins in a broader discrepancy between the state of algorithmic spectral graph theory in undirected and directed settings.
 While the undirected case has a richly developed theory and a powerful collection of algorithmic tools, similar results have remained somewhat elusive for directed graphs.
 In particular, the problems mentioned above can be expressed in terms of the   linear algebraic properties of the Laplacian matrix of a graph, and it was shown in~\cite{cohen2016faster} how to reduce all these problems to the solution of a small number of Laplacian linear systems.
 In the undirected case, there has been a tremendously successful line of research on how to use the combinatorial properties of graphs to accelerate the solution of such systems, culminating in algorithms that can solve them in nearly-linear time~\cite{SpielmanTengSolver:journal,KoutisMP10,KoutisMP11,KelnerOSZ13,lee2013efficient,CohenKMPPRX14,PengS14,KyngLPSS16,KyngS16}.
 Unfortunately, these solvers relied heavily on several features that seemed intrinsic to the undirected case and did not appear to be available for directed graphs, thereby precluding an analogous solver for directed Laplacians. 
 In particular, the undirected solvers relied on:
 \begin{description}
 	\item[Knowledge of the kernel/stationary distribution:] Up to a simple rescaling by the vertex degrees, vectors in the kernel of a Laplacian  correspond to stationary distributions of the corresponding random walk.
 	For undirected graphs, the kernel is spanned by the all-ones vector on each of the connected components, so it and the space of stationary distributions can be easily computed in linear time.
 	For directed graphs, however, this is no longer the case, and finding the stationary distribution does not seem to be any easier than the original problem of solving Laplacian linear systems.
 	In fact, while  stationary distributions of random walks on directed graphs have been studied for over 100 years~\cite{BasharinLN04}, and computing them has been extensively investigated in both theory and practice (see e.g.~\cite{Stewart94:book,Saad03:book}),
 	the $\otilde{(nm^{3/4} +n^{2/3}m)}$ result in~\cite{cohen2016faster} was the first to find them in less time than is required to find the kernel of a general matrix.

 	\item[Symmetry and positive semidefiniteness:] Undirected Laplacians are symmetric and positive semidefinite.
 	Essentially every aspect of algorithmic spectral graph theory uses this symmetry to treat the Laplacian as a quadratic form and relies on its expression as a sum of positive semidefinite contributions from each of the edges to analyze its properties.
 	This includes the Laplacians' connection to the graph's cut structure, their relationship to electrical circuits and effective resistances, the notion of graph inequalities and spectral approximation, the combinatorial construction of preconditioners, and the iterative methods used to solve Laplacian systems.
 	On the other hand, directed Laplacians are asymmetric matrices, and their naive symmetrizations are not typically positive semidefinite.

 	\item[Sparsification]  One of the most powerful algorithmic tools in the undirected setting is the ability to construct \emph{sparsifiers}~\cite{BenczurK96,SpielmanT11,BatsonSST13}.
 	These allow one to approximate an arbitrarily dense graph by a sparse graph that has only a slightly super-linear number of edges.
 	The classical notion of cut sparsification requires that the value of every cut in the original graph be approximately preserved in the sparsifier; the more recent notion of spectral sparsification is stronger, and also implies the former property.
 	For directed graphs, it can be shown that, even for the weaker notion of cut sparsification,
 	such sparsifiers do not generally exist. One simple example is the complete bipartite graph.  (See Section~\ref{sub:directed_sparse_hard}.)
 	In fact, it was not known how to define any other useful notion of sparsification
 	for which this would not be the case.
 \end{description}

 In this paper, we show how to cope with these fundamental differences, and begin to address the algorithmic gap between general and reversible Markov chains.
 Our core technical result is the first almost-linear-time solver for directed Laplacian systems.
Using the work from~\cite{cohen2016faster}, this yields the first almost-linear-time algorithms for computing a host of fundamental objects and quantities associated with a random walk on a directed graph, including the stationary distribution, hitting and commute times, escape probabilities, and personalized PageRank vectors.

 More broadly, constructing our solver required the development of directed versions of several foundational tools and techniques from undirected algorithmic spectral graph theory.
 Most notably, and perhaps surprisingly, we show that it is possible to develop a useful notion of spectral approximation and sparsification of directed graphs, and that sparsifiers under this  definition exist and can be constructed efficiently.

 In addition to their direct application to the analysis of Markov chains, we hope that both the solver itself and the sparsification machinery will prove to be useful tools in the further development of fast graph algorithms.
 In the undirected case, sparsifiers have been a core algorithmic tool since
 the early 1990s~\cite{BenczurK96,Karger00:journal,KargerL02,FungHHP11,AbrahamDKKP16:arxiv},
  and fast solvers for undirected Laplacian solvers have recently led to an explosion of algorithms operating in the so-called ``Laplacian Paradigm''~\cite{Teng10}, in both cases leading to asymptotic improvements for many of the core algorithmic problems for undirected graphs.
 Given the success these methods have enjoyed in the case of undirected graphs, we hope that their directed analogues will spark similar progress in the directed setting.

\subsection{Previous Work}
\label{sec:previous_work}
	In this section, we briefly review some of the previous work related to our results and techniques.  Given the extensive prior research on Markov chains, spectral graph theory, sparsification, solving general and Laplacian linear systems, and computing PageRank, we do not attempt to give a comprehensive overview of the literature; instead we simply describe the work that most directly relates to or motivates this paper. 
	
	\subsubsection{Directed Laplacian Systems, Stationary Distributions, and PageRanks}
		The most direct precursor to this work is a recent paper by a subset of the authors~\cite{cohen2016faster}.  As mentioned above, it showed that, by exploiting linear algebraic properties of directed Laplacians, one could obtain faster algorithms for a wide range of problems involving directed random walks.  
		Prior to this paper, it seemed quite possible that the similarities between directed and undirected Laplacians were largely syntactic, and that there was no way to use the structure of directed Laplacians or random walks to obtain asymptotically faster algorithms.
		In particular, despite extensive theoretical and applied work in computer science, mathematics, statistics, and numerical scientific computing, 
		all algorithms that we are aware of prior to~\cite{cohen2016faster} for obtaining high-quality%
        \footnote{By high-quality, we mean that the algorithm should be able to find a solution with error $\epsilon$ in time that is sub-polynomial in $1/\epsilon$, i.e. $(1/\epsilon)^{o(1)}$. For PageRank
         there were some known techniques for achieving better dependence on $n$ and $m$ at the expense of a polynomial dependence on~$1/\epsilon$~\cite{AndersenCL07,ChungZ10,ChungS13,ChungS14}.}
		solutions for directed Laplacian systems, stationary distributions, or personalized PageRank vectors either have a polynomial dependence
		on a condition number or related parameter (such as a random walk's mixing time or PageRank's restart probability),
		or they apply a general-purpose linear algebra routine and thus run in at least the $\Omega(\min(mn,n^\omega))$ time currently required to solve arbitrary linear systems.	

		By showing that this was not the case,~\cite{cohen2016faster} provided the first indication that one could actually use the structure of directed Laplacian systems to accelerate their solution, which provided a strong motivation to see how much of an improvement was possible.  It also
		created hope that the recently successful research program in building and applying fast algorithms for solving (symmetric) Laplacian systems~\cite{SpielmanTengSolver:journal,KoutisMP10,KoutisMP11,KelnerOSZ13,lee2013efficient,PengS14,KyngLPSS16} could be applied to give more direct improvements to running times for solving combinatorial optimization problems on directed graphs.

		In addition to motivating the search for faster Laplacian solvers,~\cite{cohen2016faster} provided a set of reductions that we will directly apply in this paper.  
		In order to prove its results, \cite{cohen2016faster} showed how to reduce a range of algorithmic questions about directed walks, such as computing the stationary distribution,  hitting and commute times, escape probabilities, and personalized PageRank vectors, to solving a small number of linear systems in directed Laplacians.  
		
		It turns out that it is easier to work with such systems in the special case where the graph is Eulerian. 
		One of the main technical tools in~\cite{cohen2016faster} is a reduction to this special case.  They did this by giving an iterative method that solved a general Laplacian system by solving a small number of systems in which the graph is Eulerian.  
	Together, this showed that to solve the aforementioned problems, it suffices to give a solver for Eulerian graphs, and that this only incurs a factor of $\otilde(1)$ overhead.
	It then obtained all of its results by constructing an Eulerian solver that runs in time $\otilde(m^{3/4}n+mn^{2/3})$. 
		In this paper we construct an Eulerian solver that runs in time 
		  $m^{1+o(1)}$ and then just directly apply these reductions to obtain our other results.

		However, while \cite{cohen2016faster} opened the door for further algorithmic improvements in analyzing Markov chains,
		the arguments in it provided little evidence that the running time could be improved to anything approaching what was known in the undirected case.  
		Indeed,  while the techniques in it suggested that it might be possible to obtain further improvements,  %
		even the most optimistic interpretations of the structural results in \cite{cohen2016faster} only gave hope for achieving running times of roughly $\otilde (m \sqrt{n})$. %
	 This would make it no faster than some of the existing algorithms that use undirected Laplacian solvers to solve problems on directed graphs, such as the 	 $\otilde(m^{10/7})$ algorithms for  unit cost maximum flow~\cite{DBLP:conf/focs/Madry13,DBLP:journals/corr/Madry16} and shortest path with negative edge lengths~\cite{DBLP:journals/corr/CohenMSV16}, or the $\otilde (m \sqrt{n})$ type bounds for minimum cost flow~\cite{LeeS14}.
			As such, while this would provide better results for the applications to Markov chains, it would rule out the hope of obtaining improved results for 
these directed problems by replacing the undirected solver with a directed one.

Intuitively, the solver in~\cite{cohen2016faster} worked by showing how one could use the existence of a fast undirected solver to solve directed Laplacians.  For a directed Eulerian Laplacian $\mlap$, it showed that the symmetrized matrix $\mU=(\mlap+\mlap^\top)/2$ is the Laplacian of an undirected graph, and that the symmetric matrix $\mlap^\top \mU^+ \mlap$ was, in a certain sense, reasonably well approximated by $\mU$.
Given a linear system $\mlap\vec{x}=\vec{b}$, one could then form the equivalent system $\mlap^\top \mU^+ \mlap\vec{x}=\mlap^\top \mU^+ \vec{b}$ and use a fast undirected Laplacian solver to apply $\mU^+$. 
One could then hope that the fact that the matrix on the left is somewhat well-approximated by $\mU$ would imply that $\mU^+$ is a sufficiently good preconditioner for it to yield an improved running time.  
It turned out that, while this would actually be the case in exact arithmetic, numerical issues provided a legitimate obstruction. 
This necessitated a more involved scheme, which gave a slightly slower running time of $\otilde(m^{3/4}n+mn^{2/3})$, rather than the roughly $\otilde(n^{2/3} m)$ running time that what would have been achieved by  exact arithmetic. 

The way this algorithm works provides a good intuitive explanation for why one would not expect it to give a solver yielding substantial improvements for combinatorial ``Laplacian Paradigm'' algorithms that rely on undirected solvers.
At its root, the solver from~\cite{cohen2016faster} works by trying to find the right way to ignore the directed structure and solve the system with an undirected solver; thus it is on essentially the same footing as the algorithms it would hope to improve.
The obstructions it faces are rooted in the fact that directed Laplacians are fundamentally not very well-approximated by undirected ones.
In essence, the difference between the solver in this paper and the one presented in~\cite{cohen2016faster} is that, instead of figuring out how to properly neglect the directed structure, the solver we present here intrinsically works with asymmetric (directed) objects, and redevelops the theory from the ground up to properly capture them.

	\subsubsection{Directed Graph Sparsification and Approximation}
	\label{sub:directed_sparse_hard}
		 While sparsification of undirected graphs has been extensively
		studied~\cite{BenczurK96,SpielmanT11,FungHHP11,SpielmanS08:journal,
		BatsonSS12,BatsonSST13,ZhuLO15,LeeS15},
		 there has been very little success extending the notion to directed graphs. 
		In fact, it was not even clear that there should exist a useful definition  under which directed graphs should have sparsifiers with a subquadratic number of edges,
and for many of the natural definitions one would propose, sparsification is provably impossible.  

For instance, a natural first attempt would be to try to generalize the classical notion of cut sparsification  for undirected graphs~\cite{karger1994random,BenczurK96}.  Given any weighted undirected graph $G$, Benczur and Karger showed that one could construct a new graph $H$ with at most $O(n\log n/\epsilon^2)$ edges such that the value of every cut in $G$ is within a multiplicative factor of $1\pm \epsilon$ of its value in $H$.  While this definition makes sense for directed graphs as well, there is no analogous existence result.  Indeed it is not hard to construct directed graphs for which any such approximation must have $\Omega(n^2)$ edges.

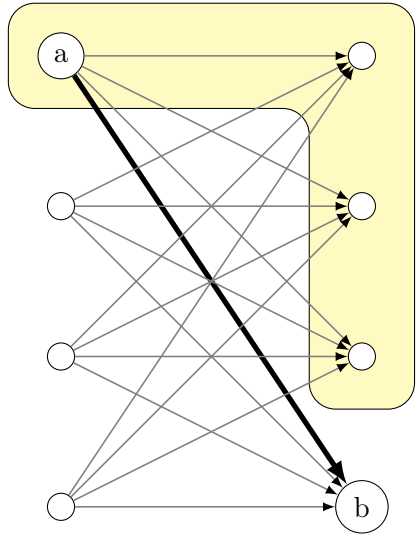
\begin{wrapfigure}{R}{2.8in}

    \begin{center}

\begin{tikzpicture}
	\tikzstyle{vert}=[circle,draw=black,fill=white]
	\tikzstyle{edge}=[-latex,semithick,draw=gray]
	
	\filldraw [rounded corners=10pt, fill=yellow!30] (-3.7, 3.7) -- (1.7, 3.7) -- (1.7, -1.7) -- (0.3, -1.7) -- (0.3, 2.3) -- (-3.7, 2.3) --cycle;
	\node [style=vert] (0) at (-3, 3) {a};
	\node [style=vert] (1) at (-3, 1) {};
	\node [style=vert] (2) at (-3, -1) {};
	\node [style=vert] (3) at (-3, -3) {};
	\node [style=vert] (4) at (1, 3) {};
	\node [style=vert] (5) at (1, 1) {};
	\node [style=vert] (6) at (1, -1) {};
	\node [style=vert] (7) at (1, -3) {b};
	
	\draw [style=edge] (0) to (4);
	\draw [style=edge] (0) to (5);
	\draw [style=edge] (0) to (6);
	\draw [style=edge, line width=2pt,draw=black] (0) to (7);
	\draw [style=edge] (1) to (4);
	\draw [style=edge] (1) to (5);
	\draw [style=edge] (1) to (6);
	\draw [style=edge] (1) to (7);
	\draw [style=edge] (2) to (4);
	\draw [style=edge] (2) to (5);
	\draw [style=edge] (2) to (6);
	\draw [style=edge] (2) to (7);
	\draw [style=edge] (3) to (4);
	\draw [style=edge] (3) to (5);
	\draw [style=edge] (3) to (6);
	\draw [style=edge] (3) to (7);
\end{tikzpicture} 
\end{center}
\caption{\label{fig:bipartite_cut}
An example of the family of cuts described in Equation~\eqref{eq:cut}.  The only edge leaving the highlighted set is $a\rightarrow b$, so any sparsifier that omits it will fail to approximate the corresponding cut.}
\end{wrapfigure}

For example, consider the directed complete bipartite graph $K$ on the vertex set $A \cup B$ with all edges directed from $A$ to $B$.
For each pair of $a \in A$ and $b \in B$, the directed cut
\begin{equation}\label{eq:cut}
	E\left(
	\{a\} \cup B \setminus \{b\},
	\{b\} \cup A \setminus \{a\}\right)
\end{equation}
		 contains only the edge $a \rightarrow b$.  (See Figure~\ref{fig:bipartite_cut}.)
	Removing this edge from the graph would  change the value of this cut from $1$ to $0$, resulting in an infinite multiplicative error.
	
	Any graph that multiplicatively approximates the cuts in $K$ must have $|E(B,A)|=0$, so it must be supported on a subset of the edges of $K$, and the above then shows that such a graph must contain the edge $a\rightarrow b$ for every $a\in A$ and $b\in B$.
It thus follows that any graph that approximately preserves every cut in $K$ must contain all $|A||B|$ potential edges, so $K$  has no nontrivial sparsifier under this definition.

		It would therefore seem that any attempt at reducing the number of edges in a directed graph while preserving
		the combinatorial structure
		is doomed to fail.
		However, Eulerian graphs present a natural setting that 
		circumvents this
		because cuts in Eulerian graphs have the same amount of edge weight going in each direction, the bipartite graph counterexamples above are precluded.
		This balancedness allows one to incorporate sparsification based
		tools for flows and routings in this setting to solve
		combinatorial flow and cut problems quickly on Eulerian graphs~\cite{EneMPS16}.
				
		Most closely related to our notion of sparsification of directed
		graphs is the work by Chung on Cheeger's inequality for
		directed graphs~\cite{Chung05}.
		This result transforms the graph into an Eulerian graph $G$
		in a way identical to how we obtain Eulerian
		graphs~\cite{cohen2016faster}: by rescaling each
		edge weight by the probability of its source vertex in a stationary distribution.
 		It then relates the convergence rate of random walks on $G$ to the eigenvalues of the undirected graph obtained by removing directions on all edges.
 		Specifically if the Eulerian directed Laplaican is $\mlap$,
 		this symmetrization is
 		$\left(\mlap + \mlap^{\top}\right)/2.$

		Since the eigenvalues of the symmetrization of an Eulerian graph give information about random walks on the original graph, it might be tempting to define approximation for Eulerian graphs in terms of whether their symmetrizations approximate each other in the conventional positive semidefinite sense. For our purposes, we require (and obtain) a substantially stronger notion of approximation that preserves much of the directed structure that would be erased by symmetrizing. The reason why we need a stronger notion of approximation is that we want graphs that approximate each other under this notion to be good preconditioners of one another. In contrast, if one defines approximation according to whether the symmetrizations approximate one another, one would have to say that the length $n$ undirected cycle and the length $n$ directed cycle approximate each other, since they are both Eulerian and have the same undirected symmetrization. However, they are not good preconditioners of one another, and using one as a substitute for the other would incur very large losses in our applications. Under the notion of approximation we introduce in this paper, these graphs differ by a factor of $\Omega(n^2)$.

	\subsubsection{Laplacian System Solvers}
 		\label{sec:prev_solvers}
		Our algorithms build heavily on the literature for solving undirected Laplacians systems.    
    Since undirected Laplacians are special cases of directed Laplacians, any directed solver will yield an undirected solver when given a symmetric input.  
    It is thus helpful to consider what undirected solver we would like our method to resemble in this case.

    There are now a fairly large number of reasonably distinct algorithms for solving such systems, and we believe that several of them provide a template that could be turned into a working directed solver.  Of these, the one that our solver most closely resembles is the parallel solver by Peng and Spielman~\cite{PengS14}, which we will briefly summarize here.
    
  To simplify the notation and avoid having to keep track of degree normalizations, we only consider regular graphs when giving the intuition behind the algorithm. Suppose that we are given a $d$-regular undirected graph $G$ with Laplacian $\mlap =d\mI-\ma=d(I-\WWhat)$, where $\WWhat=\ma/d$ has $\|\WWhat\|<1$ on $\ker(\mlap)^\perp$.
		For simplicity, in the equations that follow, we restrict our attention to the space perpendicular to the kernel of $\mlap$. With this convention, the algorithm of \cite{PengS14} is then motivated by the series expansion 
\begin{equation}\label{eq:series_expansion}
	(\mI-\WWhat)^{-1}=\sum_{i\geq 0} \WWhat^i = \prod_{k\geq 0}\left( \mI + \WWhat^{2^k}\right),
\end{equation}  
  which is a matrix version of the standard scalar identity $1/(1-x)=1+x+x^2+x^3+\dots=1(1+x)(1+x^2)(1+x^4)\cdots$. If $\lambda$ is the smallest nonzero eigenvalue of $\mI-\WWhat$, then truncating this product at $k=\Theta(\log 1/\lambda)$ yields a constant relative error, which can be made arbitrarily small by further increasing $k$. 
Hence if $\lambda > 1/\mathrm{poly}(n)$, we obtain a small error by multiplying the first $O(\log n)$ terms of the product.
  This seems to suggest a good algorithm for solving a system $\mlap \vec{x}=\vec{b}$: simply compute $\mI+\WWhat^{2^k}$ for $k=0,\dots,t=O(\log n)$ and then return $\frac{1}{d}(\mI+\WWhat^{2^0})\dotsm (\mI+\WWhat^{2^t})\vec{b}$. %
  
  Unfortunately, this algorithm (implemented naively) would be too slow. As $k$ grows, $\WWhat^k$  quickly becomes dense, so computing it requires repeatedly squaring dense matrices, which takes time $O(n^\omega)$.  To deal with this, their algorithm  instead replaces these matrices with sparse approximations of them.  Peng and Spielman showed that given a graph with $n$ vertices and $m$ edges, one can compute a sparse approximation of the requisite squared matrix in nearly-linear-time.

Making this idea work requires care, since in general it is not true that the product of two matrices will be well approximated by the product of their approximations.   
For positive semidefinite matrices, however, there is a variant of this statement that  holds if one takes the products symmetrically:
if $\ma$ and $\mb$ are PSD and $\ma$ is a good approximation of $\mb$, then for any matrix $\mv$, $\mv^\top \ma \mv$
is a good approximation of $\mv^\top \mb \mv$ .
This led the authors of~\cite{PengS14} to work with a more stable symmetric version of the series described above, which allowed them to obtain their result.
 
This turns out to be a reasonably convenient template for our directed solver.  In particular, it has fewer moving parts than many of the other methods, and it does not require constructing combinatorial objects, like low-stretch spanning trees.  
Instead it directly relies on sparsification, which is our main new technical tool for directed graphs.
 
Unfortunately, we cannot directly apply the methods described above, since the symmetric product constructions that are used to control the error are no longer available for the (asymmetric) Laplacians of directed graphs.  Moreover, the strong notions of graph approximation and positive semidefinite inequalities that facilitate the analysis for the undirected solver are unavailable in the directed setting.  As such, we end up having to work with weaker error guarantees, and correct the extra error they introduce using a more involved iterative method.

\newcommand{\timebound}{\mathcal{T}}
\subsection{Our Results}
\label{sec:intro:results}

In this paper, we show that, in spite of these seemingly fundamental differences between the directed and undirected settings, we can develop directed analogues of several of the core spectral primitives that have been deployed to great effect on undirected graphs, and we use them to obtain the first almost-linear-time algorithms for many of the central problems in the analysis of non-reversible Markov chains.
The main new theoretical tools and algorithmic primitives we introduce are:
\begin{itemize}
	\item \textbf{Directed graph approximation:} 	We develop a well-behaved notion of spectral approximation for directed graphs, despite the fact that the corresponding Laplacians lack  the symmetry and positive semidefiniteness properties that the undirected version crucially relies on. Our definition specializes to the standard version based on PSD matrix inequalities when applied to undirected graphs, and it retains many of the useful features of the undirected definition.  For example, our notion of graph approximations roughly preserve the behavior of random walks, behave well under composition and change of basis, retain certain key aspects of the combinatorial structure, and provide good preconditioners for iterative methods.
\item \textbf{Directed sparsification:} We show that, under our notion of approximation, any strongly-connected directed graph can be approximated by a \emph{sparsifier} with only $\otilde(n/\epsilon^2)$ edges, and we give an algorithm to compute such a sparsifier in almost-linear time.
 To our knowledge, this is the first time that directed sparsifiers with $o(n^2)$ edges have been proven to exist, even non-algorithmically, for any computationally useful definition that retains the directed structure of a graph.

\item \textbf{Almost-linear-time solvers for directed Laplacian systems:}
Given the Laplacian $\mlap=\md - ma^\top$ of a directed graph with $n$ vertices and $m$ edges, we provide an algorithm that leverages our sparsifier construction to solve the linear system $\mlap \vec{x}=\vec{b}$ in time
\begin{equation}\label{eq:timebound}
\timebound=\tsolvelap=\tsolvelapsimple,  
 \end{equation}
where $\kappa=\max(\kappa(\mlap),\kappa(\md))$ is the maximum of the condition numbers of $\mlap$ and $\md$, improving on the best previous running time of
\tsolveprev. (See Theorem~\ref{thm:laplacian_general}.) 
To do so, we introduce a novel iterative scheme and analysis that allows us to mitigate the accumulation of errors from multiplying sparse approximations without having access to the more stable constructions and bounds available for symmetric matrices. 

\end{itemize}

In~\cite{cohen2016faster}, we provided a suite of reductions that used a solver for directed Laplacians to solve a variety other problems.
Plugging our new solver's running time into these reductions immediately gives the following  almost-linear-time algorithms:%
\footnote{We use $\timebound$ to denote anything of the form given in equation~\ref{eq:timebound}, not the time required for one call to the solver.  Some of the reductions call the solver a logarithmic number of times, so precise value of the $\log^{O(1)}(n\kappa \epsilon^{-1})$ term varies among the applications. Also, note that in this paper we give solving running times in terms of the condition number of symmetric Laplacian whereas in \cite{cohen2016faster} they are often given in terms of the condition number of the corresponding diagonal matrix.  However it is well-known that these differ only by a $O(\poly(n))$ factor and as they are in the logarithmic terms, this does not affect the running times.}
\begin{itemize}
	\item \textbf{Computing stationary distributions:}	 We can compute a vector within $\ell_2$ distance $\epsilon$ of the stationary distribution of a random walk on a strongly connected directed graph in time $\timebound$.
		\item \textbf{Computing Personalized PageRank vectors:}  We can compute a vector within $\ell_2$ distance $\epsilon$ of the Personalized PageRank vector with restart probability $\beta$ for a directed graph in time~$\timebound \log^2 (1/\beta)$. 
\item \textbf{Simulating random walks:} We can compute the escape probabilities, commute times, and hitting times for a random walk on a directed graph and estimate the mixing time of a lazy random walk up to a polynomial factor in time $\timebound$.
	\item \textbf{Estimating all-pairs commute times:} We can build a data structure of size $\otilde(n\epsilon^{-2}\log n)$ in time $\timebound/\epsilon^2$ that, when queried with any two vertices $a$ and $b$, outputs a $1\pm\epsilon$ multiplicative approximation to the expected commute time between $a$ and $b$.
\item\textbf{Solving row- and column-diagonally dominant linear systems:} We can solve linear systems that are row- or column-diagonally dominant in time $\timebound\log K$, where $K$ denotes the ratio of the largest and smallest diagonal entries.  
\end{itemize}
This gives the first almost-linear-time algorithm for 
each of these problems.  For all of them, the best previous running time for obtaining high-quality solutions was what is obtained by replacing $\timebound$ with $\tsolveprev$ and was proven in~\cite{cohen2016faster}.

\subsection{Paper Overview}
 	The rest of this paper is organized as follows.
	\begin{itemize}
	\item Section~\ref{sec:prelim} -- we cover preliminaries such as notation, facts about directed Laplacians that we use throughout the paper, and an overview of our approach.
 	\item Section~\ref{sec:sparsification} -- we introduce a notion of asymmetric approximation, and prove that we can, in nearly linear time, produce sparsifiers which are good approximations under this notion.
 	\item Section~\ref{sec:solver} -- we show how to employ the sparsification routines from the previous section in order to obtain our fast Eulerian Laplacian system solver.

 	\item Appendix~\ref{sec:entry_sparsification} -- we prove a matrix concentration result concerning entrywise sampling, which is the basic building block for the results from Section~\ref{sec:sparsification}.
	\item Appendix~\ref{sec:linear_algebra} -- we provide various general linear algebra facts used throughout the paper.
	\item Appendix~\ref{sec:decomposition} -- we show how to obtain the graph decompositions required for sparsifying arbitrary Eulerian Laplacians using a decomposition of undirected graphs into expanders.
	\item Appendix~\ref{sec:complete} -- we provide some details on the full algorithm for computing stationary distributions and solving directed Laplacians using a Eulerian Laplacian system solver.
	\item Appendix~\ref{sec:harmonic_approx} -- we relate our sparsification results to certain systems considered in \cite{cohen2016faster}. 
	\item Appendix~\ref{sec:reduction} -- we provide a general reduction that improves the dependence of our Eulerian solver  on the condition number from $\exp(\sqrt{\log \kappa})$ to $\log \kappa$.
  
	\end{itemize}

\section{Preliminaries \label{sec:prelim}}

First we give notation in Section~\ref{sec:prelim:notation} and then we give basic information about directed Laplacians in Section~\ref{sec:prelim:laplacian}. Much of this is inherited from~\cite{cohen2016faster}. With this notation in place we give an overview of our approach in Section~\ref{sec:intro:approach}.

\subsection{Notation} 
\label{sec:prelim:notation}

\textbf{Matrices:} We use bold to denote matrices and let $\mI,\mzero \in \R^{n\times n}$
denote the identity matrix and zero matrix respectively. For a matrix $\ma$ we use $\nnz(\ma)$ to denote the number of non-zero entries in $\ma$. When $\ma \in \R^{n \times n}$ we use $\supp(\ma)$ to denote the subset of $[n]$ corresponding to the indices for which at least one of the corresponding row or column in $\ma$ is non-zero.\\
\\
\textbf{Vectors:} We use the vector notation when we wish to highlight that we are representing a vector. We let $\allzeros,\allones\in\R^{n}$ denote
the all zeros and ones vectors, respectively. We use $\indic_{i}\in\R^{n}$ to denote the $i$-th basis vector,  i.e. $(\indic_{i})_{j}=0$
for $j\neq i$ and $(\indic_{i})_{i}=1$. Occasionally, when it is obvious from the context, we apply scalar
operations to vectors with the interpretation that they be
applied coordinate-wise. As with matrices, we use $\supp(\vec{x})$ to denote the indices of $\vec{x}$ with non-zero entries.\\
\\
\textbf{Positive Semidefinite Ordering:} For symmetric
matrices $\ma,\mb\in\R^{n\times n}$ we use $\ma\preceq\mb$ to denote
the condition that $x^{\top}\ma x\leq x^{\top}\mb x$, for all $x$. We define
$\succeq$, $\prec$, and $\succ$ analogously. We call a symmetric
matrix $\ma\in\R^{n\times n}$ positive semidefinite (PSD) if $\ma\succeq\mzero$. 
For vectors $x$, we let $\norm x_{\ma}\defeq\sqrt{x^{\top}\ma x}$. For asymmetric $\ma \in \R^{n \times n}$ we let $\mU_{\ma} \defeq \frac{1}{2} (\ma + \ma^\top)$ and note that $x^\top \ma x = x^\top \ma^\intercal x = x^\top \mU_\ma x$ for all $x \in \R^{n}$.
 \\
\\
\textbf{Operator Norms:} For any norm
$\norm{\cdot}$ defined on vectors in $\R^{n}$ we define the \emph{seminorm} it induces on $\R^{n\times n}$ by $\norm{\ma}=\max_{x\neq0}\frac{\norm{\ma x}}{\norm x}$
for all $\ma\in\R^{n\times n}$. When we wish to make clear that we are considering such a ratio we use the $\rightarrow$ symbol; e.g., $\norm{\ma}_{\mh \to \mh}=\max_{x\neq0}\frac{\norm{\ma x}_\mh}{\norm{x}_\mh}$, but we may also simply write $\norm{\ma}_{\mh} \defeq \norm{\ma}_{\mh \rightarrow \mh}$ in this case. For symmetric positive definite $\mh$ we have that $\norm{\ma}_{\mh \to \mh}$ can be equivalently expressed in terms of $\norm{\cdot}_2$ as $\norm{\ma}_{\mh \rightarrow \mh}=\norm{\mh^{1/2} \ma \mh^{-1/2}}_2$. Also note that $\norm{\ma}_1$ is the is the maximum $\ell_{1}$ norm of a column of $\ma$, and $\norm{\ma}_{\infty}$ is the maximum $\ell_{1}$ norm of a row of $\ma$. \\
\\
\textbf{Diagonals} 
For $x\in\R^{n}$ we let $\mdiag(x)\in\R^{n\times n}$
denote the diagonal matrix with $\mdiag(x)_{ii}=x_{i}$ and typically use $\mx\defeq\mdiag(x)$.
For $\ma\in\R^{n\times n}$ we let $\diag(\ma) \in \R^{n}$ denote the vector corresponding to the diagonal of $\ma$, i.e. $\diag(\ma)_i = \ma_{ii}$ and we let $\mdiag(\ma)$ denote the diagonal matrix having the same diagonal as $\ma$.
\\
\\
\textbf{Linear Algebra} For a matrix $\ma$, we let $\ma^\dagger$ denote the (Moore-Penrose) pseudoinverse of $\ma$. For a symmetric positive semidefinite matrix $\mb$, we let $\mb^{1/2}$ denote the square root of $\mb$, that is the unique symmetric positive semidefinite matrix such that $\mb^{1/2} \mb^{1/2} = \mb$. Furthermore, we let $\mb^{\dagger/2}$ denote the pseudoinverse of the square root of $\mb$. We use $\ker(\ma)$ to denote nullspace (kernel) of $\ma$. We use $\mathrm{span}(x_{1},x_{2},...,x_{k})$
to denote the subspace spanned by $x_{1},...,x_{k}$. For a symmetric PSD matrix $\ma$ we let $\lambdanonzero(\ma)$ denote the smallest non-zero eigenvalue of $\ma$.\\
\\
\noindent\textbf{Misc}: We let $[n]\defeq\{1,...,n\}$. For $\ma\in\R^{n\times n}$,
let $\kappa(\ma)\defeq\norm{\ma}_{2}\cdot\norm{\ma^{\dagger}}_{2}$
denote the condition number of $\ma$.
For symmetric PSD matrices $\ma$ and $\mb$ with the same kernel, let $\kappa(\ma,\mb) \defeq \kappa(\ma^{\dagger/2} \mb \ma^{\dagger/2})$ denote their relative condition number (e.g. if $\alpha \mb \preceq \ma \preceq \beta \mb$ then $\kappa(\ma,\mb) \leq \beta / \alpha$).  Note that our use of pseudoinverse rather than inverse in these definitions is non-standard but convenient. 

\subsection{Directed Laplacians}
\label{sec:prelim:laplacian}

Here we provide notation regarding directed Laplacians and review basic facts regarding these matrices that were proved in~\cite{cohen2016faster}. We begin with  some basic definitions and notation regarding Laplacians:\\
\\
\textbf{Directed Laplacian:} A matrix $\mlap\in\R^{n\times n}$ is
called a \emph{directed Laplacian} if (1) its off diagonal entries are non-positive, i.e. $\mlap_{i,j} \leq 0$ for all $i \neq j$, and (2) it satisfies $\vones^\top \mlap = \vzero$, i.e.  $\mlap_{ii}=-\sum_{j\neq i}\mlap_{ji}$ for all $i$.\\
\\
\textbf{Associated Graph:} To every directed Laplacian $\mlap \in \R^{n \times n}$ we associate a
graph $G_{\mlap}=(V,E,w)$ with vertices $V = [n]$, and edges $(i,j)$
of weight $w_{ij} =-\mlap_{ji}$, for all $i\neq j\in[n]$ with $\mlap_{ji}\neq0$.
Occasionally we write $\mlap=\md-\ma^{\top}$ to denote that we decompose
$\mlap$ into the diagonal matrix $\md$ (where $\md_{ii}=\mlap_{ii}$
is the out degree of vertex $i$ in $G_{\mlap}$) and non-negative matrix $\ma$ (which is weighted
adjacency matrix of $G_{\mlap}$, with $\ma_{ij}=w_{ij}$ if $(i,j)\in E$,
and $\ma_{ij}=0$ otherwise).\\
\\
\textbf{Eulerian Laplacian:} A matrix $\mlap$ is called an \emph{Eulerian Laplacian} if it is a directed Laplacian with $\mlap\allones = \allzeros$. Note that $\mlap$ is an \emph{Eulerian Laplacian} if and only if its associated graph is Eulerian.\\ 
\\
\textbf{(Symmetric) Laplacian}: A matrix $\mU\in\R^{n\times n}$ is
called a \emph{symmetric} or \emph{undirected Laplacian} or just a \emph{Laplacian} if it is symmetric and a directed Laplacian.  Note that the graph associated with an undirected Laplacian is undirected, i.e. for every forward edge there is a backward edge of the same weight. Given a symmetric Laplacian $\mU = \md-\ma$, we let its \textit{spectral gap} be defined as the smallest nonzero eigenvalue of $\md^{-1/2} \mU \md^{-1/2}$, i.e. $\lambda_2(\md^{-1/2} \mU \md^{-1/2}) = \min_{x \perp \ker(\md^{-1/2}\mU\md^{-1/2}), \norm{x}=1} x^{\top} \md^{-1/2} \mU \md^{-1/2} x$.\\
\\
\textbf{Running Times: } Our central object is almost always a directed Laplacian $\mlap = \md - \ma \in \R^{n \times n}$, where $m = \nnz(\ma)$, $U \defeq \max_{i,j} |\ma_{ij}|/\min_{i,j : \ma_{ij} \neq 0} |\ma_{ij}|$. We use $\tilde{O}(\cdot)$ notation to suppress factors polylogarithmic in $n$, $m$, $U$, and $\kappa$, the natural condition number of the particular problem.

 \subsection{Overview of Approach}
	\label{sec:intro:approach}

	Here we provide an overview of our approach for solving linear systems in directed Laplacians. We split it into three parts. In the first part, Section~\ref{sec:approach_reductions}, we describe how to reduce the problem to the special case of solving Eulerian Laplacians with polynomial condition number.  
	In the second part, Section~\ref{sec:approach_sparsification} we cover the efficient construction of sparsifiers.
	Finally,
	and in the third part, Section~\ref{sec:solverOverview}, we discuss how to use the sparsifier construction to build an almost-linear-time solver for polynomially well-conditioned Eulerian Laplacian systems.

\subsubsection{Reductions}
\label{sec:approach_reductions}

We begin by applying two reductions.  The first is a result from~\cite{cohen2016faster}, which states that one can solve row- and column-diagonally dominant linear systems, which include general directed Laplacian systems, by solving a small number of Laplacian systems in which the graphs are Eulerian:

\begin{thm}[Theorem 42 from~\cite{cohen2016faster}]\label{thm:dd-solver}
Let $\mm$ be an
arbitrary $n\times n$ column-diagonally-dominant or row-diagonally-dominant
matrix with diagonal $\md$. Let $b\in\im{\mm}$. Then for any $0<\epsilon\leq1$,
one can compute, with high probability and in time 
\[
O\left(\tsolve\log^{2}\left(\frac{n\cdot\kappa(\md)\cdot\kappa(\mm)}{\epsilon}\right)\right)
\]
a vector $x'$ satisfying $\|\mm x'-b\|_{2}\leq\epsilon\norm b_{2}$.

Furthermore, all the intermediate Eulerian Laplacian solves required
to produce the approximate solution involve only matrices $\mr$ for
which $\kappa(\mr+\mr^{\top}),\:\kappa(\mdiag(\mr))\leq(n\kappa(\md)\kappa(\mm)/\epsilon)^{O(1)}$.
\end{thm}

If we were to combine this directly with the algorithm from Section~\ref{sec:solver}, it would give a running time of
$\Otil \left( \left( m + n  \exp{ O\left(\sqrt{\log \kappa \cdot \log \log \kappa} \right)} \right)
	\log \left(1/\epsilon \right)\right) $
to solve linear systems in a directed Laplacian $\mlap= \md-\ma^T$,
where $\kappa$ is the condition number of the normalized Laplacian $\md^{-1/2}\mlap \md^{-1/2}$.  
While $\kappa$ is typically polynomial in $n$, it is possible for it to be exponential, so we would like our running time to depend on it logarithmically, instead of just sub-polynomially.
We show how to do this in Appendix~\ref{sec:reduction}, where we give an algorithm to solve an arbitrarily ill-conditioned Eulerian Laplacian systems by solving $O(\log (n\kappa))$ Eulerian Laplacians whose condition numbers are polynomial in $n$.  
This allows us to restrict our attention for the rest of the paper to the case where  $\kappa$ is polynomial in $n$ and, when applied to the algorithm from Section~\ref{sec:solver}, gives our final running time of
$\log^{O(1)} (n \kappa \epsilon^{-1}))$.

\subsubsection{Sparsification}
\label{sec:approach_sparsification}
Our primary new graph theoretic tool is a directed notion of spectral sparsifiers, along with efficient techniques for constructing them for an Eulerian graph and its square.
As discussed in the introduction,
there are seemingly intrinsic problems with many of the notions of 
 directed sparsification that one would propose based on analogies to the undirected case.
In particular, both the cut-based and spectral notions have seemingly fatal issues that preclude their use in directed graphs.
For the cut-based notion, as shown in Section~\ref{sub:directed_sparse_hard}, good sparsifiers provably don't exist for some graphs.  

If one instead seeks to generalize the undirected definition of spectral sparsifiers, which requires a sparsifier $H$ of a graph $G$ to obey $(1-\epsilon)\vec{x}^T \mlap_H \vec{x} \leq \vec{x}^T \mlap_G \vec{x}\leq(1+\epsilon)\vec{x}^T \mlap_H \vec{x}$, 
the problems are perhaps even more severe. 
For instance, when $G$ is directed $\mlap_G$ is no longer symmetric, so it's not clear that it makes sense to use it as a  quadratic form $\vec{x}^\top \mlap_G \vec{x}$, and doing so essentially symmetrizes it and discards the directed structure, since 
$\vec{x}^\top \mlap_G \vec{x}=\vec{x}^\top \mlap_G^\top \vec{x}=\vec{x}^\top \left(\frac{\mlap_G+\mlap_G^\top}{2} \right)\vec{x}$.  In addition, the resulting quadratic form is not typically PSD, i.e. there often exist $\vec{x}$ for which $\vec{x}^\top\mlap_G\vec{x}<0$, in which case $G$ would not approximate itself under the definition given for $\epsilon > 0$. 

One also has to deal with the fact that, unlike in the undirected case, the kernels of directed graph Laplacians are rather subtle objects: for a strongly-connected graph $G$, the kernel of $\mlap_G=\md-\ma^\top$ is given by $\md^{-1}\phi$, where $\phi$ is the stationary distribution of the random walk on $G$.  
This carries various problematic consequences, including the fact that $\mlap$ and $\mlap^\top$ typically have different kernels, and even small changes in the graph can change whether $\mlap \vec{x}=0$ for a given vector $\vec{x}$.

Our approach to this is based on the fact that many of these problems do not occur for Eulerian graphs.  In particular, if $\mlap$ is the Laplacian of an  Eulerian directed graph $G$, $\mU_\mlap=(\mlap+\mlap^\top)/2$ is the Laplacian of an undirected graph and thus positive semidefinite, and cuts in the corresponding undirected graph are the same as those in $G$.  In addition,
 the kernel of $\mlap$ is spanned by the all-ones vector and is the same as the kernel of $\mlap^\top$.  
In addition, the following was shown in~\cite{cohen2016faster}, which says that the Laplacian of any strongly connected graph can be turned into an Eulerian Laplacian by applying a diagonal scaling:
\begin{lem}[Lemma 1 from~\cite{cohen2016faster}, abridged]
	\label{lem:stationary-equivalence} Given a directed Laplacian $\mlap=\md-\ma^{\top}\in\R^{n\times n}$
	whose associated graph is strongly connected, there exists a positive
	vector $\vec{x} \in\R_{>0}^{n}$ (unique up to scaling) such that 
	$\mlap \cdot \mdiag(\vec{x})$ is an Eulerian Laplacian. Furthermore, $\ker(\mlap) = \sspan(\vec{x})$, and $\ker(\mlap^\top) = \sspan(\vones)$.
\end{lem}
Moreover, it was shown in~\cite{cohen2016faster} that one could find a high-precision approximation to this scaling efficiently given access to an Eulerian solver.

Intuitively, we define our notion of sparsification and approximation for Eulerian graphs, and we show that this notion induces a well-behaved definition for other strongly-connected graphs through the Eulerian scaling.
As we do not want to neglect the directed structure, we will think of Laplacians as linear operators, not quadratic forms, and we study their sizes through various operator norms.

For Laplacians of Eulerian graphs, we use the fact that their symmetrizations are PSD, and our definition of approximation will demand that the difference between the two operators be small relative to the corresponding quadratic form.  More precisely, we say that an Eulerian Laplacian $\mlap_H$ $\epsilon$-approximates another Eulerian Laplacian $\mlap_G$ if
$\big\|\mU_{\mlap_G}^{+/2}(\mlap_H-\mlap_G)\mU_{\mlap_G}^{+/2}\big\|_2\leq \epsilon$.
We note that this use of $\mU_{\mlap_G}$ is closely related to the $\mlap_G^\top \mU_{\mlap_G}^+ \mlap_G$ matrix that appeared in~\cite{cohen2016faster}.  The difference, however, is that we are not trying to directly use this matrix as a symmetric stand-in for our Laplacian; we are working directly with the original (asymmetric) Laplacians and are just using it to help define a matrix norm.

To construct sparsifiers of Eulerian graphs with respect to this notion, we follow a similar approach to the one originally used by Spielman and Teng for spectral sparsification, but carefully tailored to the directed setting.  The idea is to first partition our graph into well-connected components.  Because the cuts in an Eulerian graph match those in its symmetrization,  it makes sense to do this partitioning by simply partitioning the corresponding undirected graph into clusters with good expansion.  We use existing decomposition techniques to argue that one can find such a partition 
with a significant fraction of the edges contained in the clusters.  We then show a concentration result for asymmetric matrices that says that appropriately sub-sampling within these clusters preserves the relevant structure reasonably well while only keeping a small number of edges relative to the cluster size. 

In the undirected case, one would just repeat this procedure until the graph is sparse.  Where our procedure differs, however, is that we keep track of the directed structure along the way, and ``patch'' the subsampled object to keep it from diverging from what it should be.  In particular, the sampling procedure, when applied to an Eulerian graph will produce a non-Eulerian graph.  However, we add additional edges to fix this after every sampling step and use our concentration bounds to show that the patches we add are sufficiently small to not decrease the quality of our approximation.

Carefully, analyzing this procedure allows us to produce a sparsifier in nearly linear time. However, in order to use our sparsification routine to produce a solver, we also need to sparsify the Laplacian of the square of a graph. To do this, we could just explicitly form the square and then sparsify it.  However, we would like to perform this  procedure in time that is nearly-linear in the number of edges of the original graph, whereas explicitly forming the square would cause the running time to grow with the number of edges of the square, which could be substantially larger.  To prevent this, we instead show how to work with an implicit representation of the square that we can manipulate more efficiently, similar to \cite{PengS14}.

\subsubsection{Linear System Solving}
\label{sec:solverOverview}
  In Section~\ref{sec:solver}, we describe our algorithm for solving Eulerian Laplacian systems of equations.
It begins with a similar template to the Peng-Spielman solver~\cite{PengS14} described in Section~\ref{sec:prev_solvers}, but with modifications to accommodate our non-symmetric setting.  
Given a linear system in an Eulerian Laplacian $\mlap = \md - \ma^{\top}$, 
we write $\mlap= \md^{1/2}\left(\II-\WWhat\right)\md^{1/2}$, where $\WWhat=\md^{-1/2}\ma^\top \md^{-1/2}$. This reduces the problem to solving linear systems in $\LL=\II-\WWhat$ where we can show that $\|\WWhat\|_2=1$.
We then apply the expansion in Equation~\eqref{eq:series_expansion}, but with some slight modifications:
\begin{itemize}
  \item We find it convenient to build up the product expansion in Equation~\eqref{eq:series_expansion}  recursively.  We do so using the identity   
  \begin{equation}\label{eq:recursive}
  (\II - \WWhat)^+ =(\II -\WWhat^2)^+ (\II+\WWhat),
  \end{equation}
  which can be thought of as a matrix analogue of the rational function identity 
  \[
  \frac{1}{1-z}=\frac{1+z}{1-z^2}.  
  \]
  Applying this identity repeatedly gives 
  \[
     (\II - \WWhat)^+ =(\II -\WWhat^2)^+ (\II+\WWhat)
     =
     (\II -\WWhat^4)^+ (\II+\WWhat^2) (\II+\WWhat) 
     =(\II -\WWhat^8)^+ (\II+\WWhat^4) (\II+\WWhat^2) (\II+\WWhat) =\dots. 
  \] 
After $k$ applications of the identity, this yields the first $k$ terms of the product expansion in~\eqref{eq:series_expansion} times $(\II-\WWhat^{2^k})^+$, which converges to the identity as $k$ gets large if $\|\WWhat\|_2<1$.
Some advantages of this compared to the infinite product expansion are that it gives an exact expression rather than an asymptotic result, which will be more convenient to work with when analyzing the growth of errors, and that the pseudoinverses in the expression gives a correct answer when $\|\WWhat\|=1$, which decreases the extent to which we need to explicitly handle the kernel of $\mlap$ as a special case.

\item If $z\neq 1$ is a complex number with $|z|=1$, $1/(1-z)$ exists but the series $1/(1-z)=1+z+z^2+\dots$ does not converge, and our matrix expansion will exhibit similar behavior.  Graph theoretically, this case corresponds periodic behavior in the random walk, and we deal with it, as usual, by adding self-loops and working with a lazy random walk.  Algebraically, we work with a convex combination with the identity,
\[
\WWhat^{(\alpha)} = \alpha \II + (1-\alpha)\WWhat,
\]
and we note that
$\II -\WWhat^{(\alpha)}=(1-\alpha)(\II -\WWhat)$.
We then replace the identity in Equation~\eqref{eq:recursive} with the modified identity 
\begin{align}\label{eq:recursive_alpha}
  (\II-\WWhat)^+&=(1-\alpha)\left(\II-\WWhat^{(\alpha)}\right)^+
  =(1-\alpha)\left(\II-\WWhat^{(\alpha)2}\right)^+ \left(\II+\WWhat^{(\alpha)}\right),
\end{align}
which leads to better convergence behavior. This step insures that each application of the identity causes a change that is more gradual than
squaring. Moreover, our analysis takes advantage of the fact that taking a linear combination with the identity makes it easier to 
relate $\II - \AAcal_{j + 1}^{(\alpha)}$ to
$\II - \AAcal_{j}^{(\alpha)}$. While it may not be necessary to do at every step, it is used to simplify the current analysis. 
Note, that this algebraic simplification through `lazy' random walks is
also present in other works involving
squaring~\cite{ChengCLPT15,JindalK15:arxiv}.
\end{itemize}

Similarly to the approach in~\cite{PengS14}, our strategy is to repeatedly apply~\eqref{eq:recursive_alpha}, but to replace $(\WWhat^{(\alpha)})^2$ with a sparsifier in each step to allow us to decrease the computational costs. 
More precisely, we show how to efficiently construct a sequence of matrices $\WWhat_0,\WWhat_1,\ldots,\WWhat_d$  and associated matrices $\WWhat_i^{(\alpha)}=\alpha \II + (1-\alpha)\WWhat_i$ such that  
each matrix in the sequence has $\otilde(n/\epsilon^2)$ nonzero entries,
$\II-\WWhat_0$ is an $\epsilon$-approximation of $\II-\WWhat$, and
$\II-\WWhat_i$ is an $\epsilon$-approximation of $\II-(\WWhat_{i-1}^{(\alpha)})^2$  for each $i\geq 1$ (note that we set $\WWhat_0$ by sparsifying the original Laplacian).  We call this a \emph{square-sparsification chain}.  
  In Section~\ref{sec:construction}, we show how to compute all of the matrices in such a chain in time $\otilde(\nnz(\mlap)+ n\epsilon^{-2}d)$, which we note is within logarithmic factors of the total number of nonzero entries.

  The length of the chain is then dictated by the condition number  $\kappa=\kappa(\UU_{\mI - \WWhat})$, the condition number of the symmetric Laplacian associated with the input Eulerian Laplacian. 
  Note that $\kappa = O(\poly(n U))$ where $U \defeq \max_{i,j} |\ma_{ij}|/\min_{i,j : \ma_{ij} \neq 0} |\ma_{ij}|$ and may be smaller. 
  If we set $d = \Omega(\log{\kappa})$, we show that $\mI-\WWhat_d^{(\alpha)}$ well-conditioned.
  We can thus stop our recursion at this point and 
  (approximately) apply $(\mI-\WWhat_d^{(\alpha)})^+$  
    using a small number of iterations of a standard iterative method, Richardson iteration.

Expanding the recurrence in~\eqref{eq:recursive_alpha} gives
\begin{equation}
\left(\mI - \WWhat_{i}\right)^\dagger
\approx \left( 1 - \alpha \right)^{j - i} 
\left(\mI-\WWhat_{j}^{\left(\alpha\right)}\right)^\dagger
\left(\mI + \WWhat_{j-1}^{\left(\alpha\right)}\right)
\left(\mI + \WWhat_{j-2}^{\left(\alpha\right)}\right) \dotsm \left(\mI + \WWhat_{i}^{\left(\alpha\right)}\right).
\label{eqn:key}
\end{equation}
If we have already computed the matrices in the chain, we can  apply the right-hand side to a vector $\vec{b}$ by performing $(j-i)$ matrix-vector multiplications and solving a linear system in $\II-\WWhat_j^{(\alpha)}$.  
It is useful to think of this as an approximate reduction from applying $(\II-\WWhat_i)^+$ to applying $(\II-\WWhat_j^{(\alpha)})^+
$.
The matrices in~\eqref{eqn:key} have at most $\otilde(n/\epsilon^2)$ nonzero entries, so the total time for the matrix vector multiplications is then at most $\otilde\left((j-i)n\epsilon^{-2}\right)$.

Because of the errors introduced by the sparsification steps, the right-hand side of~\eqref{eqn:key} is only an approximation of $(\II-\WWhat_i)^+$, so applying it directly to $\vec{b}$ only yields a (typically somewhat crude) approximation to solution to $(\II-\WWhat_i)\vec{x}=\vec{b}$.
To obtain a better solution, we instead use it as a preconditioner inside an iterative method for the linear system.  
This allows us to obtain an arbitrarily good solution to the system, and the quality of the approximation in~\eqref{eqn:key} then determines the number of iterations required.
 
This suggests that we quantify the error in our approximations using a notion that directly bounds the convergence rate of such a preconditioned iterative method. 
We do so with the notion of an $\epsilon$-approximate pseudoinverse (defined with respect to some PSD matrix $\mU$),
  which we introduce in Section~\ref{sec:richardson}.
  Roughly speaking, solving a linear system  with an appropriate iterative method using such a matrix as a preconditioner will 
  guarantee the $\mU$-norm of the error to decrease by a factor of $\epsilon$ in each iteration.
We note that this is only useful for $\epsilon<1$. For technical reasons, we measure the quality of approximate pseudoinverses with respect to different $\mU$ matrices
at different stages of the algorithm and translate between them. For simplicity, we just refer to an ``$\epsilon$-approximate pseudoinverse'' in this overview, but in our algorithm we set the value of $\epsilon$ in our sparsification routines and apply iterative methods, again Richardson iterations, to appropriately pay for the costs of translating between norms.

To analyze the errors introduced by sparsification, we therefore need to:
\begin{enumerate}
  \item Relate our notion of graph approximation to  approximate pseudoinverses, and
  \item Bound the rate at which the quality of the approximate pseudoinverse we produce decreases as we increase the number of terms in~\eqref{eqn:key}.  We use~\eqref{eqn:key} recursively, so it is also be useful to bound how this is affected if we use an approximate pseudoinverse instead of the exact operator $\big(\II-\WWhat_j^{(\alpha)}\big)^+$.
\end{enumerate}

For the former, we show in Theorem~\ref{thm:spectral-to-approxInv} that our notion of an $\epsilon$-sparsifier leads to an $O(\epsilon)$-approximate pseudoinverse.  
For the latter, we show in Lemma~\ref{lem:chainproperty} that using a square-sparsifier chain of length $d$ with some given $\epsilon$, and using an $\epsilon'$-approximate pseudoinverse of $\big(\II-\WWhat_j^{(\alpha)}\big)$ in place of $\big(\II-\WWhat_j^{(\alpha)}\big)^+$, produces an $(\epsilon+\epsilon')\cdot 2^{O(d)}$-approximate pseudoinverse for $\II-\WWhat_j^{(\alpha)}$.

The exponential dependence of the error on length of the chain is a key difference between our analysis and the undirected case, and it is what prevents us from having a simpler and more efficient algorithm.
If the dependence on the chain length were polynomial,  applying~\eqref{eqn:key} with $i=0$ and $j=d$ would provide an $\epsilon\cdot \mathrm{polylog}(n)$-approximate pseudoinverse.  We could thus set $\epsilon=1/\mathrm{polylog}(\kappa)$ in our sparsifier chain and get an $O(1)$-approximate pseudoinverse in $\otilde(n)$ time.  
An iterative method 
could then call this $\log(1/\delta)$ times to obtain a solution with error $\delta$.
However, because of the exponential dependence on the chain length, we would only get an $\epsilon\cdot \poly(\kappa)$-approximate pseudoinverse.  We would thus need to set  $\epsilon =1/\poly(\kappa)$ to get a value less than $1$, which would lead to ``sparsifiers'' with $\widetilde{\Omega}(n\cdot\poly(\kappa))$ edges.
In the typical case where $\kappa=\poly(n)$, simply writing these down would exceed the desired almost-linear time bound.

\newcommand{\rectime}[2]{\mathcal{T}_{#1,#2}}
\newcommand{\esparse}{\epsilon_{\mathrm{spar}}}
\newcommand{\elow}{\epsilon_{\mathrm{lo}}}
\newcommand{\ehigh}{\epsilon_{\mathrm{hi}}}
\newcommand{\ehighval}{(\esparse+\elow) 2^{O(\Delta)}}
\newcommand{\richiters}{O(\log \elow/\log \ehigh) }
\newcommand{\richitersdisp}{O\left(\frac{\log \elow}{\log \ehigh} \right)}
To prevent this, we do not wait until the end to apply an iterative method to reduce the error.  Instead, we break our sparsification and squaring steps into $\ceil{d/\Delta}$ blocks of size $\Delta \ll d$, each of which we 
will wrap in
several steps  of 
Richardson iteration (which we review in Section~\ref{sec:richardson}), in order to keep the error under control.  

  Our algorithm first computes (once, not recursively) a square-sparsifier chain of length $d=O(\log \kappa)$ in which the sparsifiers are $\esparse$-approximations.  
It then recursively combines two  types of steps that are suggested by the discussion above:
\begin{itemize}  
  \item \textbf{High error $\big(\mI-\WWhat_i^{(\alpha)}\big)^+$ from low error $\big(\mI-\WWhat_{i+\Delta}^{(\alpha)}\big)^+$}: Given  a routine to apply an $\elow$-approximate pseudoinverse of $\mI-\WWhat_{i+\Delta}^{(\alpha)}$ in time $\rectime{i+\Delta}{\elow}$, we can use the expansion in~\eqref{eqn:key} to apply an $\ehigh$-approximate pseudoinverse of  $\mI-\WWhat_i^{(\alpha)}$ in time $\rectime{i}{\ehigh}=\rectime{i+\Delta}{\elow}+\otilde(\Delta n\esparse^{-2})$, where $\ehigh=\ehighval$.
  \item \textbf{Low error $\big(\II-\WWhat_i^{(\alpha)}\big)^+$ from high error $\big(\II- \WWhat_{i}^{(\alpha)}\big)^+$}: By running Richardson iteration for $\richiters$ steps, we can turn an $\ehigh$-approximate pseudoinverse of $\mI- \WWhat_i^{(\alpha)}$ into a $\elow$-approximate pseudoinverse.  This applies the former once in each iteration, so it takes time 
\begin{equation}\label{eq:recurrence}
    \rectime{i}{\elow}=\richitersdisp \rectime{i}{\ehigh}=\richitersdisp\left(\rectime{i+\Delta}{\elow}+\otilde\left(\Delta n\esparse^{-2}\right)\right).
\end{equation}
\end{itemize}

If we set $\ehigh$ to be a constant (say, $1/10$), we  
get $\esparse + \elow=2^{-\Omega(\Delta)}$, so we set $\esparse =\elow=2^{-\Theta(\Delta)}$, and~\eqref{eq:recurrence} simplifies to
 \[\rectime{i}{\elow}=
O(\Delta) \left(\rectime{i+\Delta}{\elow}+\otilde\big(\Delta n2^{\Theta(\Delta)}\big)\right)
=O\left(\Delta\right) \rectime{i+\Delta}{\elow}+\otilde\big(n 2^{\Theta(\Delta)}\big).
\]
For the base case of our recurrence, $\II-\WWhat_d^{(\alpha)}$ is well-conditioned, so we can approximately apply its pseudoinverse using a standard iterative method in time $\rectime{d}{\elow}=\otilde(\nnz(\WWhat_d^{(\alpha)})\log \elow^{-1})=\otilde(n\esparse^{-2}\log \elow^{-1})
=\otilde\big(n 2^{\Theta(\Delta)}\big)$.
This can be folded into the additive $\otilde\big(n 2^{\Theta(\Delta)}\big)$ term in the recurrence, so it does not significantly affect the time bound.

To estimate the solution to the recurrence, we note that depth of the recursion is $\ceil{ d/\Delta}$, and at each stage we multiply by $O(\Delta)$.  We can think of this as producing a recursion tree with $O(\Delta)^{\ceil{ d/\Delta}}$ nodes, and we add
 $\otilde\big(n 2^{\Theta(\Delta)}\big)$ 
 at each, so we get that
\[\rectime{0}{\elow}=
O(\Delta)^{\ceil{ d/\Delta}} \otilde\big(n 2^{O(\Delta)}\big)
= 
 nO(\Delta)^{O(d/\Delta)} 2^{O(\Delta)}
 =
 n 2^{O\left(\Delta+ \frac{d\log\Delta}{\Delta}\right)}.
\] 
Setting $\Delta=\sqrt{d \log d}=\sqrt{\log\kappa \log \log \kappa}$ approximately balances the two terms in the exponent.
Plugging this in and adding the $\otilde(m)$ for the overhead from the non-recursive parts of the algorithm gives our running time
bound of   
$\otilde(m)+n 2^{O(\sqrt{\log\kappa \log \log \kappa})}$.

\section{Sparsification of Directed Laplacians \label{sec:sparsification}}

 In this section, we define what it means for one strongly connected directed graph to approximate another, and we use this to define directed sparsifiers.  
 We show that such sparsifiers exist for any directed graph and give efficient algorithms to construct them.
For our almost linear time directed Laplacian system solver (Section~\ref{sec:solver}), we only need to be able to sparsify Eulerian graphs. However, we have included the more general case of non-Eulerian graphs because we believe it is of independent interest and may be useful in other settings.

As discussed in Section~\ref{sec:approach_sparsification}, we define our notion of graph approximation by first giving a notion of approximation for matrices and then applying it to (possibly rescaled versions of) directed graph Laplacians.  This notion will be qualitatively better-behaved for Laplacians of Eulerian graphs than for general directed Laplacians.  As such, our definition will make use of the existence of a scaling that makes any strongly connected graph Eulerian.

Our notion of approximation for asymmetric matrices is defined as follows:

\begin{definition}[Asymmetric Matrix Approximation]\label{def:epsaprox}
A (possibly asymmetric) matrix $\widetilde{\ma}$ is said to be an  \emph{$\epsilon$-approximation of $\ma$} if:
\begin{enumerate}
\item $\mU_{\ma}$ is a symmetric PSD matrix, with $\ker(\mU_{\ma})\subseteq\ker(\widetilde{\ma}-\ma)\cap\ker((\widetilde{\ma}-\ma)^{\top})$,
and
\item
$
\norm{\mU^{\dagger/2}_{\ma}(\widetilde{\ma}-\ma)\mU^{\dagger/2}_{\ma}}_{2}\leq\epsilon.
$
\end{enumerate}
When these properties hold for some constant $\epsilon \in (0, 1)$ 
we simply say $\widetilde{\ma}$ approximates $\ma$.
\end{definition}

In Section~\ref{sub:sparse_facts} we provide an equivalent definition and several facts which justify this choice of matrix approximation. In particular, we prove the following facts regarding Definition~\ref{def:epsaprox} 
\begin{itemize}
\item it generalizes spectral approximation (ie., small relative
condition number) of symmetric matrices,
and behaves predictably under perturbations;
\item it implies the symmetrizations $\mU_\ma$ and $\mU_{\widetilde{\ma}}$ of $\ma$ and $\widetilde{\ma}$, spectrally approximate each other;
\item its behavior under composition is natural.
\end{itemize}
Furthermore, in Appendix~\ref{sec:harmonic_approx} we show that our notion of approximation also yields approximations of the symmetric systems solved in \cite{cohen2016faster}, known as harmonic symmetrizations.

We extend this notion of approximation from asymmetric matrices to directed graphs as follows:

\begin{definition}[Directed Graph Approximation] Let $\mlap,\widetilde{\mlap}\in \mathbb{R}^{n\times n}$ be the Laplacians of strongly-connected directed graphs $G$ and $\widetilde{G}$ respectively, and let
 $\mx=\mathrm{diag}({\vec{x}})$ and $\widetilde{\mx}=\mathrm{diag}(\vec{\tilde{x}})$ be the diagonal matrices  for which $\mlap \mx$ and $\widetilde{\mlap} \widetilde{\mx}$ are Eulerian Laplacians that are guaranteed to exist by
	 Lemma~\ref{lem:stationary-equivalence}, normalized to have
	  $\mathrm{Tr} (\mx)=\mathrm{Tr}(\widetilde{\mx})=n$.
		
		We say that \emph{$\widetilde{G}$ is an $\epsilon$-approximation of $G$} if:
	\begin{enumerate}
		\item $(1-\epsilon)\mx \leq \widetilde{\mx} \leq (1+\epsilon) \mx$,
     and
		\item  $\widetilde{\mlap}\widetilde{\mx}$ is an $\epsilon$-approximation of $\mlap\mx$.
		 \end{enumerate}
		 If $\widetilde{\mx}=\mx$, we say that \emph{$\widetilde{G}$ is a strict $\epsilon$-approximation of $G$}.
\end{definition}

In words, we say that a graph approximates another graph if their Eulerian scalings are within small multiplicative factors of one another and the resulting Eulerian graphs obey our definition of asymmetric matrix approximation.   We call the approximation ``strict'' if their Eulerian scalings are not just within small multiplicative factors of one another but are actually identical.

Our main use of this notion is to define \emph{sparsifiers}, which are approximations that have a small number of nonzero entries.
\begin{definition}(Graph Sparsifier)
	 Let $\mlap,\widetilde{\mlap}\in \mathbb{R}^{n\times n}$ be the Laplacians of strongly-connected directed graphs $G$ and $\widetilde{G},$ respectively.
	We say that \emph{$\widetilde{G}$ is a (strict) $\epsilon$-sparsifier of $G$} if:
	\begin{enumerate}
		\item $\widetilde{G}$ is a (strict) $\epsilon$-approximation of $G$, and
		\item $\nnz(\widetilde{\mlap}) \leq \tilde{O}(n\epsilon^{-2})$, where $n$ is the number of vertices in $G$.
	\end{enumerate}
\end{definition}

We note that, if we show that strict sparsifiers exist for Eulerian graphs, this will imply that they exist for general strongly connected graphs as well.  One can simply apply the graph's Eulerian scaling, find a sparsifier for the resulting Eulerian graph, and then ``unscale'' the graph by applying the inverse of the Eulerian scaling.  
The results of this paper allow us to compute an Eulerian scaling for any graph in almost-linear time. Thus, the almost linear time sparsification procedure we will give for Eulerian graphs will imply an almost linear time sparsification procedure for all graphs.%
\footnote{One has to exercise some care with the numerics here, since we will only be able to compute a finite precision estimate of the Eulerian scaling of a graph. So, if we apply this scaling and then want to sparsify, the graph we want to sparsify will not be perfectly Eulerian.  However, it is straightforward to show that---as long as the approximate Eulerian scaling being used is fairly precise---one can ``patch'' the rescaled graph to become Eulerian while only incurring a very small loss in the approximation quality.}
As such, we will focus on the Eulerian case.  In this case, graph approximation will be the same as matrix sparsification, and it will often be convenient to refer directly to the Laplacian, instead of to the graph.  Moreover, we will seek to exactly preserve the fact that the graph is Eulerian, so we will exclusively consider strict approximations.  We thus define:\footnote{We could ask for sparsifiers under other notions of approximation, such as the weaker conditions required to obtain a preconditioner (see Section~\ref{sec:richardson}); as our algorithms always give this notion we use this terminology.} 
\begin{definition}[Eulerian Sparsifier]
\label{def:strongApprox}
$\widetilde{\mlap} \in \R^{n \times n}$ is an \emph{$\epsilon$-sparsifier of an Eulerian Laplacian $\mlap$} if 
\begin{enumerate}
\item $\widetilde{\mlap}$ is a strict $\epsilon$-approximation of $\mlap$
\item $\nnz(\widetilde{\mlap}) \leq \tilde{O}(n\epsilon^{-2})$.
\end{enumerate}
\end{definition}

In the remainder of the section, we show how to produce such sparsifiers of Eulerian Laplacians. I.E., given an Eulerian Laplacian, we will obtain an Eulerian Laplacian that approximates the original and has a small number of nonzero entries.

Moreover, we show how to construct these sparsifiers in nearly linear time. Specifically, we give an algorithm that produces an $\epsilon$-sparsifier
of an Eulerian Laplacian $\mlap$ with high probability in $\otilde(\nnz(\mlap)/\epsilon^2)$ time. Even without any additional work, this result alone immediately implies some improvement in the runtime for solving arbitrary directed Laplacian systems. Specifically, one can write down the harmonic symmetrization of the original matrix, the harmonic symmetrization of the sparsify, and solve the original harmonic symmeterization preconditioned by the harmonic symmetrization of the sparsifier. We prove in Appendix~\ref{sec:harmonic_approx} that the these Harmonic symmetrizations have small relative condition number, so the runtime of this solver will be dominated by the time it takes to solve systems in the sparsified matrix plust the time to apply the unsparsified matrix to a vector. Using the solver in \cite{cohen2016faster}, this comes out to an $\otilde(m+n^{7/4})$ time algorithm for solving directed Laplacian systems. 

In order to construct a better, almost linear time solver, we'll also need to be able to sparsify a normalized version of the Laplacian of the square of graph. Specifically, we also show how to sparsify any matrix of the form $\md-\ma^{\top}\md^{-1}\ma^{\top}$, where we are given the adjacency matrix $\ma$ of some Eulerian graph $G$ and the degrees $\md$ of $G$. Note that if $G$ is regular and has all (weighted) degrees equal to one, this formula simplifies to $I-(A^\intercal)^2$. Thus, it corresponds to the Laplacian of the square of the graph in this special case, and to a normalized version of it in general. Our algorithm for sparsifying matrices of this form takes outputs a strict $\epsilon$-sparsifier in $\otilde(\nnz(\mlap)\epsilon^{-2})$ time.
Combining these results, and using our properties regarding asymmetric
approximations, we show that we can use this $\epsilon$-approximation to efficiently obtain an
$\epsilon$-sparsifier of $\md-\ma^{\top}\md^{-1}\ma^{\top}$. Applying a closely related routine recursively we obtain a faster, almost linear time algorithm for solving Eulerian Laplacian systems in Section~\ref{sec:solver}.

The remainder of this section is structured as follows.
First, in Section~\ref{sub:sparse_facts}, we provide various facts
regarding our notion of asymmetric approximation.
Then, in Section~\ref{sub:sample_directed}, we provide one of
our main technical tools: an algorithm for crudely sparsifying
an arbitrary (not necessarily Eulerian) directed Laplacian by randomly  sampling its adjacency matrix. On its own, this algorithm achieves relatively weak guarantees.
In Section~\ref{sub:sparse_eulerian} we then combine this tool with
known decomposition results for undirected graphs to
sparsify any Eulerian Laplacian. In Section~\ref{sub:sparse_square}, we build on this to sparsify the normalized square of any Eulerian Laplacian.

\subsection{Approximation Facts\label{sub:sparse_facts}}

Here we provide various basic facts regarding the notion
of approximation for asymmetric matrices given in
Definition~\ref{def:strongApprox}.
These facts both motivate and justify our choice of definition
and are used extensively throughout this section.

First we provide the following lemma, giving  alternative
definitions of approximation in terms of a quantity
reminiscent of Rayleigh quotients.

\begin{lem}[Equivalent Approximation Definitions]
\label{lem:asym_strong_equiv}
Let $\ma \in \R^{n \times n}$ be such that $\mU_{\ma}$ PSD. A matrix $\widetilde{\ma} \in \R^{n \times n}$ is an $\epsilon$-approximation
of $\ma$ if and only if
\[
\max_{\vec{x},\vec{y}\neq0}
	\frac{\vec{x}^{\top}(\widetilde{\ma}-\ma)\vec{y}}
		{\sqrt{\vec{x}^{\top}\mU_{\ma} \vec{x}\cdot \vec{y}^{\top}\mU_{\ma} 
\vec{y}}} \leq \epsilon
\enspace\text{ or equivalently }\enspace
\max_{\vec{x}, \vec{y} \neq 0}
	\frac{\vec{x}^{\top}(\widetilde{\ma}-\ma)\vec{y}}
		{\vec{x}^{\top}\mU_{\ma} \vec{x}+\vec{y}^{\top}\mU_{\ma} \vec{y}}\leq\frac{\epsilon}
{2}\,{,}
\]
under the convention that $0/0=0$.
\end{lem}
\begin{proof}
This lemma follows from a more general result,
Lemma~\ref{lem:spectral_equivalence}, which we prove in Appendix~
\ref{sec:decomposition},
and by noting that if $\ker(\mU_{\ma})$ is not a subset of both $\widetilde{\ma}-\ma$ and its transpose then the maximization problems are
infinite in value.
\end{proof}
This lemma will allow us to show that our notion of $\epsilon$-approximation does coincide with the standard notion in the case of
symmetric matrices, and is therefore a stronger notion.
More generally, we prove that $\epsilon$-approximation of
asymmetric matrices implies that their symmetrizations are
$\epsilon$-approximations in the traditional spectral (small relative condition number) sense. 
\begin{lem}
\label{lem:asym_strong_implies_undir}
Suppose $\widetilde{\ma}$ is an $\epsilon$-approximation of $\ma$.
Then
\[
(1-\epsilon)\mU_{\ma}\preceq
\mU_{\widetilde{\ma}}\preceq(1+\epsilon)\mU_{\ma} ~. 
\] 
\end{lem}
\begin{proof}
Suppose $\widetilde{\ma}$ is an $\epsilon$-approximation of $\ma$, and let $\vec{x}\in\R^{n}$ with $\vec{x} \neq \allzeros$ be arbitrary.
Applying Lemma~\ref{lem:asym_strong_equiv} twice with
$\vec{y} =\pm \vec{x}$ we have 
\[
\frac{\abs{\vec{x}^{\top}(\widetilde{\ma}-\ma)\vec{x}}}
	{\vec{x}^{\top}\mU_{\ma}\vec{x}}
\leq\norm{\mU_{\ma}^{\dagger/2}
(\widetilde{\ma}-\ma)\mU_{\ma}^{\dagger/2}}_{2}\leq\epsilon~. 
\]

The desired result follows from the fact that $\vec{z}^{\top}\ma \vec{z} = \vec{z}^{\top}\mU_{\ma}\vec{z}$
and $\vec{z}^{\top}\widetilde{\ma}\vec{z} = \vec{z}^{\top} \mU_{\widetilde{\ma}} \vec{z}$ for all $z$.
\end{proof}

Next we use Lemma~\ref{lem:asym_strong_equiv} to show that just as in the symmetric case, asymmetric
approximation is preserved when taking symmetric products.
\begin{lem}
\label{cor:strong_basis_change}
If $\widetilde{\ma} \in \R^{n \times n}$ is an $\epsilon$-approximation of $\ma \in \R^{n \times n}$ and $\mm \in \R^{n \times n}$ satisfies $\ker(\mm^\intercal) \subseteq\ker(\mU_{\AA})$ then 
$\mm^{\top}\widetilde{\ma}\mm$
is an $\epsilon$-approximation of 
$\mm^{\top}\ma\mm$.
\end{lem}
\begin{proof}
Define $\mb \defeq \mm^{\top}\ma\mm$ and $\widetilde{\mb} \defeq \mm^{\top}\widetilde{\ma}\mm$. We first wish to show $\ker(\mU_\mb)
\subseteq\ker(\widetilde{\mb} - \mb) \cap \ker((\widetilde{\mb} - \mb)^{\top})$. Suppose we have any $x \in \ker(\mU_\mb)$. Then,
\[
\mm x \in \ker(\mm^\intercal \mU_\ma) = \ker(\mU_\ma) \subseteq \ker(\widetilde{\ma} - \ma) \cap \ker((\widetilde{\ma} - \ma)^{\top}) \subseteq \ker(\mm^\intercal(\widetilde{\ma} - \ma)) \cap \ker(\mm^\intercal(\widetilde{\ma} - \ma)^{\top})
\]
which implies the portion of the definition of approximation concerning kernels.

Then we have, using the convention that $0/0=0$ and by applying Lemma~\ref{lem:spectral_equivalence} and Lemma~\ref{lem:asym_strong_equiv}, that
\begin{align*}
\norm{\mU_\mb^{\dagger/2}(\widetilde{\mb} - \mb) \mU_\mb^{\dagger/2}}_{2}
& = \max_{\vec{x},\vec{y} \neq 0}
	\frac{\vec{x}^{\top}(\widetilde{\mb} - \mb)\vec{y}}
		{\sqrt{\vec{x}^{\top}\mU_\mb\vec{x} \cdot \vec{y}^{\top} \mU_\mb \vec{y}}}  v
= \max_{\vec{x}, \vec{y} \neq 0}
	\frac{\vec{x}^{\top} \mm^{\top}(\widetilde{\ma}-\ma)\mm \vec{y}}
	{\sqrt{\vec{x}^{\top}\mm^{\top}\mU_\ma \mm \vec{x}
		\cdot \vec{y}^{\top}\mm^{\top}\mU_\ma \mm \vec{y}}}\\
&\leq \max_{\vec{x}, \vec{y} \neq 0}
	\frac{\vec{x}^{\top} (\widetilde{\ma}-\ma)\vec{y}}
	{\sqrt{\vec{x}^{\top} \mU_\ma \vec{x}
		\cdot \vec{y}^{\top} \mU_\ma \vec{y}}}
=
\norm{\mU_{\ma}^{+/2} (\widetilde{\ma} - \ma) \mU_{\ma}^{+/2}}_2 
\leq
\epsilon.
\end{align*}
\end{proof}

Finally, we provide a transitivity result for strong-approximation. 

\begin{lem}[Approximation Transitivity]
\label{lem:strong_transitivity}
If $\mc$ is an $\epsilon$-approximation of $\mb$ and $\mb$ is an $\epsilon$-approximation of $\ma$ then $\mc$ is an $\epsilon(2+\epsilon)$-approximation of $\ma$.
\end{lem}

\begin{proof}
Note that by triangle inequality 
\[
\norm{\mU_{\ma}^{\dagger/2}(\mc-\ma)\mU_{\ma}^{\dagger/2}}_{2}\leq\norm{\mU_{\ma}^{\dagger/2}(\mc-\mb)\mU_{\ma}^{\dagger/2}}_{2}+\norm{\mU_{\ma}^{\dagger/2}(\mb-\ma)\mU_{\ma}^{\dagger/2}}_{2}\,.
\]
Now, $\mU_{\mb} \preceq (1+\epsilon)\mU_{\ma}$ by Lemma~\ref{lem:asym_strong_implies_undir} and therefore $\mU_{\ma}^{\dagger} \preceq (1 + \epsilon) \mU_{\mb}^{\dagger}$. Applying Lemma~\ref{lem:simple_spec_inequalities} yields
\[
\norm{\mU_{\ma}^{\dagger/2}(\mc-\mb)\mU_{\ma}^{\dagger/2}}_{2}
\leq(1+\epsilon)\cdot\norm{\mU_{\mb}^{\dagger/2}(\mc-\mb)\mU_{\mb}^{\dagger/2}}_{2}\,.%
\]
The result follows as $\norm{\mU_{\ma}^{\dagger/2}(\mb-\ma)\mU_{\ma}^{\dagger/2}}_{2}\leq\epsilon$
and $\norm{\mU_{\mb}^{\dagger/2}(\mc-\mb)\mU_{\mb}^{\dagger/2}}_{2}\leq\epsilon$ by assumption.
\end{proof}

\subsection{Sampling a Directed Laplacian\label{sub:sample_directed}}

Here we show how to compute a crude, sparse approximation
to an arbitrary directed Laplacian by randomly sampling its entries. We provide both a general bound on the effect of such sampling for a directed Laplacian as well as a more specific result in the case where the directed Laplacian can be related to the symmetric Laplacian of an expander. The latter result (Lemma~\ref{lem:subgraph_sparse}) and the terminology relevant to it, is all we use from this subsection in order to obtain and analyze our sparsification algorithms.

The main tool for our analysis is Theorem~\ref{thm:concentration_entry}, a general bound on concentration when sampling the entries of an asymmetric matrix. Its proof follows directly from standard matrix concentration inequalities, so we defer its proof to Appendix~\ref{sec:entry_sparsification}.
\begin{thm}
\label{thm:concentration_entry} Let $\ma\in\R^{d_{1}\times d_{2}}_{\geq 0}$ be a matrix where no row or column is all zeros. Let $\epsilon,p\in(0,1)$, $s=d_1+d_2$, $\vec{r} = \ma \allones$, $\vec{c} = \ma^\top \allones$, and $\dist$ be a distribution over $\R^{d_1\times d_2}$
such that $\mx \sim \dist$ takes value
\[
\mx = \left(\frac{\ma_{ij}}{p_{ij}} \right)\indic_{i}\indic_{j}^{\top}
\text{ with probability }
p_{ij} = \frac{\ma_{ij}}{s} \left[\frac{1}{\vr_i} +\frac{1}{\vc_j} \right]
\text{ for all }
\ma_{ij} \neq 0~.
\]
If  $\ma_{1},..,\ma_{k}$  are sampled independently from $\dist$ for $k \geq 128 \cdot \frac{s}{\epsilon^2} \log \frac{s}{p}$, $\mr=\mdiag(\vr)$, and $\mc=\mdiag(\vc)$ then the average  $\empircalA\defeq\frac{1}{k}\sum_{i\in[k]}\ma_{i}$, satisfies
\begin{align*}
\Pr&\left[\norm{\mr^{-1/2}\left(\empircalA-\ma\right)\mc^{-1/2}}_{2}\geq\epsilon\right]\leq p\,{,} \\
\Pr&\left[\norm{\mr^{-1}(\empircalA-\ma)\allones}_{\infty}\geq\epsilon\right]\leq p\,{,\text{ and}} \\
\Pr&\left[\norm{\mc^{-1}(\empircalA-\ma)^{\top}\allones}_{\infty}\geq\epsilon\right]\leq p\,{.}
\end{align*}
\end{thm}

Theorem~\ref{thm:concentration_entry} shows that by sampling the entries of a rectangular matrix we can compute a new matrix such that spectral norm of the differences is bounded and the row sums and column sums are approximately preserved. In the next lemma we should how we can use this procedure to obtain a matrix with the same bound on the difference in spectral norm, but the row and column norms preserved \emph{exactly}. In short, we show how to add a matrix to the result of Theorem~\ref{thm:concentration_entry} preserving its properties while fixing the row norms, we call this procedure \emph{patching}.

\begin{lem}[Sparsifying Non-negative Matrices]
\label{lem:sparsifying_adjacency_matrix}
Let $\ma \in \R^{n \times n}_{\geq 0}$ be a matrix with non-negative entries and having no all zeros row or column. Let $\epsilon,p \in (0,1)$. In $O(\nnz(\ma)+n\epsilon^{-2}\log(n/p))$ time we can compute
a matrix $\widetilde{\ma} \in \R^{n \times n}_{\geq 0}$ with non-negative entries such that for $\mr \defeq \mdiag(\ma \vones)$ and  $\mc \defeq \mdiag(\ma^\top \vones)$ we have
\begin{enumerate}
\item $\nnz(\widetilde{\ma}) = O(n\epsilon^{-2}\log(n/p))$,\footnote{Note that it is possible to remove the dependence in $p$ from the sparsifier simply by increasing the running time by a constant factor and making it an expected running time. This can be achieved by using the power method to approximately compute the value of $\norm{\mr^{-1/2}(\ma-\widetilde{\ma})\mc^{-1/2}}_{2}$ and resampling when this is large.}
\item the row and column sums of $\ma$ and $\widetilde{\ma}$ are the same,
i.e. $\ma \vones = \widetilde{\ma} \vones$ and $\ma^\top \vones = \widetilde{\ma}^\top \vones$.
\item for $i \in [n]$, if $\ma_{ii} = 0$ then $\widetilde{\ma}_{ii} = 0$, and
\item with probability at least $1-p$,
$\norm{\mr^{-1/2}(\ma-\widetilde{\ma})\mc^{-1/2}}_{2}\leq\epsilon$.
\end{enumerate}
\end{lem}
\begin{proof}
We prove this using Theorem~\ref{thm:concentration_entry}. Let $\epsilon'=\epsilon/4$. By sampling
as in Theorem~\ref{thm:concentration_entry} we can compute
$\widehat{\ma}\in\R^{n\times n}$ such that $\nnz(\widehat{\ma})=O(n\epsilon^{-2}\log(n/p))$,
$\norm{\mr^{-1/2}(\ma-\widehat{\ma})\mc^{-1/2}}_{2}\leq\epsilon'$,
\[
\norm{\mr^{-1}(\widehat{\ma}-\ma)\allones}_{\infty}\leq\epsilon'
\text{ and }
\norm{\mc^{-1}(\widehat{\ma}-\ma)^{\top}\allones}_{\infty}\leq\epsilon'\,.
\]
Note that the latter two conditions imply that entrywise the row sums and column sums of $\ma$ are approximately the same as $\widehat{\ma}$. Formally, it implies that entrywise the following inequalities hold
\begin{equation}
\label{eq:nonnegat_approx_eq1}
(1-\epsilon') \ma \vones \leq \widehat{\ma} \vones \leq(1+\epsilon') \ma \vones
\text{ and }
(1-\epsilon') \ma^\top \vones \leq \widehat{\ma}^\top \vones \leq(1+\epsilon') \ma^\top \vones\,.
\end{equation}
Therefore $(1+\epsilon')^{-1} \cdot \widehat{\ma}$ has row and column sums that are less than or equal to those of $\ma$. 

Next, we compute a matrix to make the row and column sums the same as in $\ma$. Formally, we let $\mE\in\R_{\geq0}^{n\times n}$ be a matrix with $\nnz(\mE)=O(n)$ such that
\[
((1+\epsilon')^{-1}\widehat{\ma}+\mE) \vones
= \ma \vones
\text{, and }
((1+\epsilon')^{-1}\widehat{\ma}+\mE)^\top = \ma^\top \vones
\]
and the $\ell_1$ norm of the entries of $\mE$ is at most $n\epsilon$.

We can compute
such a matrix $\mE$ in $O(\nnz(\widehat{\ma}))$ time 
by greedily adding in values to $\widehat{\ma}$ to make one of the row
or column sums as large as that of $\ma$, while maintaining the invariant
that no row or column sum is larger than it is in $\ma$. Note that
if $\ma_{ii}=0$ for all $i\in[n]$, then $\widehat{\ma}_{ii}=0$ for
all $i\in[n]$, and it can be ensured that $\mE_{ii}=0$ for all $i\in[n]$.

Finally, we output $\widetilde{\ma}=(1+\epsilon')^{-1}\widehat{\ma}+\mE$.
By construction $\nnz(\widetilde{\ma})=O(n\epsilon^{-2}\log(n/p))$ and $\ma$ and $\widetilde{\ma}$ have the same row and column sums, i.e.
$\widetilde{\ma} \vones= \ma \vones$,
$\widetilde{\ma}^\top \vones = \ma^\top \vones$. All that remains, is to show that the last property holds, ie., that $\widetilde{\ma}$ is still a good approximation of $\ma$. The previous conditions, together with \eqref{eq:nonnegat_approx_eq1} imply that
\[
\norm{\mr^{-1}\mE}_{\infty}=\norm{\mr^{-1}\mE\allones}_{\infty}=\norm{\mr^{-1}(\ma-(1+\epsilon')^{-1}\widehat{\ma})\allones}_{\infty}
\leq
\left(1-\frac{1-\epsilon'}{1+\epsilon'}\right)\leq2\epsilon'\,.
\]
and similarly, 
\[
\norm{\mc^{-1}\mE}_{1}=\norm{\mE^{\top}\mc^{-1}\allones}_{\infty}
= \norm{( \ma - (1+\epsilon')^{-1} \widehat{\ma} )^{\top}\mc^{-1}\allones}_{\infty}
\leq \left(1-\frac{1-\epsilon'}{1+\epsilon'}\right)
\leq2\epsilon'\,{.}
\]
Applying Lemma~\ref{lem:matrix_two_norm} then yields:
\begin{align*}
&\norm{\mr^{-1/2}(\ma-\widetilde{\ma})\mc^{-1/2}}_{2} \\
&\leq\norm{\mr^{-1/2}(\ma-\widehat{\ma})\mc^{-1/2}}_{2}+\left(1-\frac{1}{1+\epsilon'}\right)\norm{\mr^{-1/2}\widehat{\ma}\mc^{-1/2}}_{2}+\norm{\mr^{-1/2}\mE\mc^{-1/2}}_{2} \\
&\leq\epsilon'+\frac{\epsilon'}{1+\epsilon'}+2\epsilon'\leq4\epsilon'=\epsilon\,.
\end{align*}
\end{proof}

This lemma immediately implies the following fact on providing sparse approximations to directed (not necessarily Eulerian) Laplacians.

\begin{cor}[Crude Sparsification of Directed Laplacian]
\label{cor:sparsifying_directed_laplacian}  Let $\mlap=\md-\ma^{\top}\in\R^{n\times n}$
be a directed Laplacian associated with a (not necessarily Eulerian)
graph $G$ that has edges incident to at most $v$ vertices and let $\epsilon, p \in (0,1)$. The routine $\sparsifysubgraph(\mlap,p,\epsilon)$ computes a directed Laplacian $\widetilde{\mlap}=\md-\widetilde{\ma}^{\top}$ in time $O(\nnz(\mlap)+v\epsilon^{-2}\log(v/p))$ such that 
\begin{enumerate}
\item $\widetilde{\mlap}$ is sparse, i.e.
$\nnz(\widetilde{\mlap})=O(v\epsilon^{-2}\log(v/p))$, 
\item the in and out degrees of the graphs associated with $\mlap$ and $\widetilde{\mlap}$ are the same, i.e. $\ma \vones = \widetilde{\ma} \vones$ and $\ma^\top \vones = \widetilde{\ma}^\top \vones$,
\item  $\norm{\md_{in}^{-1/2}(\mlap-\widetilde{\mlap})\md_{out}^{-1/2}}_{2}\leq\epsilon$ with probability at least $1-p$  where $\md_{in} = \mdiag(\ma^\top \vones)$ and $\md_{out} = \mdiag(\ma \vones)$ are the diagonal matrices associated with the in and out degrees of $G$.
\end{enumerate}
\end{cor}
\begin{proof}
This follows by applying the sampling result from Lemma~\ref{lem:sparsifying_adjacency_matrix} 
to $\ma$, after removing the rows and columns corresponding to isolated
vertices. The guarantees still hold on the larger matrix after inserting the zero rows and columns back. We can then substitute the corresponding directed Laplacians matrices into the formulas, since the diagonal terms will cancel (note that this follows again from Lemma~\ref{lem:sparsifying_adjacency_matrix}, since the sparsified Laplacian has the same out degrees as the original one).
\end{proof}

Using this, we prove the main result of this section,
how to sparsify a subgraph that is contained in an expander, or more formally a particular undirected graph with large spectral gap. For notational convenience, we first formally define the type of graph symmetrization we consider for these subgraph arguments. Using this notation we then provide the main result of this subsection, Lemma~\ref{lem:subgraph_sparse}. Pseudocode of this routine is in Figure~\ref{fig:sparsifySubgraph}.

\begin{definition}[Graph Symmetrization]
For a directed Laplacian, $\mlap = \md - \ma^\top$, its \emph{graph symmetrization} $\ms_\mlap$ is the symmetric Laplacian given by $\ms_\mlap \defeq \mdiag(\mU_\ma \vones) - \mU_{\ma}$. Equivalently, if we considering the graph associated with $\mlap$ and replace every directed edge with an undirected edge of half the weight, then $\ms_\mlap$ is the symmetric Laplacian associated with this undirected graph.
\end{definition}

\begin{figure}[ht]
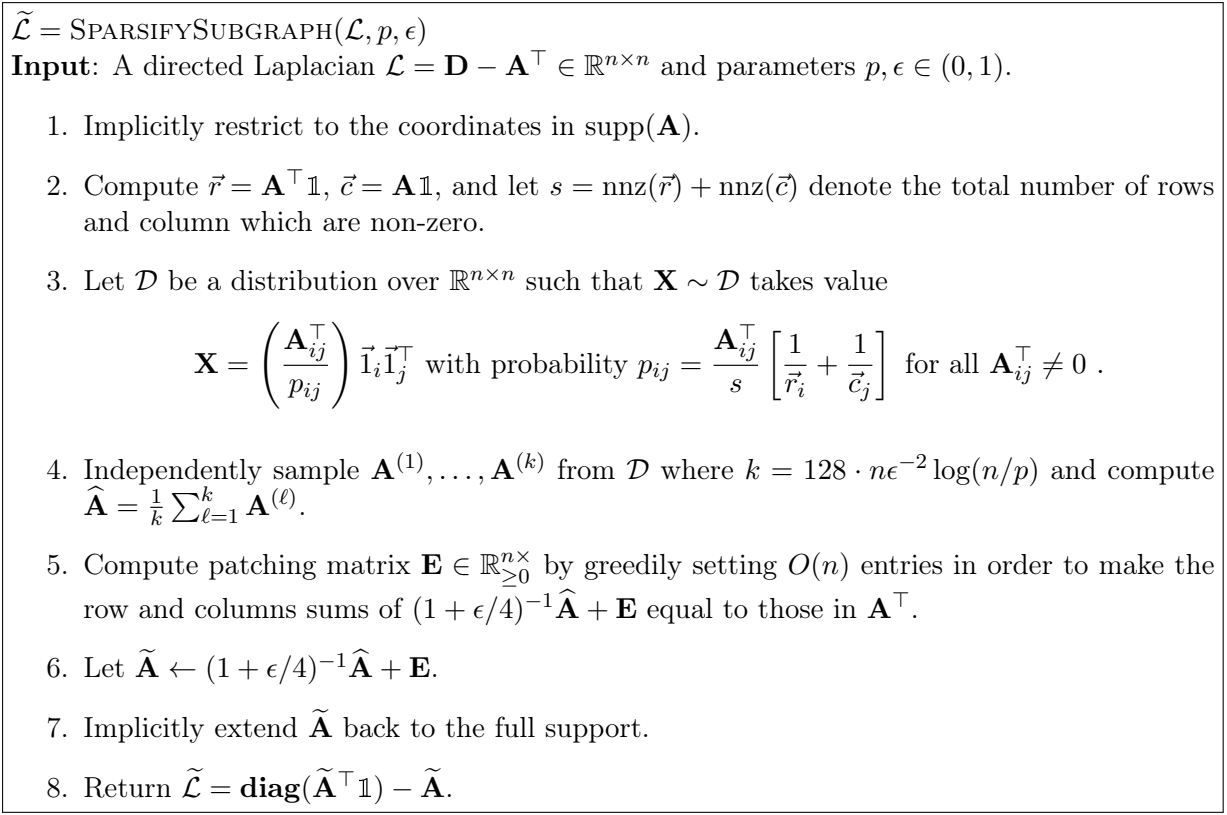

	\begin{algbox}
		$\widetilde{\mlap}=\sparsifysubgraph(\mlap, p, \epsilon)$\\
		\textbf{Input}: A directed Laplacian $\mlap = \md - \ma^\top \in \R^{n \times n}$ and parameters $p, \epsilon \in (0,1)$.
		\begin{enumerate}
		        \item Implicitly restrict to the coordinates in $\supp(\ma)$.
		        \item Compute $\vec{r} = \ma^\top \vones$, $\vec{c} = \ma \vones$, and let $s = \nnz(\vr) + \nnz(\vc)$ denote the total number of rows and column which are non-zero.
		        \item Let $\dist$ be a distribution over $\R^{n \times n}$
		        such that $\mx \sim \dist$ takes value
		        \[
		        \mx = \left(\frac{\ma_{ij}^\top}{p_{ij}} \right)\indic_{i}\indic_{j}^{\top}
		        \text{ with probability }
		        p_{ij} = \frac{\ma_{ij}^\top}{s} \left[\frac{1}{\vr_i} +\frac{1}{\vc_j} \right]
		        \text{ for all }
		        \ma_{ij}^\top \neq 0~.
		        \]
		        \item Independently sample $\ma^{(1)}, \dots, \ma^{(k)}$ from 
		       $\dist$ where $k = 128\cdot n\epsilon^{-2} \log(n/p)$ and compute $\widehat{\ma} = \frac{1}{k} \sum_{\ell=1}^{k} \ma^{(\ell)}$. 
			\item Compute patching matrix $\mE \in \R_{\geq 0}^{n \times }$ by greedily setting $O(n)$ entries in order to make the row and columns sums of $(1+\epsilon/4)^{-1}\widehat{\ma} + \mE$ equal to those in $\ma^\top$.
			\item Let $\widetilde{\ma} \leftarrow (1+\epsilon/4)^{-1} \widehat{\ma}+\mE$.
			\item Implicitly extend $\widetilde{\ma}$ back to the full support. 
			\item Return $\widetilde{\mlap} = \mdiag(\widetilde{\ma}^\top \vones)  - \widetilde{\ma}$.
		\end{enumerate}
	\end{algbox}
	\caption{Pseudocode of the subgraph sparsification routine}
	\label{fig:sparsifySubgraph}
\end{figure}

\begin{lem}[Subgraph Sparsification]
\label{lem:subgraph_sparse}  Let $\mlap$ be a directed Laplacian
and $\mU$ be an undirected Laplacian with spectral gap
at least $\alpha$, support size $v$, and  $\mdiag(\ms_{\mlap})\preceq\mdiag(\mU)$. 
For $\delta \leq \alpha \epsilon / 2$ the routine 
$\sparsifysubgraph(\mlap, p, \delta)$ in time 
$O(\nnz(\mlap)+ v \delta^{-2} \log(v/p))$ computes a Laplacian  $\widetilde{\mlap}$ with the same in and out degrees as $\mlap$, $\nnz(\widetilde{\mlap}) = O(v \delta^{-2} \log(v/p)))$, and 
$
\norm{\mU^{\dagger/2}(\mlap-\widetilde{\mlap})\mU^{\dagger/2}}_{2}\leq\epsilon
$ with probability at least $1-p$.
\end{lem}
\begin{proof}
Without loss of generality let $\mlap=\md_{out}-\ma^{\top}$.
Furthermore, let $\md=\mdiag(\mU)$, $\md_{in} = \mdiag(\ma \vones)$, and $\md_{out} = \mdiag(\ma^\top \vones)$. By Corollary~\ref{cor:sparsifying_directed_laplacian}, we can compute a  directed Laplacian
$\widetilde{\mlap}=\md_{out}-\widetilde{\ma}$ with  in and out degrees  
being the same as those of $\mlap$, $\nnz(\widetilde{\mlap})=O(v \delta^{-2}\log(v/p))$,
and with probability at least $1 - p$
\[
\norm{\md_{in}^{-1/2}(\mlap-\widetilde{\mlap})\md_{out}^{-1/2}}_{2}=\norm{\md_{in}^{-1/2}(\ma-\widetilde{\ma})\md_{out}^{-1/2}}_{2}\leq \delta\,.
\]
If this happens, then since $\mlap$ and $\widetilde{\mlap}$ have the same in and out degrees
and $\mdiag(\ms_{\mlap})\preceq\mdiag(\mU)$ we have that $\mU$ has larger
support than $\mlap$ and $\widetilde{\mlap}$, and therefore $\ker(\mU)\subseteq\ker(\mlap-\widetilde{\mlap})\cap\ker((\mlap-\widetilde{\mlap})^{\top})$.
Consequently, by Lemma~\ref{lem:asym_strong_equiv} we have 
\[
\norm{\mU^{\dagger/2}(\mlap-\widetilde{\mlap})\mU^{\dagger/2}}_{2}
=\max_{\vec{x},\vec{y} \neq 0}
	\frac{\vec{x}^{\top}(\mlap-\widetilde{\mlap})\vec{y}}
	{\sqrt{\vec{x}^{\top}\mU \vec{x}\cdot \vec{y}^{\top}\mU \vec{y}}}\,.
\]
Now, clearly, there is a maximizing $\vec{x}, \vec{y}\perp\allones$.
Consequently
$\vec{x}' = \vec{x} - \frac{\vec{x}^{\top}d}{\norm{\vec{d}_{1}}}\allones$
and $\vec{y}' = \vec{y} - \frac{\vec{y}^{\top}\vec{d}}{\norm{\vec{d}_{1}}}\allones$
are nonzero and satisfy $(\vec{x}')^{\top} \vec{d} = 0$
and $(\vec{y}')^{\top} \vec{d} = 0$.
The spectral gap of $\mU$ being at least $\alpha$
implies
\[
\mU \succeq \alpha \left(\md -
	\frac{1}{\norm {\vec{d}}_{1}}\vec{d}\vec{d}^{\top}\right)~,
\]
which then gives
\begin{align*}
\norm{\mU^{\dagger/2}(\mlap-\widetilde{\mlap})\mU^{\dagger/2}}_{2} & =\frac{(\vec{x}')^{\top}(\mlap-\widetilde{\mlap})\vec{y}'}
	{\sqrt{(\vec{x}')^{\top}\mU \vec{x}'\cdot(\vec{y}')^{\top}\mU \vec{y}'}}
\leq \frac{1}{\alpha}\cdot\frac{(\vec{x}')^{\top}(\mlap-\widetilde{\mlap})\vec{y}'}{\sqrt{(\vec{x}')^{\top}\md \vec{x}' \cdot(\vec{y}')^{\top}\md \vec{y}'}}\,{.}
\end{align*}
Since $\md_{in}\preceq2\cdot\md$ and $\md_{out}\preceq2\cdot\md$,
applying Lemma~\ref{lem:asym_strong_equiv} again yields
\[
\norm{\mU^{\dagger/2}(\mlap-\widetilde{\mlap})\mU^{\dagger/2}}_{2}
\leq\frac{2}{\alpha}\cdot\frac{(\vec{x}')^{\top}(\mlap-\widetilde{\mlap}) \vec{y}'}{\sqrt{(\vec{x}')^{\top}\md_{in} \vec{x}'\cdot(\vec{y}')^{\top}\md_{out} \vec{y}'}}
\leq\frac{2}{\alpha}\cdot\norm{\md_{in}^{-1/2}(\ma-\widetilde{\ma})\md_{out}^{-1/2}}_{2}\leq \frac{2}{\alpha} \delta
\]
and the result follows by our restriction on $\delta$. 
\end{proof}

\subsection{Sparsifying an Eulerian Laplacian\label{sub:sparse_eulerian}}

Here we show how to produce an $\epsilon$-sparsifier
of an Eulerian Laplacian in nearly linear time.
We achieve this by applying our result on sparsifying
subgraphs (Lemma~\ref{lem:subgraph_sparse})
proved in Section~\ref{sub:sample_directed} on a
decomposition of the Eulerian graph into well-connected
pieces on an associated undirected graph.

The decomposition we use is essentially identical to the
expander decomposition used in
Spielman and Teng's work on graph sparsification~\cite{SpielmanT11}.
Interestingly, the quality of our decomposition
is measured only in terms of properties 
of the symmetrized graph, rather than of the original directed graph.
Ultimately, only the sampling probabilities that we use on the decomposition take into account edge direction.\footnote{Even this can possibly be overcome by 
choosing different sampling probabilities.
It is the patching of the graph, i.e. adding edges to preserve
degree imbalance, where we truly use the directed structure
of the graph.}

Below we formally define the type of decomposition we need,
and provide a theorem about computing such decompositions.
Finding these decompositions has been done in prior works~\cite{SpielmanTengSolver:journal, KLOS14, PengS14, OrecchiaV11},
and we defer the discussion of it to Appendix~\ref{sec:decomposition}.
\begin{definition}
An \emph{$(s,\alpha,\beta)$-decomposition} of a directed Laplacian $\mlap$  is a decomposition of $\mlap$ into directed Laplacians $\mlap^{(1)},...,\mlap^{(k)}\in\R^{n\times n}$, i.e. $\mlap = \sum_{i \in {k}} \mlap^{(i)}$, such that $\sum_{i\in[k]}\left|\supp(\mlap^{(i)})\right|\leq s$ and such that there exists undirected Laplacians $\mU^{(1)},...,\mU^{(k)}$ such that:
\begin{enumerate}
\item $\mdiag(\ms_{\mlap^{(i)}}) \preceq \mdiag(\mU^{(i)})$,
	for all $i\in[k]$,
\item $\mU^{(i)}$ has spectral gap at least $\alpha$,
	for all $i\in[k]$, and
\item $\sum_{i\in[k]}\mU^{(i)}\preceq\beta\mU_{\mlap}$.
\end{enumerate}
We call the $\mU^{(1)}, ... ,\mU^{(k)}$ with these properties an \emph{$(\alpha,\beta)$ undirected cover} of $\mlap^{(1)}, ... ,\mlap^{(k)}$.
\end{definition}

\begin{thm}
\label{thm:decomposition_thm} Given a directed Laplacian, $\mlap \in \R^{n \times n}$, the routine $\expanderpartition(\mlap)$ returns an $(\otilde(n),1/\alpha, \beta)$-decomposition, with $\alpha, \beta = \tilde{O}(1)$, in $\otilde(\nnz(\mlap))$ time.
\end{thm}

We produce our sparsifiers by computing the decomposition using Theorem~\ref{thm:decomposition_thm}
and then applying Lemma~\ref{lem:subgraph_sparse} repeatedly to
obtain the sparsifier.
Pseudocode of this algorithm is given
in Figure~\ref{fig:sparsifyEulerian}.
\begin{figure}[ht]
	\begin{algbox}
		$\widetilde{\mlap}=\sparsifyeulerian(\mlap, p, \epsilon)$\\
		\textbf{Input}:
		$\mlap$ an $n \times n$ directed Laplacian, parameters $p, \epsilon \in (0,1)$.
		\begin{enumerate}
			\item $((\mlap^{(1)}, \dots, \mlap^{(k)}), \alpha, \beta) \leftarrow \expanderpartition(\mlap)$.
		        \item For $i = 1, \dots, k$
			\begin{enumerate}
			\item $\widetilde{\mlap}^{(i)} \leftarrow \textsc{SparsifySubgraph}(\mlap^{(i)}, p/n^2, \epsilon/(2\alpha\beta))$.
			\end{enumerate}
			\item Return $\widetilde{\mlap} = \sum_{i=1}^k \widetilde{\mlap}^{(i)}$.
		\end{enumerate}
	\end{algbox}
	\caption{Pseudocode of the Eulerian graph sparsification routine}
	\label{fig:sparsifyEulerian}
\end{figure}
and the analysis of this algorithm is given in the following theorem.
\begin{thm}
\label{thm:order_n_sparsifier} For Eulerian Laplacian $\mlap\in\R^{n\times n}$
and $\epsilon,p\in(0,1)$ with probability at least $1 - p$ the routine $\textsc{SparsifyEulerian}(\mlap, p, \epsilon)$ computes 
in $\otilde(\nnz(\mlap)+n\epsilon^{-2}\log(1/p))$
time an Eulerian Laplacian $\widetilde{\mlap}\in\R^{n\times n}$ such that
\begin{enumerate}
\item $\widetilde{\mlap}$ is an
$\epsilon$-sparsifier
of $\mlap$,
\item the in and out degrees of the graphs associated with $\mlap$ and $\widetilde{\mlap}$ are identical.
\end{enumerate}
\end{thm}
\begin{proof}
Using the $\expanderpartition$ routine (Theorem~\ref{thm:decomposition_thm}) we compute Laplacians
$\mlap^{(1)},...,\mlap^{(k)}\in\R^{n\times n}$ that are a $(s,\alpha,\beta)$
partition of $\mlap$ with $(\alpha,\beta)$ undirected cover $\mU^{(1)},...,\mU^{(k)}$
for $s=\otilde(n)$, $\alpha=1/\otilde(1)$, and $\beta=\otilde(1)$.

We then apply the $\sparsifysubgraph$ routine (Lemma~\ref{lem:subgraph_sparse}) to each $\mlap^{(i)}$
to compute $\widetilde{\mlap}^{(i)}$ in $O(\nnz(\mlap)+s\epsilon^{-2}\beta^{-2}\alpha^{-2}\log (n/p))$
time  such that each $\widetilde{\mlap}^{(i)}$ has the same in and out degree as
$\mlap^{(i)}$ and $\norm{\left(\mU^{(i)}\right)^{+1/2}(\mlap^{(i)}-\widetilde{\mlap}^{(i)})\left(\mU^{(i)}\right)^{+/2}}_{2}\leq \epsilon / \beta$.  The running time follows
from the fact that $\sum_{i\in[k]}\left|\supp(\mlap^{(i)})\right|\leq s$ and thats happens with probability at least $1 - p$ by union bounding over the success probability of each call to $\sparsifysubgraph$.

Finally, considering $\widetilde{\mlap}=\sum_{i\in[k]}\widetilde{\mlap}^{(i)}$,  Lemma~\ref{lem:asym_strong_equiv}
yields that for all $\vec{x}, \vec{y}\neq0$ it is the case that
\[
\vec{x}^{\top}(\mlap-\widetilde{\mlap})\vec{y}
=\sum_{i\in[k]} \vec{x}^{\top}(\mlap^{(i)}-\widetilde{\mlap}^{(i)}) \vec{y}
\leq \sum_{i\in[k]}\frac{\epsilon}{2 \beta}
	\left[\vec{x}^{\top}\mU^{(i)}\vec{x} + 
		\vec{y}^{\top}\mU^{(i)}\vec{y}\right]
\leq\frac{\epsilon \beta}{2 \beta} 
	\left[\vec{x}^{\top}\mU_{\mlap} \vec{x}
		+ \vec{y}^{\top}\mU_{\mlap}\vec{y}\right]\,
\]
The result follows from Lemma~\ref{lem:asym_strong_equiv}
applied above and the the bounds on $s$,$\alpha$, and $\beta$. Note that the fact that in and out degrees are preserved is guaranteed by the fact that for each component in the decomposition the degrees are preserved, according to Lemma~\ref{lem:subgraph_sparse}.
\end{proof}

\subsection{Sparsifying a Squared Eulerian Laplacian\label{sub:sparse_square}}

\newcommand{\mlapsquared}{\mathcal{M}}
\newcommand{\magen}{\WWhat}

Here we build upon Section~\ref{sub:sparse_eulerian} and show how to
sparsify certain implicitly represented Eulerian Laplacians. In particular,
given an Eulerian Laplacian $\mlap = \md-\ma^{\top}$ associated with
a strongly connected graph we show how to compute a sparsifier
for the Eulerian Laplacian $\mlapsquared = \md-\ma^{\top}\md^{-1}\ma^{\top}$
in nearly-linear time with respect to $\mlap$, i.e. without explicitly constructing $\mlapsquared$. Note that running time we achieve may be sublinear in size of the matrix we are sparsifying, i.e
$\mlapsquared$.

Our approach is a natural directed extension of the approach taken by Peng and Spielman \cite{PengS14} for solving the same problem in the case when $\mlap$ is symmetric. Broadly speaking, we decompose $\mlapsquared$ into a directed Laplacian for each vertex. Each of these directed Laplacians may be dense but we show that they have a compact representation that allows us to efficiently implement a sampling scheme analogous to $\sparsifysubgraph$ for each of these Laplacians such that adding these approximation yields an Eulerian approximation to $\mlapsquared$ that has $\nnz(\mlap)$ non-zero entries. Applying $\sparsifyeulerian$ from Section~\ref{sub:sparse_eulerian} to the result then yields our desired sparsifier. 

Formally, we consider the slightly more general setting where we have a square matrix with non-negative entries, $\magen \in \R^{n \times n}_{\geq 0}$, that has the same same row and column sums, i.e. $\magen \vones = \magen^\top \vones$. We show how to compute a sparsifier of $\mlapsquared = \md - \magen \md^{-1} \magen$ for $\md = \mdiag(\magen \vones)$ in time nearly linear in $\nnz(\magen)$. This setting is more general as we allow entries on the diagonal of the squared matrix. We consider this case as it simplifies our analysis in Section~\ref{sec:solver}.

As discussed, we first decompose $\mm$ into a directed Laplacian for each vertex. For $i \in [n]$ we let 
\[
\mlap^{(i)} = \mdiag(\magen_{i,:}) - \frac{1}{\md_{i,i}} \magen_{:,i} \magen_{i,:}^\top
\]
where $\magen_{i,:}, \magen_{:,i} \in \R^{n}$ are the vectors corresponding to the row and column $i$ of $\magen$ respectively. In Theorem~\ref{thm:sparsify_square} we show that $\mlapsquared$ and each $\mlap^{(i)}$ are directed Laplacians such that $\mlapsquared = \sum_{i \in [n]} \mlap^{(i)}$. 

Note that while $\mlapsquared$ may be dense and forming each of the $\mlap^{(i)}$ explicitly maybe expensive, we have a compact representation of each $\mlap^{(i)}$ in terms of a single row and column of $\magen$. Moreover, if we look at the total support of $\mlap^{(i)}$ then since each row and column only appears once we see that the total support is just $O(\nnz(\magen))$. Furthermore, we show that due to the low rank structure of the graph symmetrization, $\ms_{\mlap^{(i)}}$, of each $\mlap^{(i)}$ has spectral gap of at least a constant (See Lemma~\ref{lem:bipartite_conductance}). Consequently, if we apply $\sparsifysubgraph$ to each $\mlap^{(i)}$ and sum the results, the analysis in Section~\ref{sub:sparse_eulerian} would imply that this matrix would be an approximation to $\mlapsquared$ with $O(\nnz(\magen))$ non-zero entries. 

The only difficulty in following this approach is to show that we can apply $\sparsifysubgraph$ to each $\mlap^{(i)}$ efficiently. We show that we can perform this operation in time proportional to the number of non-zeros in each row and column of $\magen$, rather than the naive $O(\nnz(\mlap^{(i)}))$ running time which could be much larger. To do this, as in Peng and Spielman \cite{PengS14}, we exploit the simple product structure of $\mlap^{(i)}$. Since $\magen$ has the same row and column sums this means that $\norm{\magen_{:,i}}_1 = \norm{\magen_{i,:}}_1 = \md_{i,i}$. Consequently, each $\mlap^{(i)}$ is of the form $\mn = \mdiag(\vy) - \frac{1}{\norm{\vy}_1} \vx \vy^\top$ for some $\vx,\vy \in \R^{n}_{\geq 0}$ with $\norm{\vx}_1 = \norm{\vy}_1$ that depend on $i$. Since, the off-diagonals and their row and column sums have simple closed form expressions we can show that with $O(\nnz(\vx) + \nnz(\vy))$ preprocessing time, we can sample from the distribution required to apply $\sparsifysubgraph$ on this matrix in $O(1)$ time. Consequently, we can implement the approach in our desired running time. 

In the remainder of this section we provide pseudocode for these routines and prove their correctness. First, in Figure~\ref{fig:sparsifyProduct} we provide $\sparsifyproduct$ the pseudocode for sparsifying matrices of the form $\mn$ given above. In Lemma~\ref{lem:bipartite_conductance} we prove that the graph symmetrizations, $\ms_{\mn}$, of these matrices have spectral gap of at least $1$ and using this fact in Lemma~\ref{lem:sparsify_product} we prove that $\sparsifyproduct$ does provide sparse approximations to these matrices. In Figure~\ref{fig:sparsifySquare} we provide $\sparsifysquare$ the pseudocode for producing sparsifiers of $\mlapsquared$ by invoking $\sparsifyproduct$ on each $\mlap^{(i)}$ and then invoking $\sparsifyeulerian$ on the result. Finally, we conclude the section with Theorem~\ref{thm:sparsify_square} which proves the correctness and analyzes the running time of $\sparsifyproduct$.   

\begin{figure}[ht]
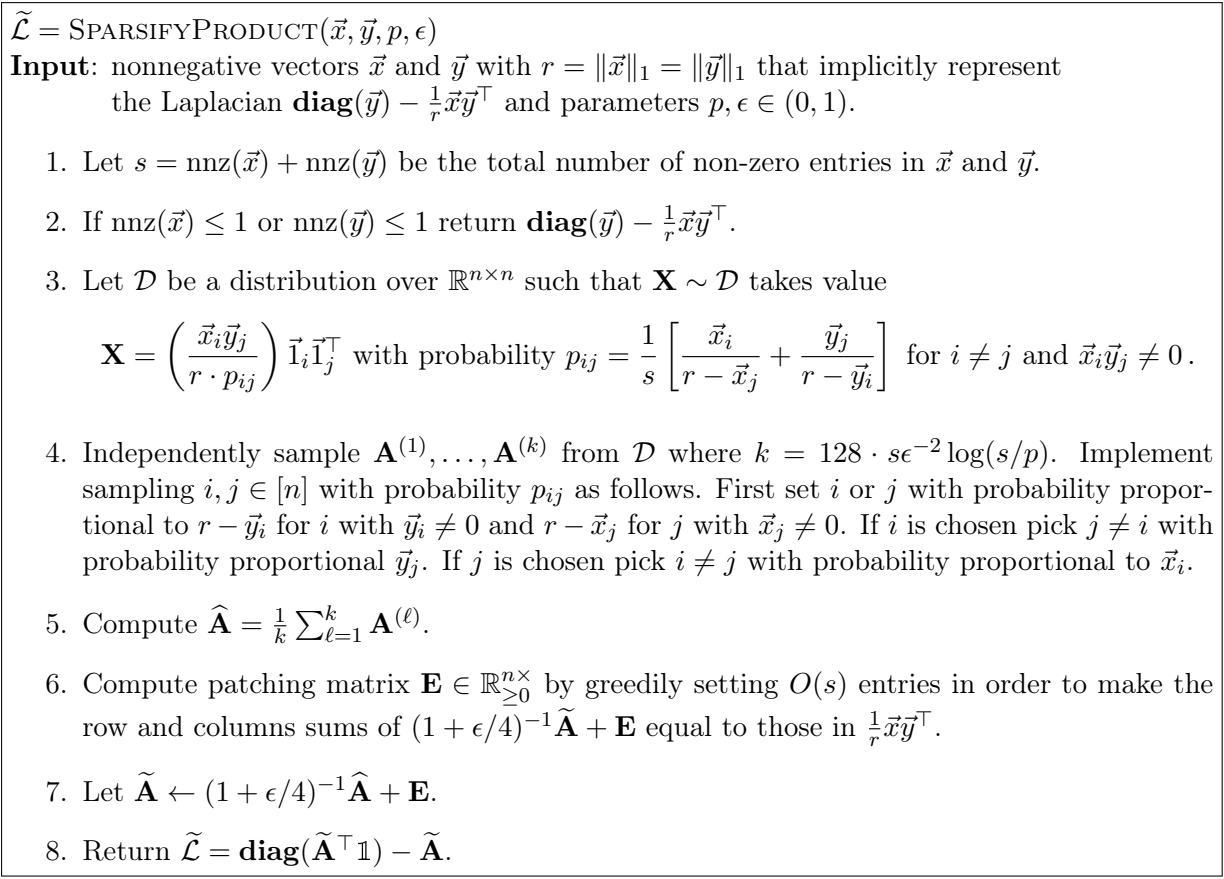

	\begin{algbox}
		$\widetilde{\mlap}=\sparsifyproduct(\vec{x}, \vec{y}, p, \epsilon)$\\
		\textbf{Input}: nonnegative vectors $\vec{x}$ and $\vec{y}$
		with $r = \norm{\vec{x}}_1 = \norm{\vec{y}}_1$
		that implicitly represent \\
		 \makebox[1.22cm]{}	 the Laplacian
		$\mdiag(\vec{y}) - \frac{1}{r} \vec{x} \vec{y}^{\top}$ and
		 parameters $p, \epsilon \in (0,1)$.
		\begin{enumerate}
			\item Let $s = \nnz(\vx) + \nnz(\vy)$ be the total number of non-zero entries in $\vx$ and $\vy$.	
			\item If $\nnz(\vx) \leq 1$ or $\nnz(\vy) \leq 1$ return $\mdiag(\vec{y}) - \frac{1}{r} \vec{x} \vec{y}^{\top}$.			
			\item Let $\dist$ be a distribution over $\R^{n \times n}$
			such that $\mx \sim \dist$ takes value
			\[
			\mx = \left(\frac{\vx_{i} \vy_{j}}{r \cdot p_{ij}} \right)\indic_{i}\indic_{j}^{\top}
			\text{ with probability }
			p_{ij} = \frac{1}{s} \left[\frac{\vx_i}{r - \vx_j} + \frac{\vy_j}{r - \vy_i} \right]
			\text{ for } i \neq j \text{ and }
			\vx_{i} \vy_{j} \neq 0\,.
			\]
			\item Independently sample $\ma^{(1)}, \dots, \ma^{(k)}$ from 
			$\dist$ where $k = 128 \cdot s \epsilon^{-2} \log(s/p)$.  Implement sampling $i,j\in[n]$ with probability $p_{ij}$ as follows.
			First set $i$ or $j$ with probability proportional to $r - \vy_i$ for $i$ with $\vy_i \neq 0$ and $r - \vx_j$ for $j$ with $\vx_{j} \neq 0$. If $i$ is chosen pick $j \neq i$ with probability proportional $\vy_j$. If $j$ is chosen pick $i \neq j$ with probability proportional to $\vx_i$. 
			\item Compute $\widehat{\ma} = \frac{1}{k} \sum_{\ell=1}^{k} \ma^{(\ell)}$. 
			\item Compute patching matrix $\mE \in \R_{\geq 0}^{n \times }$ by greedily setting $O(s)$ entries in order to make the row and columns sums of $(1+\epsilon/4)^{-1}\widetilde{\ma} + \mE$ equal to those in $\frac{1}{r} \vx \vy^\top$.
			\item Let $\widetilde{\ma} \leftarrow (1+\epsilon/4)^{-1} \widehat{\ma}+\mE$.
			\item Return $\widetilde{\mlap} = \mdiag(\widetilde{\ma}^\top \vones)  - \widetilde{\ma}$. 
		\end{enumerate}
	\end{algbox}
	\caption{Pseudocode for sparsifying a single product graph}
	\label{fig:sparsifyProduct}
\end{figure}

\begin{figure}[h!]
	\begin{algbox}
		$\widetilde{\magen}=\sparsifysquare(\magen, p, \epsilon)$\\
		\textbf{Input}: $\magen \in \R^{n \times n}_{\geq 0}$ with $\magen \vones = \magen^\top \vones$ that implicitly represents the Laplacian \\
		\makebox[1.22cm]{}	 $\mlapsquared = \md - \magen \md^{-1} \magen$ for $\md = \mdiag(\magen \vones)$ and parameters $p, \epsilon \in (0, 1)$.
		\begin{enumerate}
			\item For all $i = 1, \dots, n$,
			\begin{enumerate}
			\item Let $\magen_{i,:}, \magen_{:,i} \in \R^{n}$ denote row and column $i$ of $\magen$ respectively.
			\item  $\widetilde{\mlap}^{(i)} \leftarrow 
			\sparsifyproduct(\magen_{:,i}, \magen_{i,:}, p/{2n}, \epsilon/6)$ 
			\end{enumerate}
			\item Let $\widehat{\mathcal{M}} = \sum_{i=1}^n \widetilde{\mlap}^{(i)}$.
			\item $\widetilde{\mathcal{M}} \leftarrow \sparsifyeulerian(\widehat{\mathcal{M}}, p/2, \epsilon/3)$.
			\item Return $\md - \widetilde{\mathcal{M}}$
		\end{enumerate}
	\end{algbox}
	\caption{Pseudocode for producing an
	$\epsilon$-sparsifier of $\mathcal{M}=\md-\magen^{\top}\md^{-1}\magen^{\top}$,
	}
	\label{fig:sparsifySquare}
\end{figure}

\begin{lemma}
\label{lem:bipartite_conductance}
For $\vec{x}, \vec{y}\in\R_{\geq0}^{n}$ with
$r=\norm {\vec{x}}_{1}=\norm {\vec{y}}_{1}>0$
and $\my=\mdiag(\vec{y})$,
the matrix $\mlap=\my-\frac{1}{r}\vec{x}\vec{y}^{\top}$ is a directed Laplacian and the spectral gap of its symmetrization $\ms_{\mlap}$ is at least $1$.
\end{lemma}

\begin{proof}
Note that $\mlap_{ij}=-\frac{1}{r}x_i y_j\leq 0$ for $i\neq j$ and $\allones^{\top}\mlap=\allones^{\top}\my-\frac{1}{r}\allones^\top \vec{x}\vec{y}^{\top}=\vec{y}^{\top}-\vec{y}^{\top}=\allzeros^{\top}$. Consequently, $\mlap$ is a directed Laplacian. All that remains is to lower bound the spectral gap of $\ms_{\mlap}$. 

Letting $\mx = \mdiag(\vec{x})$, we see that the graph symmetrization $\ms_{\mlap}$ of $\mlap$ is the undirected Laplacian  
\[
\ms_{\mlap}
=
\frac{1}{2}
\left(\mx+\my - \frac{1}{r}\left(\vec{x}\vec{y}^{\top} +\vec{y}\vec{x}^{\top}\right)\right)\,{.}
\]
Furthermore, the diagonal entries of $\ms_{\mlap}$, denoted $\vd = \diag(\ms_{\mlap})$, are given by
\[
\vec{d}_i={[\ms_{\mlap}]}_{ii}
=\frac{1}{2}(\vec{x}_i+\vec{y}_i)-\frac{1}{r}\vec{x}_i \vec{y}_i
\leq \frac{1}{2}(\vec{x}_i+\vec{y}_i)\,{.}
\]
Now recall that the spectral gap of $\ms_{\mlap}$ is defined to be the smallest nonzero eigenvalue of the normalized Laplacian $\md^{-1/2}\ms_{\mlap}\md^{-1/2}$,
where $\md=\mdiag(\vec{d})$.
\newcommand{\boundMat}{\mm}
Since $d_i\leq \frac{1}{2}\left(\vec{x}_i+\vec{y}_i\right)$, we have 
\[
\md^{-1/2}\ms_{\mlap}\md^{-1/2}
\succeq \left(\frac{1}{2}\left(\mx+\my\right)\right)^{-1/2} \ms_{\mlap} \left(\frac{1}{2}\left(\mx+\my\right)\right)^{-1/2}
=2\left(\mx+\my\right)^{-1/2} \ms_{\mlap} \left(\mx+\my\right)^{-1/2}=: \boundMat.
\]
This implies that the eigenvalues of  $\md^{-1/2}\ms_{\mlap}\md^{-1/2}$ dominate those of $\boundMat$.  The multiplicity of zero as an eigenvalue is the same for the two matrices, so it thus suffices to show that the smallest nonzero eigenvalue of $\boundMat$ is at least 1. 
Plugging the definition of $\ms_{\mlap}$ in our expression for $\boundMat$ gives
\begin{align*}
\boundMat
&=2\left(\mx+\my\right)^{-1/2} \ms_{\mlap} \left(\mx+\my\right)^{-1/2}\\
&=\left(\mx+\my\right)^{-1/2}\left(\mx+\my - \frac{1}{r}\left(\vec{x}\vec{y}^{\top} +\vec{y}\vec{x}^{\top}\right)\right) \left(\mx+\my\right)^{-1/2}%
=\mI-\mn,
\end{align*}
where
$\mn=\frac{1}{r}\left(\mx+\my\right)^{-1/2}\left(\vec{x}\vec{y}^{\top} +\vec{y}\vec{x}^{\top}\right)\left(\mx+\my\right)^{-1/2}$.
The matrix $\mn$ has rank 2, so $\mI-\mn$ has at most 2 eigenvalues that are not equal to 1.  Furthermore, we know 
that $\mm$ has a nontrivial kernel, so at least one of these eigenvalues is 0. 
Let $\lambda$ be the one remaining eigenvalue.  Since $\tr(\mm)$ equals the sum of these eigenvalues, we have 
$\tr(\mm)=(n-2)\cdot 1 + 0 +\lambda = n-2+\lambda$, so 
\[\lambda=\tr(\mm)-n+2=\tr(\mI)-\tr(\mn)-n+2 = 2-\tr(\mn).\]
The inequality $\frac{2ab}{a+b}\leq \frac{a+b}{2}$ between the harmonic and arithmetic means then gives
\[
\tr(\mn)
= \sum_{i \in [n]} \mn_{ii}=\frac{1}{r}\sum_{i \in [n]} \frac{2\vec{x}_i \vec{y}_i}{\vec{x}_i+\vec{y}_i}
\leq \frac{1}{r}\sum_{i \in [n]} \frac{\vec{x}_i+\vec{y}_i}{2}=1,
\]
so $\lambda=2-\tr(\mn) \geq 1$.  The nonzero eigenvalues of $\mm$ are thus all at least 1, as desired.
\end{proof}

\begin{lem}[Product Sparsification]
	\label{lem:sparsify_product} 
	Let $\vx, \vy \in \R^{n}_{\geq 0}$ be non-negative vectors with $\norm{\vx}_1 = \norm{\vy}_1 = r$ and let $\epsilon, p \in (0,1)$. Furthermore, let $s$ denote the total number of non-zero entries in $\vx$ and $\vy$, i.e. $s = \nnz(\vx) + \nnz(\vy)$ and let $\mlap \defeq \mdiag(\vy) - \frac{1}{r} \vx \vy^\top$. The routine $\sparsifyproduct(\vx, \vy, p, \epsilon/2)$ in time 
	$O(s \epsilon^{-2} \log(s/p))$ computes with probability at least $1-p$ a Laplacian  $\widetilde{\mlap}$ with the same in and out degrees as $\mlap$, $\nnz(\widetilde{\mlap}) = O(s \epsilon^{-2}   \log(s/p)))$, and 
	$
	\norm{\ms_{\mlap}^{\dagger/2}(\mlap-\widetilde{\mlap})\ms_{\mlap}^{\dagger/2}}_{2}\leq\epsilon
	$.
\end{lem}

\begin{proof}
	Note that if $\nnz(\vx) \leq 1$ or $\nnz(\vy) \leq 1$. Then clearly $\frac{1}{r} \vx \vy^\top$ has at most $s$ non-zero entries and $\mlap$ has $O(s)$ non-zero entries. Furthermore, we can clearly compute $\mlap$ in $O(s)$ time and therefore the result follows. In the remainder of the proof we therefore assume that $\nnz(\vx) \geq 2$ and $\nnz(\vy) \geq 2$. 
	
	First, we show that $\widetilde{\mlap}$ is precisely the output of an execution of $\sparsifysubgraph(\mlap, p, \epsilon/2)$. Let $\mx = \mdiag(\vx)$, $\my = \mdiag(\vy)$, $\ma^\top \defeq \frac{1}{r} [\vx \vy^\top - \mx \my]$, and $\md = \my - \frac{1}{r} \mx \my$. Clearly, $\mlap = \md - \ma^\top$. Furthermore, we see that $\ma^\top$ is non-negative with a zero-diagonal $\diag(\ma^\top) = \vzero$ and therefore $\md$ is diagonal and this is the standard decomposition of $\mlap$. 
	
	Now as in $\sparsifysubgraph(\mlap, p, \epsilon)$ let $\vr = \ma^\top \vones = \vx - \frac{1}{r} \mx \my \vones$ and $\vc = \ma \vones = \vy - \frac{1}{r} \mx \my \vones$. Furthermore, since $\nnz(\vx) \geq 2$ and $\nnz(\vy) \geq 2$ we see that a row or column of $\frac{1}{r} \vx \vy^\top$ is non-zero if and only if the corresponding row and column in $\ma^\top$ is non-zero and thus $s$ is the same in $\sparsifysubgraph(\mlap, p, \epsilon)$ and $\sparsifyproduct(\vx, \vy, p, \epsilon)$. Now, for all $i \neq j$  with $\ma_{ij}^\top = \vx_i \vy_j \neq 0$ we have
	\[
	\frac{\ma_{ij}^\top}{s} \left[\frac{1}{\vr_i} + \frac{1}{\vc_j}\right]
	= \frac{\vx_i \vy_j}{r \cdot s} \left[\frac{1}{x_i - \frac{1}{r} \vx_i \vy_i} + \frac{1}{{y_j - \frac{1}{r} \vx_j \vy_j}}\right]
	= \frac{1}{s} \left[\frac{\vx_i}{r - \vx_j} + \frac{\vy_j}{r - \vy_i} \right]\,.
	\]
	Consequently, we see that $\widetilde{\mlap}$ is precisely the output of an execution of $\sparsifysubgraph(\mlap, p, \epsilon/2)$. Furthermore, since $\ms_\mlap$ has spectral gap at least $1$ by Lemma~\ref{lem:bipartite_conductance} we have by Lemma~\ref{lem:subgraph_sparse} which analyzed $\sparsifysubgraph$ that $\widetilde{\mlap}$ has the desired properties. 
	
	All the remains is to bound the running time of $\sparsifyproduct(\vx, \vy, p, \epsilon/2)$. Note that computing $s$ takes $O(s)$ time and computing all the $r - \vy_i$ and $r - \vx_j$ can be done in $O(s)$ time. Consequently, with $O(s)$ preprocessing we can build a table so that each sample from $\dist$ takes $O(1)$ time and computing $\widehat{\ma}$ takes $O(s + k) = O(s \epsilon^{-2} \log(s/p)$ time. Furthermore, since computing the patching $\mE$ also takes $O(s + k)$ time we have that the total running time is as desired.
\end{proof}

\begin{thm}[Sparsifying Squares] \label{thm:sparsify_square}
Let $\magen \in \R^{n \times n}_{\geq 0}$ have the same row and column sums, i.e. $\magen \vones = \magen^\top \vones$, and let $\epsilon, p \in (0,1)$. Let $\md = \mdiag(\magen \vones)$, $\mlap = \md - \magen$, and $\mlapsquared = \md - \magen \md^{-1} \magen$. Both $\mlap$ and $\mlapsquared$ are Eulerian Laplacians and  in $\otilde(\nnz(\mlap)\epsilon^{-2}\log(n/p))$ time the routine $\sparsifysquare(\mlap,p,\epsilon)$ computes $\widetilde{\magen}$ such that with probability at least $1 - p$ the matrix $\mapp{\mathcal{M}} = \md - \widetilde{\magen}$ is an $\epsilon$-sparsifier of $\mathcal{M}$.
\end{thm}

\begin{proof}
Clearly, both $\magen$ and $\magen \md^{-1} \magen$ are entrywise non-negative and therefore the off diagonals of $\mlap$ and $\mlapsquared$ are non-positive. Furthermore, since clearly $\magen \vones = \magen^\top \vones = \magen \md^{-1} \magen \vones = [\magen \md^{-1} \magen]^\top \vones = \md \vones$ we have that $\mlap \vones = \mlap^\top \vones = \mlapsquared \vones = \mlapsquared^\top \vones = \vzero$ and both $\mlap$ and $\mlapsquared$ are Eulerian Laplacians. 

Next, for all $i \in [n]$ let $s_i = \nnz(\magen_{:,i}) + \nnz(\magen_{i,:})$ and
\[
\mlap^{(i)} = \mdiag(\magen_{i,:}) - \frac{1}{\md_{i,i}} \magen_{:,i} \magen_{i,:}^\top\,.
\]
Note that since $\magen \vones = \magen^\top \vones$ and $\magen$ is entrywise non-negative we have that $\md_{i,i} = \norm{\magen_{i,:}}_1 = \norm{\magen_{:,i}}_1$. Consequently, by Lemma~\ref{lem:sparsify_product} and union bound with probability at least $1 - \frac{1}{2p}$ it is the case that each $\widetilde{\mlap}^{(i)}$ is a directed Laplacian with the same in and out degrees as $\mlap^{(i)}$, $\nnz(\widetilde{\mlap}^{(i)}) = O(s_i \epsilon^{-2} \log(s n /p))$, and
\begin{equation}
\label{eq:square_subgraph_apx}
\left\Vert{\ms_{\mlap^{(i)}}^{\dagger/2}(\mlap^{(i)} - \widetilde{\mlap}^{(i)})
	\ms_{\mlap^{(i)}}^{\dagger/2}}\right\Vert_{2}\leq\epsilon/3\,{.}
\end{equation}
Furthermore, since clearly $\sum_{i} s_i = 2 \nnz(\magen)$ we have that 
\[
\nnz(\widehat{\mlapsquared}) \leq \sum_{i \in [n]} \nnz(\widetilde{\mlap}^{(i)}) 
\leq \sum_{i \in [n]} s_i \epsilon^{-2} \log(n s_i / p)
\leq O(\nnz(\magen) \epsilon^{-2} \log(n/p))
\]
and therefore the total running time for computing $\widehat{\mlapsquared}$ is $O(\nnz(\mlap) \epsilon^{-2} \log(n/p)$. Using Theorem~\ref{thm:order_n_sparsifier} to reason about the effect of $\sparsifyeulerian$ then completes our running time analysis and union bounding yields that $\widetilde{\mlapsquared}$ has the desired degrees and sparsity. Consequently, $\widetilde{\magen}$ also has the desired sparsity and since the degrees of the graph associated with $\mlapsquared$ are at most the degrees of the graph associated with $\mlap$ we see that $\widetilde{\magen} \in \R^{n \times n}_{\geq 0}$.

All that remains is to verify that $\widetilde{\mlapsquared}$ is an $\epsilon$-approximation of $\mathcal{M}$. By Lemma~\ref{lem:asym_strong_equiv}, \eqref{eq:square_subgraph_apx} implies that for all $\vec{x}, \vec{y}\neq0$ it is the case that
\[
\vec{x}^{\top}(\widehat{\mathcal{M}} - \mathcal{M})\vec{y}
=\sum_{i=1}^n\vec{x}^{\top}(\widetilde{\mlap}^{(i)}-\mlap^{(i)})\vec{y}
\leq\sum_{i=1}^n\frac{\epsilon}{3}
\left[\vec{x}^{\top}\ms_{\mlap^{(i)}}\vec{x} + 
\vec{y}^{\top}\ms_{\mlap^{(i)}}\vec{y}\right]
=\frac{\epsilon}{3}
	\left[\vec{x}^{\top}\mU_{\mathcal{M}}\vec{x}+\vec{y}^{\top}\mU_{\mathcal{M}}\vec{y}\right]
\]
where in the last identity we used that $\sum_{i=1}^n \ms_{\mlap^{(i)}} = \mU_{\mathcal{M}}$ from the fact that $\sum_{i \in [n]} \mlap^{(i)} = \mathcal{M}$. Applying Lemma~\ref{lem:asym_strong_equiv}
again on the above bound, we obtain that $\widehat{\mlapsquared}$ is an $\epsilon/3$ approximation of $\mlapsquared$. Since, Theorem~\ref{thm:order_n_sparsifier} implies that $\widetilde{\mlapsquared}$ is an $\epsilon/3$-approximation of $\widehat{\mlapsquared}$, invoking the transitivity bound, Lemma~\ref{lem:strong_transitivity}, yields that $\widetilde{\mlapsquared}$ is an $(\epsilon/3)(2 + \epsilon/3) \leq \epsilon$-approximation of $\mlapsquared$ as desired. 
\end{proof}

\section{Solving Directed Laplacian Systems}
\label{sec:solver}
In this section, we show how to solve directed Laplacian systems in almost-linear time.  Our main result is as follows:
\begin{theorem}\label{thm:laplacian_general}
  Let $\mm$ be an arbitrary $n\times n$ column-diagonally-dominant or row-diagonally-dominant matrix with diagonal $\md$ and
  $m$ non-zero entries.
Let $\kappa(\md)$
  be the ratio between the maximum and minimum diagonal
  entries of $\md$.
  Then for any $\vec{b}\in\im{\mm}$ and $0<\epsilon\leq1$,
  one can compute, with high probability and in time 
  \[
  \Otil\left(\left( m + n 2^{O\left(\sqrt{\log{n}\log\log{n}}\right)} \right)
  	\log^{3}\left(\frac{\kappa(\md)\cdot \kappa(\mm)}{\epsilon}\right)\right)
  \]
  a vector $\vec{x}'$ satisfying $\|\mm\vec{x}'-\vec{b}\|_{2}\leq\epsilon\norm{\vec{b}}_{2}$.
\end{theorem} %
Note that column-diagonally-dominant matrices include
Laplacians of directed graphs.
This bound follows from combining the reduction to solving
linear systems in Eulerian Laplacians stated in Theorem 42
of~\cite{cohen2016faster} with our main solver result.
The condition number of the undirected Laplacians that arise
can be bounded by $O(\kappa(\md) \cdot \kappa(\mm))$ by the preceding
Theorem 41 in~\cite{cohen2016faster}.
This condition number becomes a logarithmic overhead
by the condition number reductions in Appendix~\ref{sec:reduction},
which allows us to focus on solving $\poly(n)$ conditioned
Eulerian systems.

The result that we will focus on in this section is an algorithm
that given an Eulerian Laplacian
$\mlap = \md - \ma^\top \in \R^{n \times n}$ with $m$ non-zero entries
computes an $\epsilon$-approximate solution to $\mlap \vx = \vb$ in
time $\Otil((m + n \exp(O(\sqrt{\log \kappa  \log\log \kappa }))) \log(1/\epsilon))$ where $\kappa = \kappa(\mU_{\md^{-1/2} \mlap \md^{-1/2}})$. Note that $\exp(O(\sqrt{\log(\kappa)\log\log\kappa}))$ is a term that is $\kappa^{o(1)}$, i.e. it grows less than $O(\kappa^\epsilon)$ for any constant $\epsilon > 0$, whereas iterative methods for solving such systems typically have a dependence of $\kappa^{1/2}$ or higher in their running time.
An overview of the main components of this algorithm
is in Section~\ref{sec:solverOverview}.

\subsection{Preconditioned Richardson Iteration and Approximate Pseudoinverse}
 \label{sec:richardson}

The key iterative method that we use to build our solver is the preconditioned Richardson iteration. It can be thought of as a general-purpose tool that boosts the quality of a linear system solver by iteratively applying the solver on the residual.\footnote{One should note that being able to boost a solver using preconditioned Richardson relies on the fact that the solver is a linear operator, which is precisely the case in our algorithm.}
In this section we describe the preconditioned Richardson iteration, and based on it, we derive a measure of quality of a linear system solver in terms of how well it functions as an approximate pseudoinverse for the matrix involved in the system we want to solve. In Section~\ref{sec:approxInverse} we analyze our solver chain in terms of this notion.

The Richardson iteration refers to perhaps one of the simplest methods for solving a linear system $\mm \vx = \vb$: start with $\vx_0 = \vzero$, then repeatedly move in the direction of the residual, i.e. $\vx_{k + 1} := \vx_k + \eta (\vb - \mm \vx)$, for some step size $\eta$. The preconditioned Richardson iteration refers to applying the same method, with the aid of a matrix $\mz$ whose purpose is to improve the quality of the iterations by producing better approximations to the matrix inverse applied to the residual: $\vx_{k + 1} = \vx_k + \eta \mz (\vb - \mm \vx)$. Note that whenever $\mz = \mm^{\dag}$, preconditioned Richardson adds in every step an $\eta$ fraction of the true solution to our current iterate; therefore, in that case, we obtain the exact solution in one iteration by setting $\eta = 1$. Intuitively, the quality of the preconditioner $\mz$ dictates the size of the steps we are allowed to take, and therefore how long it takes to get close to optimum. This motivates the notion of approximation we introduce in this section.

The preconditioned Richardson iteration is very well studied and fundamental to numerical methods (see e.g. Section 13.2.1 of~\cite{Saad03:book}). However,
 we are not aware of an operator-based
analysis involving asymmetric matrices in the
different matrix norms required in our algorithm. Therefore, we provide the algorithm (see  Figure~\ref{fig:richardson}) and its short analysis in Lemma~\ref{lem:precond_richardson} below.  

\begin{figure}[ht]
	
	\begin{algbox}
		
		$\vec{x}=\textsc{PreconRichardson}(\MM,  \ZZ, \vec{b}, \eta, N)$
		
		\textbf{Input}: $n \times n$ matrix $\MM$,\par
 \makebox[1.22cm]{}		preconditioning linear operator $\ZZ$ (in the unpreconditioned case, $\ZZ = \mI$), \par
 \makebox[1.22cm]{}		right hand side vector $\vec{b} \in \im{\MM}$, step size $\eta$, iteration count $N$.
			\begin{enumerate}
			\item Initialize $\vec{x}_0 \leftarrow 0$.
			\item For $k = 0, \ldots, N - 1$
				\begin{enumerate}
						\item $\vec{x}_{k + 1} \leftarrow \vec{x}_{k} + \eta \ZZ \left(\vec{b}- \MM \vec{x}_{k} \right)$. \label{ln:richardsonStep}
				\end{enumerate}
			\item Return $\vec{x}_N$.
		\end{enumerate}
	\end{algbox}
	
	\caption{Pseudocode for the (preconditioned) Richardson Iteration}
	
	\label{fig:richardson}
	
\end{figure}

\begin{lemma}[Preconditioned Richardson]
\label{lem:precond_richardson}
Let $\vec{b} \in \R^{n}$ and $\MM, \ZZ,  \mU \in \R^{n \times n}$ such that $\mU$ is symmetric positive semidefinite, $\ker(\mU) \subseteq \ker(\MM)=\ker(\MM^\intercal) = \ker(\ZZ)=\ker(\ZZ^\intercal)$, and $b \in \im{\mm}$. Then $N \geq 0$ iterations of preconditioned Richardson with step size $\eta > 0$, results in a vector $\vx_N =\textsc{PreconRichardson}(\MM, \ZZ, \vec{b}, \eta, N)$ such that 
$$\normFull{\vx_N - \mm^{\dag} \vb}_{\mU} \leq \normFull{\mI_{\imFull{\MM}} -\eta \ZZ \MM}_{\mU \rightarrow \mU}^{N} 
\normFull{\mm^\dag \vb}_{\mU}\,{.}$$ 
Furthermore, preconditioned Richardson implements a linear operator, in the sense that $\vec{x}_N = \ZZ_{N} \vec{b}$, for some matrix $\ZZ_N$ only depending on $\ZZ$, $\MM$, $\eta$ and $N$.
\end{lemma}

\begin{proof}
	Let $\vec{x}^{*} \defeq \MM^{\dag} \vec{b}$. The iteration on Line~\ref{ln:richardsonStep}, together with the fact that $\vec{b}$ lives inside the image of $\MM$, implies
\begin{align*}
\vec{x}_{k+1} - \vec{x}^* 
&= ((\mI_{\imFull{\MM}} - \eta\mz \MM)\vec{x}_k + \eta \mz \vec{b}) - \vec{x}^* = \left((\mI_{\imFull{\MM}} - \eta\mz \MM)\vec{x}_k + \eta \mz \MM \vec{x}^*\right) - \vec{x}^* \\
&\hfill = (\mI_{\imFull{\MM}}-\eta \mz\MM)(\vec{x}_k -\vec{x}^*)\,{,}
\end{align*}
and therefore,
	 $$\normFull{\vec{x}_{k+1}-\vec{x}^*}_{\mU} = \normFull{(\mI_{\imFull{\MM}}-\eta\mz\MM)(\vec{x}_{k}-\vec{x}^*)}_{\mU} \leq \normFull{\mI_{\imFull{\MM}}-\eta\mz\MM}_{\mU\rightarrow\mU} \normFull{\vec{x}_{k}-\vec{x}^*}_{\mU}\,{.}$$
	 By induction, this shows that
	 $$\normFull{\vec{x}_N - \vec{x}^*}_{\mU} \leq \normFull{\mI_{\imFull{\MM}} - \eta \mz \MM}_{\mU\rightarrow\mU}^N \normFull{\vec{x}_0-\vec{x}^*}_{\mU}=\normFull{\mI_{\imFull{\MM}} - \eta \mz \MM}_{\mU\rightarrow\mU}^N \normFull{\vec{x}^*}_{\mU}\,{.}$$
Now, by writing the iteration as $\vec{x}_{k+1} = (\mI_{\imFull{\mm}}-\eta \mz\mm)\vec{x}_k + \eta\mz\vec{b}$, and expanding, we see by induction that $\vec{x}_N = \sum_{k=1}^{N} (\mI_{\imFull{\mm}}-\eta \mz\mm)^{k-1} \eta\mz\vec{b}$, and therefore $$\mz_N = \eta \sum_{k=1}^{N} (\mI_{\imFull{\mm}}-\eta \mz\mm)^{k-1} \mz\,{.}$$

\end{proof}

Lemma~\ref{lem:precond_richardson} shows that if $\eta \mz \mm$ is sufficiently close to the identity, then preconditioned Richardson converges quickly when solving a linear system. This highlights a precise way to quantify how good a matrix is as a preconditioner for the Richardson iteration.

\begin{defn}[Approximate Pseudoinverse]
\label{defn:approxInv}
Matrix $\ZZ$ is an \emph{$\epsilon$-approximate pseudoinverse of
	matrix $\MM$ with respect to a symmetric positive semidefinite matrix $\UU$}, if
$\ker(\UU) \subseteq \ker(\MM)=\ker(\MM^\intercal)=\ker(\ZZ)=\ker(\ZZ^\intercal)$, and 
\[
\normFull{ \II_{im(\MM)} - \ZZ \MM }_{\UU \rightarrow \UU} \leq \epsilon.~\footnote{Note that the ordering of $\ZZ$ and $\MM$ is crucial: this definition is not equivalent to $\norm{\II_{\im{\MM}} - \MM\ZZ}_{\UU \rightarrow \UU}$
being small.}
\]
\end{defn}

With this definition and Lemma~\ref{lem:precond_richardson}, we see that our problem of solving a linear system can be reduced to producing an approximate pseudoinverse that we can apply efficiently. This is the approach we take in the rest of the paper. In the remainder of this section we give two tools towards producing such pseudoinverses. We show how to use preconditioned Richardson to improve the quality of an approximate pseudoinverse (Lemma~\ref{lem:precon}), and how to produce one for a matrix whose symmetrization is well conditioned (Lemma~\ref{lem:solveWellConditioned}). 

\begin{lem}[Pseudoinverse Improvement]
\label{lem:precon}
If $\ZZ$ is an $\epsilon$-approximate pseudoinverse of $\MM$ with respect to $\mU$, for $\epsilon \in (0, 1)$, $\vec{b} \in \im{\MM}$, and $N \geq 0$, then $\textsc{PreconRichardson}(\MM, \ZZ, \vec{b}, 1, N)$ computes $\vx_N =\ZZ_N$, for some
matrix $\ZZ_N$ only depending on $\ZZ$, $\MM$ and $N$, such that $\ZZ_N$ is an $\epsilon^{N}$-approximate pseudoinverse of $\mm$ with respect to $\mU$. 
\end{lem} 

\begin{proof}
By Lemma~\ref{lem:precond_richardson} we know that $\vx_{N} = \mz_N \vb$, for some $\mz_N$ that only depends on $\mz$, $\mm$, and $N$; furthermore, we know that:
\[
\normFull{(\mI_{\imFull{\mm}} - \mz_N\mm) \mm^{\dag} \vec{b} }_{\mU\rightarrow\mU} = 
\normFull{\vec{x}_N - \mm^{\dag}\vec{b}}_{\mU} \leq \normFull{\mI_{\imFull{\mm}}-\mz\MM}_{\mU\rightarrow\mU}^N \normFull{\mm^{\dag}\vec{b}}_{\mU} \leq \epsilon^N \normFull{\mm^{\dag}\vec{b}}_{\mU}\,{,}
\] 
and that this holds for any vector $\vec{b} \in \imFull{\mm}$ (since $\mz_N$ does not depend on $\vec{b}$). Equivalently, this means that
\[
\normFull{\mI_{\imFull{\mm}}-\mz_N\mm}_{\mU\rightarrow\mU} \leq \epsilon^N\,{.}
\]

In order to complete the proof we need to show that $\ker(\mz) = \ker(\mz^{\top}) = \ker(\mz)$.

As we saw in the proof of Lemma~\ref{lem:precond_richardson}, over $\mI_{\im{\mm}}$, $\mz_N$ is a polynomial in $\mz$ and $\mm$ with no constant term. Also, since all the kernels and cokernels of $\mz$ and $\mm$ are identical by definition, $\ker{\mz_N} \supseteq \ker{\mz}$, and similarly $ \ker{\mz^{\top}}\supseteq\ker{\mz}$. 

Now suppose the first inclusion is strict, i.e. there exists $\vec{x} \perp \ker(\mz)$ such that $\mz_N\vec{x} = 0$. Then, $\norm{\mm^{\dag}\vec{x}}_{\mU} > 0$ since $\mm^{\dag} \vec{x} \perp \ker(\mU)$, as by definition, the kernel of $\mU$ is a subset of that of $\mz$. This implies that
$$\normFull{ (\mI_{\imFull{\mm}} - \mz_N \mm)\mm^{\dag} \vec{x}  }_{\mU}=\normFull{\mm^{\dag} \vec{x}}_{\mU}$$
This shows that $\norm{\mI_{\im{\mm}} - \mz_N \mm}_{\mU\rightarrow\mU} \geq 1$, which contradicts the fact that it is at most $\epsilon^N$.

Similarly, if there exists $\vec{x} \perp \ker(\mz^{\top})$ such that $\mz^{\top}\vec{x} = 0$, then we obtain 
$$\normFull{ (\mI_{\imFull{\mm}} - \mz_N^{\top} \mm^{\top})(\mm^{\top})^{ \dag} \vec{x}  }_{\mU}=\normFull{(\mm^{\top})^{\dag} \vec{x}}_{\mU}\,{,}$$
and thus $\normFull{\mI_{\imFull{\mm}}-\mz_N^{\top} \mm^{\top}}_{\mU\rightarrow\mU}  \geq 1$. Equivalently, this shows that $\norm{\mI_{\im{\mm}} - \mz_N \mm}_{\mU\rightarrow\mU} \geq 1$, which yields a contradiction. The fact that the norm is the same when taking transposes follows from writing it in terms of the $\ell_2$ norm, and using the fact that $\ker(\mU)$ is a subset of both the kernel of the matrix and that of its transpose.
\end{proof}

In addition, we show that preconditioned Richardson converges quickly whenever the matrix $\mm$ is well conditioned (as a matter of fact, for our purposes we only care about the case when the ratio between $\norm{\mm}$ and $\lambdanonzero(\mU_\mm)$ is at most a constant). The lemma below gives precise bounds on the number of iterations required to obtain a good approximate pseudoinverse with respect to $\mI$. Such an approximate pseudoinverse is the standard object that can be used as a preconditioner in order to obtain a small number of preconditioned Richardson iterations when measuring error with respect to the $\ell_2$ norm.

\begin{lemma}[Building a Pseudoinverse]
\label{lem:solveWellConditioned}
Let $\mm \in \R^{n \times n}$ such that $\mU_\mm$ is positive semidefinite, and $\ker(\mm) = \ker(\mm^\top)$. Let the vector $\vb \in \im{\mm}$, the step size $\eta \leq \lambdanonzero(\mU_{\mm}) / \norm{\mm}_2^2$, and the number of iterations $N$. Then
$\textsc{PreconRichardson}(\MM, \eta\mI_{\im{\mm}}, \vec{b}, 1, N)$ computes $\vx_{N} = \ZZ_{N} \vec{b}$, for some
matrix $\ZZ_N$ only depending on $\ZZ$, $\MM$ and $N$, such that $\ZZ_N$ is an $\exp(- N \eta \lambdanonzero(\mU_\mm)/2)$-approximate pseudoinverse of $\mm$ with respect to $\mI$. 
\end{lemma}

\begin{proof}
We begin by showing that $\eta\mI_{\im{\MM}}$ is a $(1-\eta\lambdanonzero(\mU_\mm))^{1/2}$-approximate pseudoinverse of $\mm$ with respect to $\mI$. First we notice that the kernel conditions for $\mI_{\im{\MM}}$ being an approximate pseudoinverse for $\MM$ are trivially satisfied. Second, we bound the matrix induced norm of $\mI_{\im{\mm}}-\eta\mI_{\im{\mm}} \cdot \mm = \mI_{\im{\mm}} - \eta\mm$:
\begin{align*}
\norm{\mI_{\im{\MM}} -\eta \MM}_{2}^2
 &= \max_{\vec{x} \in \im{\MM},\norm{\vec{x}}_{2} = 1} \vec{x}^{\top}
\left(\mI_{\im{\MM}}-\eta \MM\right)^{\top}
\left(\mI_{\im{\MM}}-\eta \MM \right) \vec{x} \\
&= \max_{\vec{x} \in \im{\MM},\normFull{\vec{x}}_{2} = 1}
\left(
\vec{x}^{\top} \mI_{\im{\MM}} \vec{x}
- \eta \vec{x}^{\top} \left(\MM + \MM^{\top} \right) \vec{x}
+ \eta^2 \vec{x}^{\top} \MM^{\top}\MM \vec{x}
\right) \\
&\leq 1 - 2 \eta \min_{\vec{x} \in \im{\MM},\normFull{\vec{x}}_{2} = 1}
\vec{x}^{\top} \mU_{\mm}  \vec{x}
+ \eta^2 \max_{\vec{x} \in \im{\MM}, \normFull{\vec{x}_{2}} = 1} \vec{x}^{\top} \MM^{\top} \MM \vec{x}\,{.} \\
&\leq 1 - 2 \eta \lambdanonzero(\mU_\mm) + \eta^2 \norm{\mm}_2^2
\leq 1 - \eta \lambdanonzero(\mU_\mm) ~.
\end{align*}
Consequently, by 
Lemma~\ref{lem:precon}
we have that $\mz_ N$ is a $(1-\eta \lambdanonzero(\mU_\mm))^{N/2}$-approximate pseudoinverse of $\mm$ with respect to $\mI$. The conclusion follows by using $(1-\eta\lambdanonzero(\mU_\mm))^{N/2} \leq \exp(-N\eta\lambdanonzero(\mU_\mm)/2)$.
\end{proof}

\subsection{Construction of
Square-Sparsification Chains}
\label{sec:construction}
Here we define the square-sparsification chain we use in our algorithm and show how to compute such a chain efficiently. In other words, we show how to create the sequence of matrices that through careful application yield an almost linear time algorithm for solving an Eulerian Laplacian system.

\begin{defn}
\label{defn:squarechain}[Square Sparsifier Chain] We call a sequence of matrices $\WWhat_0, \WWhat_1, \ldots \WWhat_{d} \in \R^{n \times n}$ a 
\emph{square-sparsifier chain of length $d$ with parameter $0 < \alpha< \frac{1}{2}$ and error $\epsilon \leq 1/2$} (or a \emph{$(d,\epsilon,\alpha)$-chain} for short) if under the definitions 
${\mLL}_i = \mI - \WWhat_i$ and 
$\WWhat^{(\alpha)}_i = \alpha \II + (1-\alpha)\WWhat_i$ for all $i$ the following hold
\begin{enumerate}
\item $\norm{\WWhat_i}_2 \leq 1$ for all $i$,
\item $\mI - \WWhat_{i}$ is an $\epsilon$-approximation of $\mI - ( \WWhat^{(\alpha)}_{i-1} )^2$ for all $i\geq 1$,
\item\label{item:ker} $\ker({\mLL}_i)=\ker({\mLL}_i^\intercal) = \ker({\mLL}_j) = \ker({\mLL}_j^\intercal)=\ker(\mU_{\mLL_i}) = \ker(\mU_{{\mLL}_j})$ for all $i,j$.
\end{enumerate}
\end{defn}

\newcommand{\buildchain}{\textsc{BuildChain}}

\begin{figure}[ht!]
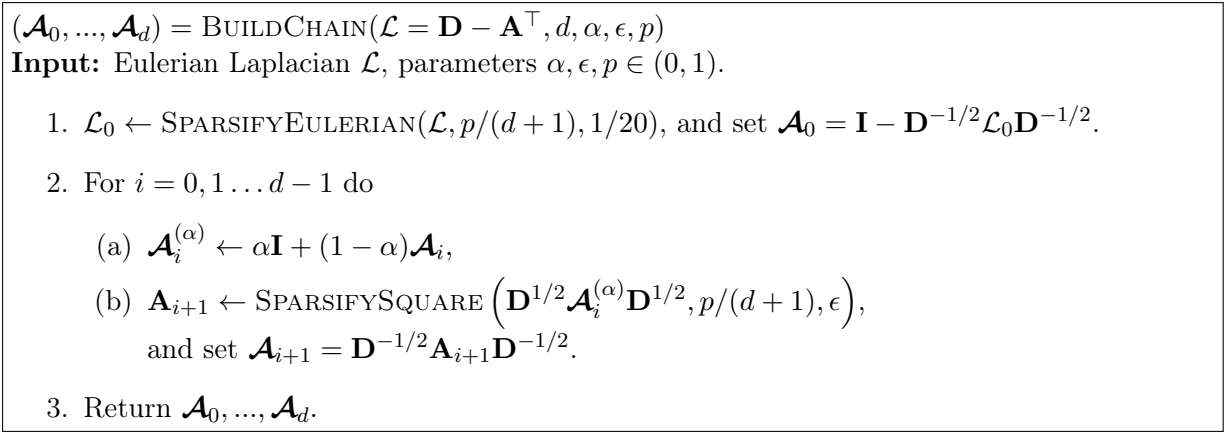

\begin{algbox}
$(\WWhat_0,...,\WWhat_{d})=\buildchain(\mlap = \md - \ma^{\top}, d, \alpha, \epsilon, p)$ 

\textbf{Input:} Eulerian Laplacian $\mlap$,
	parameters $\alpha, \epsilon, p \in (0,1)$.
\begin{enumerate}

\item 
\label{itm:sparsify}
$\mlap_0 \leftarrow \textsc{SparsifyEulerian}(\mlap, p/(d + 1), 1/20)$, and set $\WWhat_0 = \mI-\md^{-1/2} \mlap_0 \md^{-1/2}$.
\item For $i=0,1 \ldots d-1$ do
\begin{enumerate}
\item $\WWhat^{(\alpha)}_{i} \leftarrow \alpha \II + (1-\alpha)\WWhat_i,$ 
\item $\ma_{i + 1}  \leftarrow \sparsifysquare \left(\md^{1/2}  \WWhat^{(\alpha)}_{i} \md^{1/2}, p/(d + 1), \epsilon \right)$, \\ 
and set $\WWhat_{i+ 1} = \md^{-1/2} \ma_{i+1} \md^{-1/2}$.
\end{enumerate}
\item Return $\WWhat_0,...,\WWhat_{d}.$
\end{enumerate}
\end{algbox}
\caption{Algorithm for Constructing the Square-Sparsification Chain.}
\label{fig:buildChain}
\end{figure}

Pseudocode for the construction of the square-sparsifier chain is given in Figure~\ref{fig:buildChain}. In the remainder of this subsection we prove correctness of this construction. We first provide a helper lemma, Lemma~\ref{lem:normalized_lap_properties} and then analyze the algorithm in Lemma~\ref{lem:chain_construction}. In Section~\ref{sec:progresstools} we prove additional properties regarding solver chains such as how the $\mU_{\LL_i}$ multiplicatively approximate each other and how the smallest eigenvalue at the end of the chain must eventually rise to at least a constant. 

\begin{lem}\label{lem:normalized_lap_properties} If for $\ma \in \R^{n \times n}_{\geq 0}$ and $\md = \mdiag(\ma \vones)$ the matrix $\mlap=\md-\ma$ is an Eulerian Laplacian associated with a strongly connected graph then, $\norm{\md^{-1/2}\ma^{\top}\md^{-1/2}}_{2} \leq 1$ and $\ker(\mlap) = \ker(\mlap^\top) = \ker(\mU_{\mlap}) = \mathrm{span}(\md^{1/2} \vones)$.\end{lem}
\begin{proof}
Since $\mlap$ is Eulerian $\ma \vones = \ma^\top \vones = \md \vones$ and therefore we have  $\norm{\md^{-1}\ma^{\top}}_{\infty}=\norm{\ma^{\top}\md^{-1}}_{1}=1$. Consequently, by Lemma~\ref{lem:matrix_two_norm}
we have that
$\norm{\md^{-1/2}\ma^{\top}\md^{-1/2}}_{2}\leq\sqrt{\norm{\md^{-1}\ma^{\top}}_{\infty}\cdot\norm{\ma^{\top}\md^{-1}}_{1}} \leq 1$. The characterization of the kernels follows from Lemma~\ref{lem:stationary-equivalence} and that $\mlap \vones = \mlap^\top \vones = \mU_{\mlap} \vones = \vzero$.
\end{proof}

\begin{lem}[Chain Construction]\label{lem:chain_construction}
Let $\mlap = \md - \ma^\top$ be an Eulerian Laplacian that is associated with a strongly connected graph, let $\alpha, \epsilon, p \in (0,1)$, and let $d \geq 1$. Then in $\otilde(\nnz(\mlap) +n\epsilon^{-2} d)$ time the routine $\buildchain(\mlap, d, \alpha, \epsilon, p)$ produces $\WWhat_0, ..., \WWhat_d \in \R^{n \times n}$ that with probability $1 - p$
\begin{enumerate}
	\item $\WWhat_0, ..., \WWhat_d$ is a $(d, \alpha, \epsilon)$-chain,
	\item $\nnz(\WWhat_i) = \otilde(n \epsilon^{-2})$ for all $i$, and
	\item $\mI - \WWhat_0$ is a $(1/20)$-approximation of $\md^{-1/2} \mlap \md^{-1/2}$.
\end{enumerate}
\end{lem}

\begin{proof}
By Theorem~\ref{thm:order_n_sparsifier}, the call to \textsc{SparsifyEulerian} computes in $\otilde(m+n)$ time $\mlap_0$ that with probability $1 - p/(d + 1)$ is a $(1/20)$-sparsifier of $\mlap$ with the same diagonal as $\mlap$. Consequently, for $\ma_0 = \md^{1/2} \WWhat_0 \md^{1/2}$ we  have 
 $\mlap_0 = \md - \ma_0$. By Lemma~\ref{cor:strong_basis_change} this implies that $\mI - \WWhat_0$ is a $(1/20)$-approximation of $\md^{-1/2} \mlap \md^{-1/2}$. Furthermore, Lemma~\ref{lem:normalized_lap_properties} then implies that $\norm{\WWhat_0}_2 = \norm{\md^{-1/2} \ma_0 \md^{-1/2}}_2 \leq 1$, and $\ker(\mI - \WWhat_0) = \ker((\mI - \WWhat_0)^\top) = \ker(\mU_{\mI - \WWhat_0}) = \mathrm{span}(\md^{1/2} \vones)$. Thus, $\WWhat_0$ has all the desired properties.

Now suppose the desired properties hold for $\WWhat_0, ..., \WWhat_k$ with probability $1 - p (k + 1) / (d + 1)$, for some $k \in [0, d - 1]$, and that for all $i \in [k]$ we have $\ma_i \in \R^{n \times n}_{\geq 0}$ such that $\md - \ma_i$ is an Eulerian Laplacian associated with a strongly connected graph. Under this assumption, clearly $\md^{-1/2} \WWhat_k \md^{-1/2} = \alpha \md + (1- \alpha)\ma_k$ has both row and column sums equal to $\md\vones$. By Theorem~\ref{thm:sparsify_square}, the call to \textsc{SparsifySquare} computes in $\otilde(m+n\epsilon^{-2})$ time a matrix $\ma_{k + 1} \in \R^{n \times n}_{\geq 0}$ such that, with probability $1 - p(k + 2)/(d + 1)$,  $\md - \ma_{k + 1}$ is an  $\epsilon$-sparsifier for
$\md - \ma_k^{(\alpha)} \md^{-1} \ma_k^{(\alpha)}$, where $\ma_k^{(\alpha)} = \alpha \mI + (1 - \alpha) \ma_k$. 
Again, using Lemma~\ref{cor:strong_basis_change},
we see that $\mI - \WWhat_{k + 1}$ is an $\epsilon$-approximation of $\mI - (\WWhat_{k}^{(\alpha)})^2$. Furthermore, since $\md - \ma_k^{(\alpha)} \md^{-1} \ma_k^{(\alpha)}$ contains $\md - \ma_k$ as a subgraph,  this Eulerian Laplacian is strongly connected and therefore so is $\md - \ma_{k + 1}$. Consequently, by Lemma~\ref{lem:normalized_lap_properties} we have that $\norm{\WWhat_{k + 1}}_2 = \norm{\md^{-1/2} \ma_{k + 1} \md^{-1/2}}_2 \leq 1$, and $\ker(\mI - \WWhat_{k + 1}) = \ker((\mI - \WWhat_{k + 1})^\top) = \ker(\mU_{\mI - \WWhat_{k + 1}}) = \mathrm{span}(\md^{1/2} \vones)$. Therefore, by induction all the desired properties hold.
\end{proof}

\subsection{Properties of the Square-Sparsification Chain}
\label{sec:progresstools}
Here we prove several properties of the square-sparsification chain which we use in the analysis of our solver algorithm. The main result of this subsection is to prove the following:

\begin{lem} \label{lem:buildChain}[Chain Properties]
For length $d\geq 1$, parameter $\alpha = 1/4$, and error $\epsilon\in(0,1/2)$, and $(d,\alpha,\epsilon)$-chain $\WWhat_0, \WWhat_1, \ldots \WWhat_{d}$ the following properties hold:
	\begin{enumerate}
		\item $\kappa(\mI-\mU_{\WWhat_{i}}, \mI-\mU_{\WWhat_{i-1}}) \leq 21$, for all $i \in [d]$\,{, and}\label{item:prop1}
				\item $\lambdanonzero(\mI- \mU_{\WWhat_d}  ) \geq \min\{1/4, \lambdanonzero(\mI - \mU_{\WWhat_0})\cdot ((1-\epsilon)1.25)^d \}$\,{.}\label{item:prop2}
	\end{enumerate}
\end{lem}
The first property shows that the matrix $\mI - \mU_{\WWhat_i}$ changes only within a constant factor for a constant change in $i$, while the second implies that the smallest non-zero eigenvalue improves geometrically with squaring, up to some fixed value that depends on $\alpha$.

The proof of Lemma~\ref{lem:buildChain} relies on several components. First, we use a lemma which shows how squaring changes the associated symmetric matrix (see Lemma~\ref{lem:square_condition_number} in Appendix~\ref{sec:linear_algebra}). Combining with the sparsification guarantees from Lemma~\ref{lem:asym_strong_implies_undir}, we derive the first property. Second, we use a bound on how the smallest non-zero eigenvalue improves after squaring (see Lemma~\ref{lem:kappa-improvement} from Appendix~\ref{sec:linear_algebra}), in order to show that this is also the case with the matrices in our chain.

\begin{proof}[Proof of Lemma~\ref{lem:buildChain}]

By definition, $\mI - \WWhat_i$ is an $\epsilon$-approximation of $\mI-( \WWhat^{(\alpha)}_{i-1} )^2$, for all $i\geq 1$. Therefore by Lemma~\ref{lem:asym_strong_implies_undir} we know that
$$(1-\epsilon)\left(\mI- \mU_{(\WWhat^{(\alpha)}_{i-1})^2}\right)\preceq \mI-\mU_{\WWhat_i} \preceq (1+\epsilon)\left(\mI - \mU_{(\WWhat^{(\alpha)}_{i-1})^2}\right)\,{.}$$ 
Applying Lemma~\ref{lem:square_condition_number}, we obtain
$$  2\alpha \left(\mI - \mU_{\WWhat^{(\alpha)}_{i-1}}\right)  \preceq \mI - \mU_{\left(\WWhat^{(\alpha)}_{i-1}\right)^2}  \preceq (4-2\alpha) \left(\mI - \mU_{\WWhat^{(\alpha)}_{i-1}}\right)\,{.}$$
Combining with the sparsification guarantee from above, we obtain:
$$(1-\epsilon) 2\alpha \left(\mI - \mU_{\WWhat^{(\alpha)}_{i-1}}\right)\preceq  \mI-\mU_{\WWhat_i} \preceq (1+\epsilon) (4-2\alpha) \left(\mI - \mU_{\WWhat^{(\alpha)}_{i-1}}\right)\,{.}$$
Finally, writing $\mI-\mU_{\WWhat_{i-1}^{(\alpha)}} = \mI - (\alpha\mI + (1-\alpha)\mU_{\WWhat_{i-1}}) = (1-\alpha)\mU_{\WWhat_{i-1}}$, this gives:
$$(1-\epsilon) 2\alpha(1-\alpha) \left(\mI - \mU_{\WWhat_{i-1}}\right)\preceq  \mI-\mU_{\WWhat_i} \preceq (1+\epsilon) (4-2\alpha)(1-\alpha) \left(\mI - \mU_{\WWhat_{i-1}}\right)\,{,}$$
which shows that $$\kappa(\mI-\mU_{\WWhat_i}, \mI-\mU_{\WWhat_{i-1}}) \leq \frac{ (1+\epsilon)(4-2\alpha) }{ (1-\epsilon)2\alpha } \leq 21\,{.}$$
For the second part, we see that Lemma~\ref{lem:kappa-improvement} yields:
$$\lambdanonzero(\mI - \mU_{ (\WWhat^{(\alpha)})_{i-1}^2  }) \geq \min\{\alpha, (1+\alpha)\lambdanonzero(\mI - \mU_{\WWhat_{i-1}})\}\,{.}$$ 
Combining with the first inequality from the sparsification guarantee, this gives
$$\lambdanonzero(\mI-\mU_{\WWhat_i}) \geq (1-\epsilon)\cdot \min\{\alpha, (1+\alpha)\lambdanonzero(\mI - \mU_{\WWhat_{i-1}})\}\,{.}$$
Applying this inequality $d$ times yields the second part of the result.

\end{proof}

\subsection{Pseudoinverse Properties}
\label{sec:approxInverse}
Here we show that $\epsilon$-approximations in the square
sparsifier chain imply useful properties for building approximate pseudoinverses. We provide key lemmas to show that an approximate pseudoinverse for $\mI-\WWhat_j$ can be transformed into one for
$\mI-\WWhat_i$, where $i < j$. More precisely, this is based on Equation~\ref{eqn:key}, and can be seen by substituting
the condition from Definition~\ref{defn:squarechain}
that $\II - \WWhat_{j}$ is an $\epsilon$-approximation of
$\II - ( \WWhat^{(\alpha)}_{j-1})^2$ into the identity
\[
\left( \II - \WWhat_{j- 1} \right)^{\dag}
= \left(1 - \alpha\right) \left( \II - \WWhat^{(\alpha)}_{j-1} \right)^{\dag}
=  \left(1 - \alpha\right)  \left( \II - \left( \WWhat^{(\alpha)}_{j-1} \right)^2 \right)^{\dag}
\left( \II + \WWhat^{(\alpha)}_{j-1} \right)\,{.}
\]

This method of producing solvers accumulates error very quickly. However, as discussed in Section~\ref{sec:richardson}, we can reduce the error accumulation using preconditioned Richardson iterations. Consequently, the main goal of the remainder of this section is to formally bound this accumulated error so that we can ensure Richardson will quickly produce a high quality approximate pseudoinverse.  

We prove this result in several steps. First we provide several properties regarding approximate pseudoinverses, showing that they are well behaved under composition, and that they exhibit desirable properties such as approximate triangle inequality, and being preserved under right multiplication. Then using these lemmas we prove the main result of this section, Lemma~\ref{lem:chainproperty}.

Given that ultimately we build our pseudoinverse recursively, in our presentation $\ZZ$ will often take the role of a solver being used as a preconditioner. This is consistent with the way we use it, since all solvers we produce are linear operators.

The following lemma bounds the quality of a preconditioner obtained via composition.
\begin{lem}[Triangle Inequality
of Approximate Pseudoinverses]
\label{lem:approxInvBasic}
If matrix $\ZZ$ is an $\epsilon$-approximate pseudoinverse of $\MM$ with respect to $\UU$, and $\MMtil^{\dag}$ is an
$\epsilon'$-approximate pseudoinverse of  $\ZZ^{\dag}$ with respect to $\UU$,
and has the same left and right kernels as  $\MM$ and $\ZZ$,
then $\MMtil^{\dag}$ is an $(\epsilon + \epsilon' + \epsilon\epsilon')$-approximate
pseudoinverse of $\MM$ with respect to $\UU$.
\end{lem}

\begin{proof}

Applying triangle inequality for the $\mU\rightarrow \mU$ norm, we see that, for all $x \in \R^n$,
\[
\normFull{ \left(\II_{\imFull{\MM}} - \MMtil^{\dag} \MM \right) \vec{x}}_{\UU }\\
\leq \normFull{ \left(\II_{\imFull{\MM}}
		- \MMtil^{\dag} \ZZ^{\dag} \right) \vec{x}}_{\UU }
+\normFull{\left(\MMtil^{\dag} \ZZ^{\dag} - \MMtil^{\dag} \MM\right) \vec{x}}_{\UU }\,{.}
\]
The first term is upper bounded by 
$\epsilon' \norm{\vec{x}}_{\UU}$ by the condition given in the statement.
The second term can be rewritten as:
\[
\normFull{\left(\MMtil^{\dag} \ZZ^{\dag} - \MMtil^{\dag} \MM\right) \vec{x}}_{\UU }
= \normFull{\MMtil^{\dag} \ZZ^{\dag} 
\left(\II_{\imFull{\MM}} - \ZZ \MM \right) \vec{x}}_{\UU}
\leq \normFull{\MMtil^{\dag} \ZZ^{\dag}}_{\UU \rightarrow \UU}
	\cdot \normFull{ \left(\II_{\imFull{\MM}} - \ZZ \MM  \right) \vec{x}}_{\UU}\,{.}
\]
Next, we bound the two components of the product. For the first one, using triangle inequality along with the condition given in the statement, we obtain 
$$\normFull{\MMtil^{\dag}\ZZ^{\dag}}_{\UU\rightarrow \UU} \leq \normFull{\II_{\imFull{\mm}}}_{\UU\rightarrow\UU}  + \normFull{\II_{\im{\MM}} - \MMtil^{\dag}\ZZ^{\dag}}_{\UU\rightarrow\UU} \leq 1+\epsilon'\,{.}$$
The second term is by definition bounded by $\epsilon\norm{\vec{x}}_{\UU}$.  Combining these bounds, we obtain
\[
\normFull{ \left( \mI_{\imFull{\MM}} -\MMtil^{\dag}\MM \right)\vec{x} }_{\UU}\leq \epsilon' \normFull{\vec{x}}_{\UU} + (1+\epsilon')\epsilon \normFull{\vec{x}}_{\UU} = (\epsilon + \epsilon' + \epsilon \epsilon')\normFull{\vec{x}}_{\UU}\,{.}
\]
\end{proof}

The following lemma is used to show that $\epsilon$-approximations obtained via the sparsification routines from Section~\ref{sec:sparsification} also yield matrices whose pseudoinverses are good preconditioners for the original.
\begin{lem}\label{thm:spectral-to-approxInv}
Let $\MM$ be any matrix with $\mU_{\MM}$ positive semidefinite, such that  $\ker(\MM)=\ker(\MM^\top)=\ker(\mU_{\MM})$.
Suppose that matrix $\MMtil$  $\epsilon$-approximates $\MM$, for some $\epsilon \leq 1/2$.
Then $\MMtil^\dagger$ is an $2\epsilon$-approximate pseudoinverse
for $\MM$ with respect to $\mU_{\MM}$.
Furthermore, $\ker(\mU_\MM) = \ker(\mU_{\MMtil})$.%
\end{lem}

\begin{proof}
First, we prove that $\ker(\MMtil)=\ker(\MMtil^\top)=\ker(\mU_{\MMtil})$ and that these kernels are the same as those of $\MM$, $\MM^\top$ and $\mU_{\MM}$.

We already know that $\ker(\MM)= \ker(\MM^\top) = \ker(\mU_{\MM})$ and it is not hard to prove that $\ker(\MMtil), \ker(\MMtil^{\top}) \subseteq \ker(\mU_{\MMtil})$. Thus, it suffices to prove that $\ker(\mU_\MM) \subseteq \ker(\MMtil),\ker(\MMtil^\top)$ and $\ker(\mU_{\MMtil}) \subseteq \ker(\mU_\MM)$.

To prove $\ker(\mU_\MM) \subseteq \ker(\MMtil),\ker(\MMtil^\top)$, consider any vector $x$ in the kernel of $\MM$. Then $x$ is also in the kernel of $\mU_{\MM}$. However, by the definition of strong approximation, this means that $x$ is in the kernel of $\MM - \MMtil$ and $\MM^\top - \MMtil^\top$. Since $x$ is in the left and right kernels of $\MM$, we have $0=(\MM - \MMtil)x= \MMtil x$ and similarly obtain $0=(\MM^\top - \MMtil^\top)x= \MMtil x$. Thus, the left and right kernels of $\MMtil$ are supersets of $\ker(\MM)$. 

For the reverse direction, note that Lemma~\ref{lem:asym_strong_implies_undir} implies that $\mU_\MM$ and $\mU_{\MMtil}$ approximate each other in the standard positive semidefinite sense for undirected matrices, and thus have the same kernel. Thus, the kernels of $\MM$ meet the requirements for being approximate pseudoinverses.

Now we can show the requisite inequality for
$\MMtil^{\dag}$ being an approximate pseudoinverse of $\MM$ with respect to $\mU_\MM$. First we show that $\MM^{\dag}$ is an $\epsilon$-approximate pseudoinverse of $\MMtil$ with respect to $\mU_\MM$.  We can apply a lemma upper bounding matrix norms whose proof can be found in the appendix (Lemma~\ref{lem:sym-hsm}) in order to obtain:
\[
\normFull{\mI_{\im{\MM}}-\MM^{\dag}\MMtil }_{\mU_\MM\rightarrow \mU_\MM}  
\leq \normFull{\mU_\MM^{\dag/2}\MM \left(\mI_{\im{\MM}}-\MM^{\dag}\MMtil\right) \mU_\MM^{\dag/2}}_2 
= \normFull{\mU_\MM^{\dag/2}\left( \MM-\MMtil\right) \mU_\MM^{\dag/2}}_2  \leq \epsilon\,{.}
\]
Next we prove that this implies the desired conclusion by writing
\begin{align*}
&\normFull{\mI_{\im{\MM}}-\MM^{\dag}\MMtil}_{\mU_\MM\rightarrow \mU_\MM}
= \normFull{(\MMtil^{\dag}\MM-\mI_{\im{\MM}})\MM^{\dag}\MMtil}_{\mU_\MM\rightarrow \mU_\MM} \\
&\geq \normFull{\MMtil^{\dag}\MM-\mI_{\im{\MM}}}_{\mU_\MM\rightarrow \mU_\MM} 
-\normFull{(\MMtil^{\dag}\MM-\mI_{\im{\MM}})(\MM^{\dag}\MMtil-\mI_{\im{\MM}})}_{\mU_\MM\rightarrow \mU_\MM} \\
&\geq \normFull{\mI_{\im{\MM}}-\MMtil^{\dag}\MM}_{\mU_\MM\rightarrow \mU_\MM} \left(1 - \normFull{\mI_{\im{\MM}}-\MM^{\dag}\MMtil}_{\mU_\MM\rightarrow \mU_\MM} \right)\,{.}
\end{align*}
For the first inequality we used triangle inequality, for the second one we used the fact that the norm of a product is upper bounded by the product of norms, which follows immediately from applying the definition of our matrix norm.
By rearranging terms, and using the fact that $\MM^{\dag}$ is an $\epsilon$-approximate pseudoinverse of $\MMtil$ with respect to $\mU_\MM$ we obtain:
\[
\normFull{\mI_{\im{\MM}}-\MMtil^{\dag}\MM}_{\mU_\MM\rightarrow\mU_\MM} \leq \frac{\epsilon}{1-\epsilon} \leq 2\epsilon\,{.}
\]
\end{proof}

Next, recall that the goal of the square-sparsifier chain is to allow
an approximate inverse for $\mI-\WWhat_j$ to be used as preconditioner for $\mI-\WWhat_i$, with $i < j$.  We show that our notion of approximate pseudoinverse
is (approximately) preserved under right-multiplications, and when
changing the reference matrix $\mU$.

\begin{lem}[Composition of Approximate Pseudoinverses]
\label{lem:approxInv-composition}
Let $\ZZ, \MM, \UU\in \R^{n \times n}$ be matrices such that
$\UU$ is symmetric positive semidefinite, and 
$\ker(\ZZ) = \ker(\ZZ^{\top}) = \ker(\MM) = \ker(\MM^{\top})
\supseteq \ker(\UU)$.
 Then the following hold.
\begin{enumerate}
\item \label{part:approxInv-multiply}
(Preserved under right multiplication)
Let $\mc \in \R^{n \times n}$ such that
both $\mc$ and $\mc^\top$ are invariant on $\ker(\MM)$, in the sense that $x\in \ker(\MM)$ if and only if $\mc \in \ker(\MM)$, and similarly for $\mc^\perp$.
Then $\ZZ$ is an $\epsilon$-approximate pseudoinverse for $\mc \MM$ with respect to $\UU$ if and only if
$\ZZ \mc$ is an $\epsilon$-approximate pseudoinverse for
  $\mm$ with respect to $\mU$.
\item \label{part:weak-norm-change}
(Approximately preserved under norm change)
If $\ZZ$ is an $\epsilon$-approximate pseudoinverse for $\MM$
with respect to $\UU$, then for any symmetric positive semidefinite matrix $\UUtil$, such that $\ker(\UUtil) = \ker(\UU)$,
$\ZZ$ is an $(\epsilon \cdot \sqrt{\kappa(\UUtil,\mU)})$-approximate pseudoinverse
of $\MM$ with respect to $\widetilde{\UU}$.
\end{enumerate}
\end{lem}

\begin{proof}
For preservation under right multiplication (claim~\ref{part:approxInv-multiply}), we immediately see that by associativity: 
\[
\normFull{\II_{\imFull{\MM}} - \left( \ZZ \mc \right) \MM }_{\mU \rightarrow \mU}\\
= \normFull{\II_{\imFull{\MM}} - \ZZ (\mc \MM) }_{\mU \rightarrow \mU}
\,.
\]
What we have left is to verify that kernel conditions are satisfied. The assumptions on $\mc$, together with the fact that all the left and right kernels of $\ZZ$ and $\MM$ coincide, ensure that $\ker(\ZZ\mc) = \ker(\MM)$, $\ker(\mc^\top \ZZ^\top) = \ker(\ZZ^\top)$, $\ker(\mc\MM) = \ker(\MM)$, $\ker(\MM^\top \mc^\top) = \ker(\MM^\top)$.
Therefore the matrices satisfy the kernel requirements for being approximate pseudoinverses.

For approximate preservation under change of norms (claim~\ref{part:weak-norm-change}),
the bound on $\kappa(\UUtil, \UU)$ means that there exist
$\alpha$ and $\beta$ such that
$
\alpha \mU \preceq \UUtil \preceq \beta \mU
$ 
and $\beta/\alpha \leq \kappa(\UUtil, \UU)$.
Using this, we obtain:
\begin{align*}
\normFull{ \II_{\imFull{\MM}} - \ZZ \MM }_{\UUtil \rightarrow \UUtil}
&= \max_{\vec{x}: \UUtil\vec{x}\neq 0}
	\frac{\normFull{ \left(\II_{\imFull{\MM}}- \ZZ \MM \right)\vec{x} }_{\UUtil}}{\normFull{\vec{x}}_{\UUtil}} 
\leq \max_{\vec{x}: \UU\vec{x} \neq 0}
	\frac{\beta \normFull{  \left(\II_{\im{\MM}} - \ZZ \MM\right)\vec{x} }_{\UU}}
		{\alpha \norm{\vec{x}}_{\UU}} \\
&\leq \sqrt{\kappa\left(\UUtil, \UU\right)}
	\cdot \normFull{\II_{\im{\MM}} - \ZZ \MM}_{\UU \rightarrow \UU} 
\leq \epsilon \cdot \sqrt{\kappa\left(\UUtil, \UU\right)}.
\end{align*}
\end{proof}

The preservation under right-multiplications
combined with Equation~\ref{eqn:key}
suggests that if we have a linear operator $\ZZ$
that is an approximate pseudoinverse for $\mI-\WWhat_{j}$,
we can right-multiply it by $(1 - \alpha) (\mI + \WWhat_{j-1}^{(\alpha)})$
to form an operator that is an approximate pseudoinverse for $\mI-\WWhat_{j-1}$.
This process can then be repeated down the chain,
but will lead to an accumulation of error. The following lemma bounds the amount of error accumulated after repeating this process $\Delta$ times.

\begin{lem}
\label{lem:chainproperty}

Let the sequence $\WWhat_0, \WWhat_1, \ldots \WWhat_{d}$ be a $(d,\epsilon,\alpha)$-chain as specified in Definition~\ref{defn:squarechain}, with $\epsilon \leq 1/2$ and $\alpha = 1/4$. Using the notation from Definition~\ref{defn:squarechain}, consider the matrix
\[
 \ZZbar_{i, i+\Delta}
	= \left(1 - \alpha\right)^\Delta (\mI-\WWhat_{i+\Delta})^{\dagger} \left( \II + \WWhat_{i+\Delta - 1}^{(\alpha)} \right ) \cdot \cdot \cdot \left( \II +  \WWhat_{i}^{(\alpha)} \right)\,{,}
\]
for any $i,\Delta \geq 0$. Then $\ZZbar_{i,i+\Delta}$
is an $(\exp(5\Delta) \cdot \epsilon)$-approximate pseudoinverse of $\mI-\WWhat_{i}$ with respect to $\mI-\UU_{\WWhat_i}$.
\end{lem}

\begin{proof}
We will prove this by induction on $\Delta$.
The base case $\Delta = 0$ follows immediately, since $\ZZbar_{i,i} = (\mI-\WWhat_i)^\dag$, which is a $0$-approximate pseudoinverse of $\mI-\WWhat_i$ with respect to $\mI-\mU_{\WWhat_i}$.

For the induction step, let us assume
that the claim is true for $\Delta-1$. Therefore, the matrix
\[
\ZZbar_{i + 1, i+\Delta}
= \left(1 - \alpha\right)^{\Delta - 1} (\mI-\WWhat_{i+\Delta})^{\dagger}
\left ( \II + \WWhat_{i+\Delta - 1}^{(\alpha)} \right ) \cdot \cdot \cdot \left ( \II +  \WWhat_{i + 1}^{(\alpha)} \right)
\]
is an $(\exp(5(\Delta - 1))\cdot \epsilon)$-approximate
pseudoinverse of $ \mI-\WWhat_{i + 1}$ with respect to $\mI-\UU_{\WWhat_{i + 1}}$.

From Lemma~\ref{lem:buildChain} we see that for our choice of $\alpha=1/4$, we have  $\kappa(\mI-\UU_{\WWhat_{i + 1}}, \mI-\UU_{\WWhat_i}) \leq 21$. Since the matrices $\mI-\mU_{\WWhat_i}$ and $\mI-\mU_{\WWhat_{i+1}}$ have the same kernel, we can use the bound on their relative condition number with Lemma~\ref{part:approxInv-multiply}, part~\ref{part:weak-norm-change} to obtain that
 $\ZZbar_{i + 1, i+\Delta}$ is also a
$(\sqrt{21} \exp(5(\Delta - 1))\cdot \epsilon)$-approximate
pseudoinverse of $\mI- \WWhat_{i + 1}$ with respect to $\mI-\UU_{\WWhat_i}$.

By definition, we have that $\mI-\WWhat_{i+1}$ is an $\epsilon$-approximation of $\mI-(\WWhat_i^{(\alpha)})^2$. Therefore, by Lemma~\ref{thm:spectral-to-approxInv}, we know that $(\mI-\WWhat_{i+1})^\dag$ is a $2\epsilon$-approximate pseudoinverse of $\mI-(\WWhat_i^{(\alpha)})^2$ with respect to $\mI-\mU_{(\WWhat_i^{(\alpha)})^2}$. In order to change norms, we use Lemma~\ref{lem:square_condition_number}, which gives us that 
\[
\kappa\left(\mI-\mU_{\left(\WWhat_i^{(\alpha)}\right)^2}, (1-\alpha)\left(\mI-\mU_{\WWhat_i}\right)\right) 
= \kappa\left(\mI-\mU_{\left(\WWhat_i^{(\alpha)}\right)^2}, \mI-\mU_{\WWhat_i^{(\alpha)}}\right) \leq \frac{4-2\alpha}{2\alpha}\,{,}
\]
and therefore
\[
\kappa\left(\mI-\mU_{\left(\WWhat_i^{(\alpha)}\right)^2}, \mI-\mU_{\WWhat_i}\right) \leq \frac{4-2\alpha}{2\alpha(1-\alpha)} \leq \frac{28}{3}\,{.}
\]
Therefore, using Lemma~\ref{part:weak-norm-change}, and since $\sqrt{28/3}\cdot 2 \leq 7$, we obtain that $(\mI-\WWhat_{i+1})^\dag$ is a $7\epsilon$-approximate pseudoinverse of $\mI-(\WWhat_i^{(\alpha)})^2$ with respect to $\mI-\mU_{\WWhat_i}$.
Combining these two results via the triangle inequality for approximate pseudoinverses (Lemma~\ref{lem:approxInvBasic}), we obtain that
$\ZZbar_{i+1,i+\Delta}$ is an $\epsilon'$-approximate pseudoinverse of $\mI-(\WWhat_i^{(\alpha)})^2$ with respect to $\mI-\mU_{\WWhat_i}$, where
\[
\epsilon' = \sqrt{21}\exp(5(\Delta-1)) \epsilon + 7\epsilon + \sqrt{21}\exp(5(\Delta-1)) \epsilon \cdot 7\epsilon \leq 50 \exp(5(\Delta-1))\epsilon\,{.}
\]
Equivalently, by writing $\mI-(\WWhat_i^{(\alpha)})^2 = (\mI+\WWhat_i^{(\alpha)})(\mI-\WWhat_i^{(\alpha)})=(1-\alpha)(\mI+\WWhat_i^{(\alpha)})\cdot(\mI-\WWhat_i)$, and applying the composition under multiplication property from Lemma~\ref{lem:approxInv-composition} 
Part~\ref{part:approxInv-multiply}, we obtain that 
$\ZZbar_{i+1,i+\Delta} \cdot (1-\alpha)(\mI+\WWhat_i^{(\alpha)})$ is an $\epsilon'$-approximate pseudoinverse of $\mI-\WWhat_i$ with respect to $\mI-\mU_{\WWhat_i}$.
Note that in order to correctly apply the lemma, we require the kernel requirement for $(1-\alpha)(\mI+\WWhat_i^{(\alpha)})$ to be satisfied, but this follows easily since all left and right kernels of the other matrices involved are identical to $\ker(\mI-\WWhat_i^{(\alpha)})$.

Finally, since $\ZZbar_{i,i+\Delta} = \ZZbar_{i+1,i+\Delta} \cdot (1-\alpha)(\mI+\WWhat_i^{(\alpha)})$, this is equivalent to saying that $\ZZbar_{i,i+\Delta}$ is an approximate pseudoinverse of $\mI-\WWhat_i$ with respect to $\mI-\mU_{\WWhat_i}$, with error bounded by 
\[
50\exp(5(\Delta - 1))\epsilon \leq \exp(5\Delta)\epsilon\,{.}
\]
\end{proof}

Note that the amount of error accumulated through this process is significantly greater than the sum of $\epsilon$'s across
the different levels of the chain.
This is because we are measuring the quality of 
the approximate inverse with respect to a matrix that may change by a constant factor at each level, rather than with respect to a fixed one.

If we only invoke Lemma~\ref{lem:chainproperty}
for the first and last matrices of the chain
($i = 0$, $j = d$), it would give an error of
$\exp(O(d))\epsilon = \poly(\kappa(\mlap)) \epsilon$,
necessitating a sparsifier accuracy that is more or
less keeping everything dense.
Instead, in our algorithms we will only invoke the
above result for $j \approx i + \sqrt{d}$.
Between such steps, we will remove the accumulated
error using the preconditioned Richardson iteration, which was described in Section~\ref{sec:richardson}.

\newcommand{\echain}{\mathrm{chain}}
\newcommand{\weps}{\widehat{\epsilon}}

\subsection{The Recursive Solver}
\label{sec:recursivesolver}

Here we combine the pseudoinverse properties of the solver chain proved in Lemma~\ref{lem:chainproperty} with the preconditioned
Richardson iteration from Lemma~\ref{lem:precon} to obtain
an almost-linear time solver for Eulerian Laplacians.
The resulting algorithm makes recursive calls on the
square-sparsifier chain.
These recursive calls can be viewed as phases.
For some moderate value of $\Delta$ which we set to
$\sqrt{\log\kappa}$, where $\kappa$ is the condition number of $\mU_{\md^{-1/2} \mlap \md^{-1/2}}$,
we utilize Lemma~\ref{lem:chainproperty} to turn an approximate
pseudoinverse of $\LL_{i + \Delta}$ into an approximate
pseudoinverse for $\LL_{i}$ with larger error.
The error accumulated in this process is then removed
via preconditioned Richardson iteration.
This iteration leads to recursive calls to $\LL_{i + \Delta}$.

The resulting algorithm is a linear operator:
its output can be viewed as multiplying the input
by a fixed matrix.
To analyze it, it is helpful to define the notion of
implicit matrices, which is a more succinct way of writing
``linear operator''-style solver statements like
those in~\cite{SpielmanTengSolver:journal} as well as
subsequent works.
An \emph{implicit matrix} is a routine that applies
a linear operator to a vector.
It's \emph{complexity} is defined as the time it
takes to run it when given a vector. 
Note that if we have a matrix explicitly given, we can
view it as an implicit matrix with complexity equal
to one plus its number of nonzero entries.

If $\ma$ is an implicit matrix, then we will use the
notation $\ma\vec{x}$ to denote $\ma(\vec{x})$.
In particular, this notation choice means we can write
$\ma(\mb(\cdot))$ as $\ma \mb$ and
$\ma(\cdot)+ \mb(\cdot)$ as $\ma + \mb$.
If we form a new implicit matrix from two (or more) explicit
matrices in either of these manners, the complexity of the
new explicit matrix is equal to the sum of the complexities
of the matrices it was formed from.

Because an implicit matrix implements a linear operator, we
are---for the purposes of analysis---free to treat it as if
it is an actual matrix and talk about things like its eigenvalues
or whether it approximates something---provided we do so with the
understanding that whenever we say such things, we are really
talking about the linear operator that the implicit matrix implements.
When possible, we will use $\ZZ$ to denote implicit
matrices that represent inverses of matrices,
and $\MM$ to represent matrices related to linear systems
that we are trying to solve.

In particular, preconditioned Richardson iteration from
Lemma~\ref{lem:precon}
can be viewed as an implicit matrix
$\textsc{PreconRichardson}(\mI- \WWhat_i, \MM, \frac{1}{2}, O(1/\Delta))$
built from $\mI-\WWhat_i$ and $\MM$, which might themselves be implicit matrices.
With this in mind we can state the function that does most of the work in our algorithm
in Figure~\ref{fig:solve}.

\begin{figure}[ht]
\begin{algbox}
$\textsc{Solve}((\AAcal_{i} \ldots \AAcal_{d}), \widehat{\lambda}, \epsilon)$ 

\textbf{Input:} Matrices $\AAcal_{i} \ldots \AAcal_{d}$
forming a subsequence of a $(d,\weps,1/4)$-chain corresponding to an \\
\makebox[1.22cm]{} Eulerian Laplacian $\mlap=\md-\ma^\top$, a lower bound $\widehat{\lambda}$ on $\lambdanonzero(\md^{-1/2}\mlap\md^{-1/2})$, accuracy $\epsilon$. \\
\textbf{Output:} Implicit matrix that is an $\epsilon$-approximate pseudoinverse of
$\mI-\WWhat_i$ with respect \\
\makebox[1.22cm]{} to $\mI-\mU_{\WWhat_i}$.
\begin{enumerate}
\item If $i = d$, 
\begin{enumerate}
\item $\ell \leftarrow \min\{1/4, 1.125^d \cdot 0.9 \cdot \widehat{\lambda} \}$.
\item Return
	$\textsc{PreconRichardson}(\II - \AAcal_{d}, \frac{\ell}{4}\mI_{\im{\mI-\WWhat_d}},1, \frac{8}{\ell^2} \log(1 / \epsilon) )$.
\end{enumerate}
\item  $\Delta \leftarrow \min\{ \sqrt{d\log d}, d-i\}$.

\item $\ZZtil_i \leftarrow  {\left(1-\alpha\right)^{\Delta}}
        \cdot \textsc{Solve}(\AAcal_{i+\Delta} \ldots \AAcal_{d}, \widehat{\lambda}, \exp(-5\Delta)/30)
        \cdot (\II + \AAcal_{i+\Delta-1}^{(1/4)})
		\cdots  (\II + \AAcal_{i}^{(1/4)})$.
\item Return $\textsc{PreconRichardson}( \mI-\WWhat_i, \ZZtil_i, 1, \log(1 / \epsilon) )$.
\end{enumerate}
\end{algbox}
\caption{Algorithm that produces a matrix polynomial
that produces an $\epsilon$-approximate pseudoinverse of $\mI-\WWhat_i$ with respect to $\mI-\mU_{\WWhat_i}$, using the global solver chain
constructed via \textsc{BuildChain}.}
\label{fig:solve}
\end{figure}

We now show that given a  square-sparsifier chain
and access to an implicit matrix that is an approximate pseudoinverse
of $\mI-\WWhat_{d}$ with respect to $\mI-\mU_{\WWhat_d}$, we can efficiently compute
an implicit matrix $\ZZ_{0}$ which is an approximate pseudoinverse of
$\mI-\WWhat_0$ with respect to $\mI-\mU_{\WWhat_0}$. 

This done by invoking the transformations of approximate
pseudoinverses in Lemma~\ref{lem:chainproperty}, but also
swapping out the exact $(\mI-\WWhat_{i})^{\dag}$ with
an operator which is an approximate pseudoinverse for it.

We first provide a helper lemma, which shows that error does not increase between recursive calls to the $\textsc{Solve}$ routine.

\begin{lem}\label{lem:chain_sequence}
Let $\WWhat_0, \WWhat_1, \ldots, \WWhat_d$ be a $(d,\weps,1/4)$-chain. Let $0 \leq i < d$, and let $\Delta = \min\{\sqrt{d\log d}, d-i\}$. Suppose that $\weps \leq \exp(-5\Delta)/30$, and that for any $\epsilon_\Delta \leq \exp(-5\Delta)/30$, calling the routine $\textsc{Solve}((\WWhat_{i+\Delta}, \ldots, \WWhat_d), \widehat{\lambda}, \epsilon_{\Delta})$ returns an implicit matrix $\ZZ_{i+\Delta}$ which is an $\epsilon_{\Delta}$-approximate pseudoinverse of $\mI-\WWhat_{i+\Delta}$ with respect to $\mI-\mU_{\WWhat_{i+\Delta}}$. Then, calling the routine $\textsc{Solve}((\WWhat_i, \ldots, \WWhat_d), \widehat{\lambda}, \epsilon)$ returns an implicit matrix $\ZZ_i$ which is an $\epsilon$-approximate pseudoinverse of $\mI-\WWhat_i$ with respect to $\mI-\mU_{\WWhat_i}$.
\end{lem}
\begin{proof}
First we notice that by Lemma~\ref{lem:buildChain} part~\ref{item:prop1}, $\kappa(\mI-\mU_{\WWhat_i}, \mI-\mU_{\WWhat_i +\Delta}) \leq 21^\Delta$. Therefore, by the norm change property from Lemma~\ref{lem:approxInv-composition} part~\ref{part:weak-norm-change}, and our hypothesis, we obtain that $\ZZ_{i+\Delta}$ is an $(\sqrt{21^\Delta}\cdot \epsilon_\Delta)$- and therefore also a $(1/30)$-approximate pseudoinverse of $\mI-\WWhat_{i+\Delta}$  with respect to $\mI-\mU_{\WWhat_i}$.

On the other hand, the error propagation down the chain, bounded in Lemma~\ref{lem:chainproperty} gives that the matrix
\[
\ZZbar_i
=
{(1-1/4)^\Delta} (\mI- \WWhat_{i + \Delta})^{\dagger}
\left( \II +  \WWhat_{i + \Delta - 1}^{(1/4)} \right) \ldots \left( \II +  \WWhat_{i}^{(1/4)} \right) 
\]
is an $(\exp(5 \Delta) \widehat{\epsilon})$- and by the bound on $\weps$ also a $(1/30)$-approximate pseudoinverse of $\mI-\WWhat_{i}$ with
respect to $\mI-\UU_{\WWhat_i}$.

Using these two facts, we can now show that the implicit matrix
\[
\ZZtil_i = {(1-1/4)^\Delta} \ZZ_{i + \Delta}
\left( \II +  \WWhat_{i + \Delta - 1}^{(1/4)} \right) \ldots \left( \II +  \WWhat_{i}^{(1/4)} \right) 
\]
is an $1/2$-approximate pseudoinverse of $\mI-\WWhat_{i}$ with
respect to $\mI-\UU_{\WWhat_i}$. 

This can be easily seen by applying the properties of approximate pseudoinverses we proved in Section~\ref{sec:approxInverse}. Letting $\MM = (1-1/4)^\Delta ( \II +  \WWhat_{i + \Delta - 1}^{(1/4)} ) \ldots ( \II +  \WWhat_{i}^{(1/4)} ) $, and applying Lemma~\ref{part:approxInv-multiply} part~\ref{part:weak-norm-change} we obtain that $(\mI-\WWhat_{i+\Delta})^\dag$ is a $(1/30)$-approximate pseudoinverse of $\MM (\mI-\WWhat_i)$ with respect to $\mI-\mU_{\WWhat_i}$.
Applying the triangle inequality from Lemma~\ref{lem:approxInvBasic} we then obtain that $\ZZ_{i+\Delta}$ is a $(1/10)$-approximate pseudoinverse of $\MM(\mI-\WWhat_i)$ with respect to $\mI-\mU_{\WWhat_i}$. Finally, applying Lemma~\ref{part:approxInv-multiply} part~\ref{part:weak-norm-change} again, we obtain that $\ZZtil_i = \ZZ_{i+\Delta} \MM$ is a $(1/10)$-approximate pseudoinverse of $\mI-\WWhat_i$ with respect to $\mI-\mU_{\WWhat_i}$.

Finally, the guarantee on the output matrix $\ZZ_i = \textsc{PreconRichardson}( \LL_i, \ZZtil_{i},
1, \log(1 / \epsilon))$
then follows from Lemma~\ref{lem:precon}.
\end{proof}

Given this, we can now analyze the quality of the implicit matrix produced by calling $\textsc{Solve}$ on the entire square sparsification chain.
\begin{lem} \label{lem:chainerror}
Given a $(d,\weps,1/4)$-square-sparsifier chain 
$\AAcal_0, \AAcal_1, \ldots \AAcal_d$ constructed for a Laplacian $\mlap = \md-\ma^{\top}$, with $\weps \leq \exp(-5\Delta)/30$,
calling the  routine $\textsc{Solve}((\AAcal_0 \ldots \AAcal_d), \widehat{\lambda}, \epsilon)$, where $\epsilon \leq \exp(-5\Delta)/30$,
returns an implicit matrix $\ZZ$ which is
an $\epsilon$-approximate pseudoinverse of $\mI-\WWhat_0$ with respect to $\mI-\mU_{\WWhat_0}$. 
\end{lem}
\begin{proof}
The proof relies on Lemma~\ref{lem:chain_sequence}, and follows from induction on the depth of the call to the \textsc{Solve} routine. 

The base case is $i = d$, for which we prove that 
the operator 
\[
\ZZ_{d}
= \textsc{PreconRichardson}\left(\mI-\WWhat_{d}, \frac{\ell}{4}\II_{\im{\mI-\WWhat_d}}, 1, \frac{8}{\ell^2}\log\left(1 / \epsilon\right)\right)
\]
is an $\epsilon$-approximate pseudoinverse of $\mI-\WWhat_d$
with respect to $\mI-\mU_{\WWhat_d}$.
In order to do so, we show that the input parameters fulfill the requirements for applying Lemma~\ref{lem:solveWellConditioned}.

First we notice that by Lemma~\ref{lem:buildChain}, we have $\lambdanonzero(\mI- \mU_{\WWhat_d}  ) \geq \min\{1/4, \lambdanonzero(\mI - \mU_{\WWhat_0})\cdot 1.125^d \}$. Also, from Lemma~\ref{lem:chain_construction}, we know that $\mI-\WWhat_0$ is a $(1/10)$-approximation of $\md^{-1/2} \mlap \md^{-1/2}$. Therefore by Lemma~\ref{lem:asym_strong_implies_undir} we know that $\lambdanonzero(\mI-\mU_{\WWhat_0}) \geq 9/10 \cdot \lambdanonzero(\mU_{\md^{-1/2}\mlap \md^{-1/2}})\geq 9/10 \cdot \widehat{\lambda}$. Hence $\lambdanonzero(\mI-\mU_{\WWhat_d}) \geq \min\{1/4, 1.125^d \cdot 0.9 \cdot \widehat{\lambda} \} = \ell$.

By  Lemma~\ref{lem:normalized_lap_properties} we have $\norm{\mI-\WWhat_d}_2 \leq 2$, and therefore, applying Lemma~\ref{lem:solveWellConditioned}, we obtain that
 $\textsc{PreconRichardson}(\mI-\WWhat_d, \frac{\ell}{4}\cdot\mI_{\im{\mI-\WWhat_d}}, 1, \frac{8}{\ell^2} \log(1/\epsilon))$ returns an implicit matrix $\ZZ$ that is an $\epsilon$-approximate pseudoinverse of $\mI-\WWhat_d$ with respect to $\mI-\mU_{\WWhat_d}$.
 
Now for the induction step, suppose that the induction hypothesis holds for $i+\Delta$. Then, by Lemma~\ref{lem:chain_sequence}, the matrix  produced when calling the chain starting at $i$, $\ZZ_i$, is an $\epsilon$-approximate pseudoinverse of $\mI-\WWhat_i$ with respect to $\mI-\mU_{\WWhat_i}$. Therefore, this property also holds for the matrix at the bottom of the call stack, which is what we wanted to prove.
\end{proof}

Having seen that the $\textsc{Solve}$ routine controls the accumulation of error, as it produces an approximate pseudoinverse for the first matrix in the square sparsification chain, we can now use its output as a preconditioner for the Richardson iteration, which yields our final algorithm, described in Figure~\ref{fig:solveEulerian}.

\begin{figure}[ht]
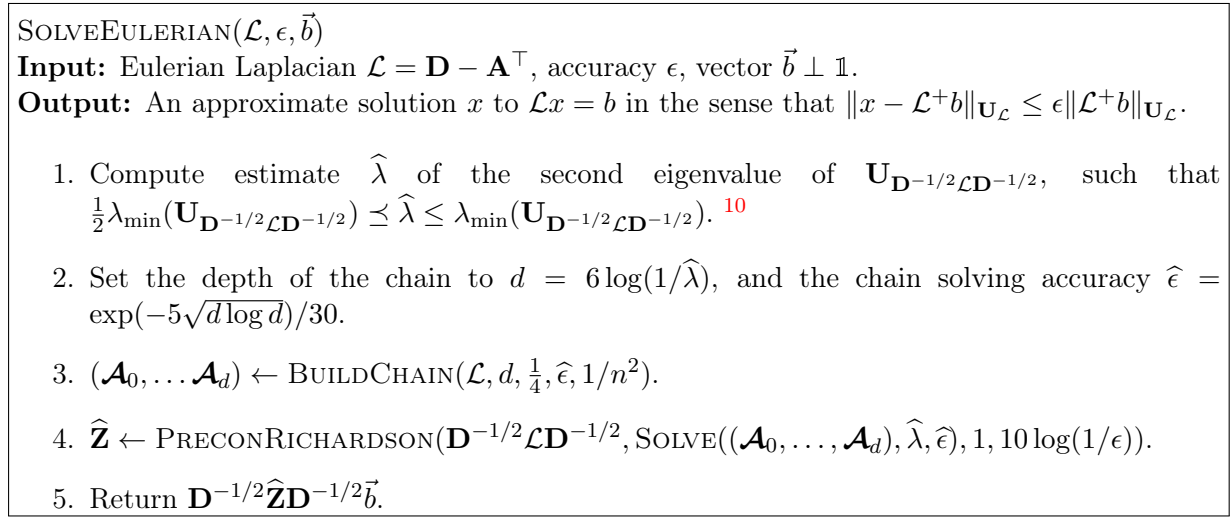

\begin{algbox}
$\textsc{SolveEulerian}(\mlap, \epsilon,\vec{b})$

\textbf{Input:} Eulerian Laplacian $\mlap = \DD - \AA^{\top}$,
	accuracy $\epsilon$, vector $\vec{b} \perp \vones$.
	
\textbf{Output:} An approximate solution $x$ to $\mlap x = b$ in the sense that $\norm{x-\mlap^\dagger b}_{\mU_{\mlap}} \leq \epsilon \norm{\mlap^\dagger b}_{\mU_{\mlap}}$.

\begin{enumerate}

\item Compute estimate $\widehat{\lambda}$ of the second eigenvalue of $\UU_{\md^{-1/2}\mlap\md^{-1/2}}$, such that $\frac{1}{2}\lambda_{\min}(\UU_{\md^{-1/2}\mlap\md^{-1/2}}) \preceq \widehat{\lambda} \leq \lambda_{\min}(\UU_{\md^{-1/2}\mlap\md^{-1/2}})$.~\footnotemark

\item Set the depth of the chain to $d = 6 \log (1/\widehat{\lambda})$, and the chain solving accuracy $\weps = \exp(-5\sqrt{d \log d})/30$.

\item $(\AAcal_0, \ldots \AAcal_d) \leftarrow
	\textsc{BuildChain}(\mlap,d, \frac{1}{4}, \weps, 1/n^2)$.
 
\item $\widehat{\ZZ} \leftarrow \textsc{PreconRichardson}( \md^{-1/2}\mlap\md^{-1/2},\textsc{Solve}((\AAcal_0, \ldots, \AAcal_{d}),\widehat{\lambda},\weps),1, 10\log(1 / \epsilon) )$.

 \item Return $\DD^{-1/2} \widehat{\ZZ} \DD^{-1/2} \vec{b}$.
\end{enumerate}
\end{algbox}
\caption{Full algorithm for solving Eulerian Laplacian systems.}
\label{fig:solveEulerian}
\end{figure}

\footnotetext{Since we know the nullspace of $\UU_{\LL}$, its minimum non-zero eigenvalue can be estimated in linear time to high accuracy  with high probability via inverse powering. See e.g. Section 7 of~\cite{SpielmanTengSolver:journal} or Chapter 8 of~\cite{Vishnoi13}.}

What we have left is to bound the running time of our solver. We do so by analyzing the recursion tree, as well as the outermost call to the preconditioned Richardson iteration. The final bound is provided in the theorem below.

\begin{thm}[Eulerian Solver Guarantee]\label{thm:eulerianGuarantee}
Given an Eulerian Laplacian $\mlap = \md-\ma^\top \in \R^{n\times n}$ with $m$ nonzero entries, and given an error parameter $0 < \epsilon \leq 1/2$, the algorithm $\textsc{SolveEulerian}(\mlap, \epsilon, \vec{b})$ returns, with probability at least $1-1/n$, an approximate solution $\vec{x}$ to $\mlap \vec{x} = \vec{b}$ in the sense that
$$\normFull{\vec{x} - \mlap^\dag \vec{b}}_{\mU_\mlap} \leq \epsilon \normFull{\mlap^\dag \vec{b}}_{\mU_\mlap}\,{.}$$
Furthermore, the total running time is 
$$\Otil \left( \left( m + n  e^{ O\left(\sqrt{\log \kappa \cdot \log \log \kappa} \right)} \right)
	\log \left(1/\epsilon \right)\right),
$$
 where $\kappa$ is the condition number of the normalized Laplacian $\md^{-1/2}\mlap \md^{-1/2}$.
\end{thm}
\begin{proof}
By Lemma~\ref{lem:chain_construction}, it takes 
\[
\Otil(m + n   \weps^{-2} d) = \Otil\left(m + n d \cdot e^{O\left(\sqrt{d \log d}\right)}\right)
\]
 time to build the square sparsification chain.

Next we analyze the cost of the recursive calls to $\textsc{Solve}$. As we can see in the description of $\textsc{Solve}$,  when invoked on $\mI-\WWhat_d$, $\textsc{PreconRichardson}$ requires $({8}/{\ell^2}) \log(1/\weps)$ iterations, where $\ell= \min\{1/4, 1.125^d \cdot 0.9 \widehat{\lambda} \}$. Therefore, whenever $d = O(\log \widehat{\lambda}^{-1})$, the base case of \textsc{Solve} requires 
\[
O\left(\frac{1}{\widehat{\lambda} e^{O(d)}} \log(1/\weps)\right ) 
= O\left( \frac{ \kappa \sqrt{d \log d}}{  e^{\Omega(d)} }\right)
\] iterations.

For the cost of invoking $\textsc{Solve}(\mI-\WWhat_0, \widehat{\lambda}, \weps)$,
note that at each recursive call, the branching factor in the recursion is
\[
O\left(\log (1/\weps) \right) = O\left(\sqrt{d \log d}\right)\,{,}
\]
due to the iterations  of preconditioned Richardson, each of them invoking the solver for a matrix from further down the chain.

Furthermore, the depth of the call stack for \textsc{Solve} (or the number of layers of the recursion tree)
is bounded by $d / \sqrt{d \log d} = \sqrt{d / \log d}$.
Therefore the total number of recursive calls made before the final solve on $\mI-\WWhat_d$ is
\[
O(\sqrt{d \log d})^{O(\sqrt{d / \log d})}
= e^{O(\sqrt{d \log{d}})}.
\]
Now, note that each of the recursive call requires one multiplication by each of the matrices in the chain. Therefore each such recursive call makes
\[
\Otil\left(n d \weps^{-2} \right) = \Otil\left(n d \cdot e^{O\left(\sqrt{d \log d}\right)}\right)
\]
work. Thus we see that the total amount of work it takes to construct and invoke the implicit matrix given by $\textsc{Solve}(\mI-\WWhat_0, \widehat{\lambda}, \weps)$
is 
\[
\Otil\left( nd \cdot e^{O(\sqrt{d \log d})}
\right) \cdot 
e^{O(\sqrt{d \log d})} \cdot
O\left(\frac{ \kappa \sqrt{d \log d}}{  e^{\Omega(d)} }\right)
=
\Otil\left(
n \kappa \cdot \frac{d}{ e^{\Omega\left(\sqrt{d / \log d}\right)}}
\right)
\,{.}
\]
Hence for our setting of $d = \Theta(\log \kappa)$, this quantity becomes
\[
\Otil\left(n e^{O(\sqrt{\log \kappa \log \log \kappa})}\right)\,{.}
\]

Finally, since by definition and Lemma~\ref{thm:spectral-to-approxInv} we know that $(\mI-\WWhat_0)^\dagger$  is a $1/10$-approximate pseudoinverse of $\md^{-1/2} \mlap \md^{-1/2}$ with respect to $\UU_{\md^{-1/2} \mlap \md^{-1/2}}$, producing the implicit matrix $\widehat{\ZZ}$
requires $O(\log(1/\epsilon))$ iterations (by  Lemma~\ref{lem:precon}), each of them requiring one multiplication by $\md^{-1/2}\mlap\md^{-1/2}$ and one call to the recursive \textsc{Solve}. Thus the total running time to construct and apply $\widehat{\ZZ}$ in \textsc{SolveEulerian} is
\[
\Otil  \left( \left ( m  +   n  e^{ O (\sqrt{\log \kappa \log \log \kappa})  } \right ) \log  \frac{1}{ \epsilon}  \right) \,{.}
\]

To analyze the solution quality, let 
\[
\vec{y} \defeq \widehat{\ZZ} \DD^{-1/2} \vec{b},
\]
and let $\vec{x} =  \DD^{-1/2} \vec{y}$ be the solution returned by the algorithm. Also, to simplify notation, let us define $\LL = \md^{-1/2} \mlap \md^{-1/2}$.
 
For any vector $\vec{v}$,
let $\mI_{\perp v}$ denote orthogonal projection projection orthogonal to $\vec{v}$.
 By Lemma~\ref{lem:precon}, we have
\begin{align*}
\normFull{\vec{y}-\LL^\dagger \md^{-1/2} \vec{b}}_{\mU_{\LL}}
&\leq \epsilon \normFull{\LL^\dagger \md^{-1/2} \vec{b}}_{\mU_{\LL}}\,{.} 
\end{align*}
Since $\mU_{\LL} = \DD^{-1/2} \mU_{\mlap} \DD^{-1/2}$,
we have that for every vector $\vec{v}$,
$\norm{\vec{v}}_{\mU_{\LL}} =   \norm{\DD^{-1/2} \vec{v}}_{\mU_{\mlap}}$.
Also, since $\ker(\mU_\mlap)=\sspan(\vones)$
we have that for every vector $\vec{v}$, 
$\norm{\vec{v}}_{\mU_{\mlap}}
=   \norm{ \mI_{\perp \vones} \DD^{-1/2} \vec{v}}_{\mU_{\mlap}}.$ Using these, the above inequality is equivalent to 
\[
\epsilon \normFull{\mI_{\perp \vones} \md^{-1/2} \LL^\dagger \md^{-1/2} \vec{b}}_{\mU_{\mlap}}
\geq \normFull{\vec{x}-\mI_{\perp \vones}\md^{-1/2}\LL^\dagger \md^{-1/2} \vec{b}}_{\mU_{\mlap}}
= \normFull{\vec{x} - \mlap^{\dagger}}_{\mU_{\mlap}},
\]
where we have used $\mI_{\perp \vones} \vec{x} = \vec{x}.$

To finish the proof, it suffices to show that
$\mI_{\perp \vones} \md^{-1/2}\LL^\dagger \md^{-1/2} \vec{b} = \mlap^\dagger \vec{b}$.
We first make the substitution $\mI_{\perp \vones} = \mlap^\dagger \mlap$
and then use $\mlap = \md^{1/2} \LL \md^{1/2}$
\[
\mI_{\perp \vones} \md^{-1/2}\LL^\dagger \md^{-1/2} \vec{b}
= \mlap^\dagger \mlap (\md^{-1/2}\LL^\dagger \md^{-1/2}) \vec{b} \\
= \mlap^\dagger \md^{1/2} \LL \LL^\dagger \md^{-1/2} \vec{b}.
\]
Since $\DD^{1/2} \vones$ is the kernel of $\LL$,
we have $\LL \LL^{\dagger}  =  \mI_{\perp \DD^{1/2} \vones}$.
By the fact that $\vec{b} \perp  \vones$, we have $\DD^{-1/2} b \perp \DD^{1/2}  \vones$,
and therefore $ \mI_{\perp \md^{1/2}  \vones} \DD^{-1/2} \vec{b}  = \DD^{-1/2} \vec{b}$.
Making these substitutions, we obtain
\begin{align*}
\mI_{\perp  \vec{1}} \md^{-1/2}\LL^\dagger \md^{-1/2} \vec{b}
&= \mlap^\dagger \vec{b}.
\end{align*}

\end{proof}

\bibliographystyle{plain}
\bibliography{ref}

\appendix
\section{Entrywise Sparsification\label{sec:entry_sparsification}}

In this section we prove Theorem~\ref{thm:concentration_entry},
our main result about entrywise sampling for sparsification. Our
main technical tool for this is a rectangular matrix concentration
of Tropp that we restate below:
\begin{thm}[Matrix Bernstein (Theorem 1.6 of \cite{Tropp12}, restated)]
\label{thm:concentration_tropp} Let $\mz_{1},...,\mz_{k}\in\R^{d_{1}\times d_{2}}$
be independent random matrices such that $\E\mz_{i}=\mzero$ and $\norm{\mz_{i}}_{2}\leq R$
almost surely for all $i$. Then 
\[
\Pr\left[\normFull{\sum_{i\in[k]}\mz_{i}}_{2}\geq t\right]\leq(d_{1}+d_{2})\cdot\exp\left(\frac{-t^{2}/2}{\sigma^{2}+Rt/3}\right)
\]
where
\[
\sigma^{2}\defeq\max\left\{ \normFull{\sum_{i\in[k]}\E\mz_{i}\mz_{i}^{\top}}_{2}\,,\,\normFull{\sum_{i\in[k]}\E\mz_{i}^{\top}\mz_{i}}_{2}\right\} \,.
\]

\end{thm}
First we simplify this theorem, tailoring it to the case where we
are sampling a sequence of matrices with the same expectation. 
\begin{thm}
\label{thm:concentration_simple} Let $\dist$ be a distribution over
$\R^{d_{1}\times d_{2}}$. Let $\mSigma\defeq\E_{\ma\sim\dist}\ma$
and let $R_{\dist}$ and $\sigma_{\dist}^{2}$ satisfy
\[
\max\left\{ \norm{\E_{\ma\sim\dist}\ma\ma^{\top}}_{2}\,,\,\norm{\E_{\ma\sim\dist}\ma^{\top}\ma}_{2}\right\} \leq\sigma_{\dist}^{2}\enspace\text{ and }\enspace\max_{\ma\in\supp(\dist)}\norm{\ma}_{2}\leq R_{\dist}\,.
\]
Then for $\ma_{1},...,\ma_{k}$ sampled independently from $\dist$
we have
\[
\Pr\left[\left\Vert\frac{1}{k}\sum_{i\in[k]}\ma_{i}-\mSigma\right\Vert_{2}\geq\epsilon\right]\leq(d_{1}+d_{2})\cdot\exp\left(\frac{-k\epsilon^{2}/2}{\sigma_{\dist}^{2}+R_{\dist}\epsilon}\right)\,.
\]
and for $k\geq 64 \cdot \left(\frac{\sigma_{\dist}^{2}}{\epsilon^{2}}+\frac{R_{\dist}}{\epsilon}\right)\log\frac{d}{p}$
it is the case that $\Pr\left[\norm{\frac{1}{k}\sum_{i\in[k]}\ma_{i}-\mSigma}_{2}\geq\epsilon\right]\leq p$.\end{thm}
\begin{proof}
Let $\mz_{i}\defeq\ma_{i}-\mSigma$. Clearly $\E\mz_{i}=\mzero$,
and by Jensen's inequality,
\[
\norm{\mz_{i}}_{2}\leq\norm{\ma_{i}}_{2}+\norm{\mSigma}_{2}\leq\norm{\ma_{i}}_{2}+\E_{\ma\sim\dist}\norm{\ma}_{2}\leq2\cdot R_{\dist}\,.
\]
Furthermore, 
\[
\mzero\preceq\E\mz_{i}^{\top}\mz_{i}=\E_{\ma\sim\dist}(\ma-\mSigma)^{\top}(\ma-\mSigma)=\E_{\ma\sim\dist}\ma^{\top}\ma-\mSigma^{\top}\mSigma\preceq\E_{\ma\sim\dist}\ma^{\top}\ma
\]
and
\[
\mzero\preceq\E\mz_{i}^{\top}\mz_{i}=\E_{\ma\sim\dist}(\ma-\mSigma)(\ma-\mSigma)^{\top}=\E_{\ma\sim\dist}\ma\ma^{\top}-\mSigma\mSigma^{\top}\preceq\E_{\ma\sim\dist}\ma\ma^{\top}\,.
\]
Consequently, 
\[
\max\left\{ \normFull{\sum_{i\in[k]}\E\mz_{i}\mz_{i}^{\top}}_{2}\,,\,\normFull{\sum_{i\in[k]}\E\mz_{i}^{\top}\mz_{i}}_{2}\right\} \leq k\cdot\sigma_{\dist}^{2}\,.
\]
Therefore, by Theorem~\ref{thm:concentration_tropp} we have that
for all $t$,
\[
\Pr\left[\normFull{\sum_{i\in[k]}\mz_{i}}_{2}\geq t\right]\leq(d_{1}+d_{2})\cdot\exp\left(\frac{-t^{2}/2}{k\cdot\sigma_{\dist}^{2}+2Rt/3}\right)\,.
\]
Since $\sum_{i\in[k]}\mz_{i}=k\cdot(\frac{1}{k}\sum_{i\in[k]}\ma_{i}-\mSigma)$
picking $t=k\cdot\epsilon$ yields the result.
\end{proof}
Using Theorem~\ref{thm:concentration_simple} we can now prove Theorem~\ref{thm:concentration_entry},
our main result of this section. 

\begin{proof}[Proof of Theorem~\ref{thm:concentration_entry}]
First note that by the definition of $s$ we have that
\[
\sum_{i,j} p_{ij} = \frac{1}{s} \sum_{i,j} \left[ \frac{\ma_{ij}}{\vr_i} + \frac{\ma_{ij}}{\vc_j}\right]
= \frac{1}{s} \left[\text{\# non-zero rows} + \text{\# non-zero columns}\right]
= 1
\]
and therefore $\dist$ is a valid probability distribution. All that remains is to prove each of the claims of Theorem~\ref{thm:concentration_entry}
by carefully applying Theorem~\ref{thm:concentration_simple}. 

First apply Theorem~\ref{thm:concentration_simple} for the distribution $\mathcal{D}$ which assigns probability $p_{ij}$ to matrix $\frac{\ma_{ij}\indic_{i}\indic_{j}^{\top}}{\sqrt{\vr_{i}\cdot \vc_{j}}}$.
For this application of Theorem~\ref{thm:concentration_simple}, using that $x \cdot y \leq \frac{1}{2} x^2 + \frac{1}{2} y^2$ we
have 
\[
R_{\dist}
= \max_{i,j}\normFull{\frac{\ma_{ij} \indic_{i}\indic_{j}^{\top} }{\sqrt{\vr_{i}\cdot \vc_{j}}}\cdot\frac{1}{p_{ij}}}
= \max_{i,j} s \cdot \frac{1}{\sqrt{\vr_i \vc_j}} \left(
\frac{1}{\vr_i} + \frac{1}{\vc_j}
\right)^{-1}
\leq
\frac{s}{2}\,.
\]
Furthermore we have
\[
\normFull{\E_{\mm\sim\dist}\mm\mm^{\top}}_{2}
=
 s \normFull{
	\sum_{i,j} \frac{1}{\vr_{i}} \cdot 
	\frac{1}{\vc_{j}} \cdot \ma_{ij} \cdot \indic_{i} \indic_{i}^{\top} \cdot \left(
	\frac{1}{\vr_i} + \frac{1}{\vc_j}
	\right)^{-1} }_{2}
\leq s
\]
and
\[
\normFull{\E_{\mm\sim\dist}\mm^{\top}\mm}_{2} 
= s \normFull{\sum_{i,j}\frac{1}{\vr_{i}}\cdot\frac{1}{\vc_{j}}\cdot\ma_{ij}\cdot\indic_{j}\indic_{j}^{\top} \left(
	\frac{1}{\vr_i} + \frac{1}{\vc_j}
	\right)^{-1}}_{2} \leq s\,.
\]
Consequently, $\sigma_{\dist}\leq s$ and since $\epsilon\in(0,1)$
and $k$ is chosen appropriately, the first inequality follows by
Theorem~\ref{thm:concentration_simple}.

Next, we apply Theorem~\ref{thm:concentration_simple} f or the distribution $\mathcal{D}$ which assigns probability $p_{ij}$ to matrix $\indic_{i}\frac{\ma_{ij}}{\vr_{i}} $.
For this application of Theorem~\ref{thm:concentration_simple} we
have 

\[
R_{\dist}=\max_{i,j} \normFull{\frac{\ma_{ij}}{\vr_{i}}\cdot\frac{1}{p_{ij}}} \leq s \cdot\frac{1}{\vr_{i}}\cdot \left(
\frac{1}{\vr_i} + \frac{1}{\vc_j}
\right)^{-1} \leq s\,.
\]
Furthermore we have
\[
\sigma_{\dist}^{2}=\normFull{\sum_{j}\frac{1}{\vr_{j}^{2}}\cdot\ma_{ij}^{2}\cdot \left(
	\frac{1}{\vr_i} + \frac{1}{\vc_j}
	\right)^{-1} \cdot s}_{2}\leq s\,{.}
\]
Since $\epsilon\in(0,1)$, with $k$ appropriately chosen, the
second inequality follows by Theorem~\ref{thm:concentration_simple}.
By symmetry the last inequality follows as well. Finally, by choosing $p$
to be $s$ times larger we can make all conditions hold simultaneously
by union bound.\end{proof}
\section{Linear Algebra Facts \label{sec:linear_algebra}}

In this section we provide various general linear algebra facts we use throughout the paper.

\begin{lem}
\label{lem:relative-diff} Suppose that $\norm{\ma-\mb}_{2}\leq\epsilon$
then for all $c>0$ we have
\[
(1-c)\mb^{\top}\mb-c^{-1}\epsilon^{2}\mI\preceq\ma^{\top}\ma\preceq(1+c)\mb^{\top}\mb+(1+c^{-1})\epsilon^{2}\mI
\]
\end{lem}
\begin{proof}
Using the trivial expansion of $\ma=\mb+(\ma-\mb)$ we have that 
\[
x^{\top}\ma^{\top}\ma x-x^{\top}\mb^{\top}\mb x=2x\mb^{\top}(\ma-\mb)x+x^{\top}(\ma-\mb)^{\top}(\ma-\mb)x
\]
Now since $xy\leq\frac{c}{2}x^{2}+\frac{1}{2c}y^{2}$ for all $x,y$
and $c>0$ we have
\[
\left|2x\mb^{\top}(\ma-\mb)x\right|\leq2\norm{\mb x}_{2}\norm{(\ma-\mb)x}_{2}\leq c\norm{\mb x}_{2}^{2}+c^{-1}\norm{(\ma-\mb)x}_{2}^{2}\,.
\]
Combining this with the fact that 
\[
0\leq x^{\top}(\ma-\mb)^{\top}(\ma-\mb)x=\norm{(\ma-\mb)x}_{2}^{2}\leq\epsilon^{2}\norm x_{2}^{2}=\epsilon^{2}x^{\top}\mI x
\]
we obtain the result.
\end{proof}

\begin{lem}
\label{lem:spectral_equivalence} For all $\ma\in\R^{n\times n}$
and symmetric PSD $\mm,\mn\in\R^{n\times n}$ such that $\ker(\mm)\subseteq\ker(\ma^{\top})$
and $\ker(\mn)\subseteq\ker(\ma)$ we have 
\[
\norm{\mm^{-1/2}\ma\mn^{-1/2}}_{2}=\max_{x,y\neq0}\frac{x^{\top}\ma y}{\sqrt{\left(x^{\top}\mm x\right)\left(y^{\top}\mn y\right)}}=2\cdot\max_{x,y\neq0}\frac{x^{\top}\ma y}{x^{\top}\mm x+y^{\top}\mn y}
\]
where in each of the maximization problems we define $0/0$ to be
$0$.\end{lem}
\begin{proof}
Let $L\defeq\norm{\mm^{-1/2}\ma\mn^{-1/2}}_{2}$. Since $\norm x_{2}=\max_{\norm y_{2}=1}y^{\top}x$ we have that 
\[
L=\max_{\norm x_{2}=\norm y_{2}=1}x^{\top}\mm^{-1/2}\ma\mn^{-1/2}y\,.
\]
Now, performing the change of basis $x:=\mm^{-1/2}x$ and $y:=\mm^{-1/2}y$
we have
\[
L=\max_{\norm x_{\mm}=\norm y_{\mn}=1}x^{\top}\ma y=\max_{x,y\neq0}\frac{x^{\top}\ma y}{\norm x_{\mm}\norm y_{\mm}}
\]
where in each of these maximization problems we restrict that $x\in\im{\mm}$
and $y\in\im{\mn}$. However, for all $x\perp\im{\mm}$ or $y\perp\im{\mn}$,
i.e. $x\in\ker(\mm)$ or $y\in\ker(\mn)$, we have that either $\norm x_{\mm}=0$
or $\norm y_{\mn}=0$ and $x^{\top}\ma y=0$. Consequently, the above
equalities hold without the $x\in\im{\mm}$ and $y\in\im{\mn}$ restriction
by our definition of $0/0=0$. The final equality we wish to prove
follows from the fact that $\norm x_{\mm}\norm y_{\mn}\leq\frac{1}{2}(\norm x_{\mm}^{2}+\norm x_{\mn}^{2})$
and that this inequality is tight when $\norm x_{\mm}=\norm y_{\mn}=1$.\end{proof}

\begin{lem}
\label{lem:simple_spec_inequalities}
For any $\MM \in \R^{n \times n}$ and symmetric positive semidefinite matrices $\ma, \mb \in \R^{n \times n}$ such that $\ma \preceq \mb$ we have that $\norm{\ma^{1/2} \mm}_2 \leq \norm{\mb^{1/2} \mm}_2$ and $\norm{\mm \ma^{1/2}}_2 \leq \norm{\mm \mb^{1/2}}_2$.
\end{lem}

\begin{proof}
The first claim follows from the fact that adopting the convention $0 / 0 = 0$
\[
\norm{\ma^{1/2} \mm}_2
= \max_{x \in \R^{n}} \frac{\norm{\ma^{1/2} \mm x}_2}{\norm{x}_2}
= \max_{x \in \R^{n}} \frac{\sqrt{x^\top \mm^\top \ma \mm x}}{\norm{x}_2}
\leq \max_{x \in \R^{n}} \frac{\sqrt{x^\top \mm^\top \mb \mm x}}{\norm{x}_2}
= \norm{\mb^{1/2} \mm}_2\,.
\]
The second follows from this and the fact that $\norm{\ma^{1/2} \mm}_2 = \norm{\mm^\top \ma^{1/2}}_{2}$.
\end{proof}

\begin{lem}
\label{lem:matrix_two_norm} Let $\mm\in\R^{n\times m},$ 
 $a\in\R_{\geq0}^{n}$ and $b\in\R_{\geq0}^{m}$ be arbitrary. Let $\AA = \diag(a)$ and $\mb = \diag(b).$
We have that for all $\alpha,\beta\in[0,1]$
\[
\norm{\ma\mm\mb}_{2}\leq\sqrt{\norm{\ma^{2\alpha}\mm\mb^{2\beta}}_{\infty}\cdot\norm{\ma^{2(1-\alpha)}\mm\mb^{2(1-\beta)}}_{1}}.
\]
Consequently, for PSD diagonal matrices $\md_{1}\in\R^{n\times n}$ and $\md_{2}\in\R^{m\times m}$
we have
\[
\norm{\md_{1}^{-1/2}\mm\md_{2}^{-1/2}}_{2}\leq\max\left\{ \norm{\md_{1}^{-1}\mm}_{\infty}\,,\,\norm{\md_{2}^{-1}\mm^{\top}}_{\infty}\right\} \,.
\]
\end{lem}
\begin{proof}
Let $x\in\R^{n}$ and $y\in\R^{m}$ be arbitrary with $\norm x_{2}=\norm y_{2}=1$.
We have 
\[
x^{\top}\ma\mm\mb y=\sum_{i,j}\mm_{i,j}a_{i}b_{j}x_{i}y_{j}\leq\sum_{i,j}|\mm_{i,j}|\cdot a_{i}\cdot b_{j}\cdot|x_{i}|\cdot|y_{j}|
\]
Consequently, by Cauchy Schwarz we have that 
\begin{align*}
\left(x^{\top}\ma\mm\mb y\right)^{2}
&\leq\left(\sum_{i,j}|\mm_{i,j}|\cdot a_{i}^{2\alpha}\cdot b_{j}^{2\beta}\cdot x_{i}^{2}\right)\cdot\left(\sum_{i,j}|\mm_{i,j}|\cdot a_{i}^{2(1-\alpha)}\cdot b_{j}^{2(1-\beta)}\cdot y_{j}^{2}\right)\\
&\leq\norm{\ma^{2\alpha}\mm\mb^{2\beta}}_{\infty}\cdot\norm{\ma^{2(1-\alpha)}\mm\mb^{2(1-\beta)}}_{1}.
\end{align*}
The final conclusion follows from the fact that for any matrix $\mc \in \R^{n\times m}$
$$ \norm{\mc}_{1} = \norm{\mc^T}_{\infty}.$$ 

\end{proof}

\begin{lem}
\label{lem:square_sym_upper_bound} If $\mm\in\R^{n\times n}$ is a symmetric matrix with $\norm{\mm}_{2}\leq1$, then
$
\mzero\preceq\mI-\mm^{2}\preceq2\cdot\left(\mI-\mm\right)\,.
$
\end{lem}
\begin{proof}
Since $\mm$ is symmetric we have that $\mm^{2}$ and $\mm$ are mutually
diagonalizable and the above inequality reduces to showing that $0\leq1-x^{2}\leq2\cdot(1-x)$
for $x\in\R$ with $|x|\leq1$. The left hand side of the inequality is true because $x^2 \leq 1$ and the right hand side follows by noticing that it is equivalent to $0\leq x^2 -2x + 1 = (x-1)^2$ after rearranging terms.
\end{proof}
The following gives a similar statement involving the symmetrization of an arbitrary matrix. While it is based on the above proof, it loses a factor of $2$ over the previous lemma.
\begin{lem}
\label{lem:square_asym_upper_bound} If $\mm\in\R^{n\times n}$ is a possibly asymmetric matrix satisfying $\norm{\mm}_{2}\leq 1$ then
\[
\mzero\preceq\mI - \mU_{\mm^2} \preceq  2 (\mI - \mU^2_{\mm}) \preceq
4 (\mI - \mU_{\mm})
\,.
\]
\end{lem}
\begin{proof}
Since the norm of $\mm$ is at most $1$, we immediately obtain $\norm{(\mm^{2})^{\top}}_2 = \norm{\mm^{2}}_{2}\leq1$, $\norm{\mm^{\top}\mm}_{2}\leq1$, $\norm{\mm\mm^{\top}}_{2}\leq1$. Then, by triangle inequality,
\[
\norm{\mU_{\mm^2}}_2 = \normFull{\frac{1}{2}(\mm^2 + (\mm^2)^{\top})}_2 \leq \frac{1}{2}(\norm{\mm^2}_2 + \norm{(\mm^2)^{\top}}_2) \leq 1\,.
\] Therefore $\mU_{\mm^2} \preceq \mI$, yielding the
left hand side of the desired inequality. 

Next, we note that these inequalities imply $\mm^{\top}\mm\preceq\mI$
and $\mm\mm^{\top}\preceq\mI$, yielding 
\[
(\mm+\mm^{\top})^2=\mm^{2}+\mm^{\top}\mm+\mm\mm^{\top}+(\mm^{\top})^{2}\preceq\mm^{2}+(\mm^{\top})^{2}+2\mI\,.
\]
Consequently,
\[
\mI - \mU_{\mm^2} = 
\mI-\frac{1}{2}(\mm^{2}+(\mm^{2})^{\top})\preceq2\mI-\frac{1}{2}(\mm+\mm^{\top})^2
=2(\mI-\mU^2_{\mm})\,.
\]
Finally, since $\mU_{\mm}$ is symmetric with $\norm{\mU_{\mm}}_{2}\leq1$, by  Lemma~\ref{lem:square_sym_upper_bound}
we have $\mI-\mU_{\mm}^2\preceq2\cdot (\mI-\mU_{\mm})$. 
\end{proof}
\begin{lem}
\label{lem:square_condition_number}Let $\mm\in\R^{n\times n}$ be
a matrix such that $\norm{\mm}_{2}\leq1$. Furthermore, for $\alpha\in[0,1)$
let $\mn=\alpha\mI+(1-\alpha)\mm$ and let $\mL_{i}=\mI-\mU_{\mn^i}$.
Then,
\[
2\alpha\mL_{1}\preceq\mL_{2}\preceq(4-2\alpha)\cdot\mL_{1}\,{.}
\]
\end{lem}
\begin{proof}
Note that $\mI-\mn=(1-\alpha)(\mI-\mm)$, and therefore $\mL_{1}=(1-\alpha)(\mI-\mU_{\mm})$.
The first identity gives us that
\[
\mn^{2}=\alpha^{2}\mI+2\alpha(1-\alpha)\mm+(1-\alpha)^{2}\mm^{2}\,.
\]
Consequently,
\begin{align*}
\mL_{2} & =\mI-\alpha^{2}\mI-2\alpha(1-\alpha)\mU_{\mm}-(1-\alpha)^{2}\mU_{\mm^2}\\
&=\left(1-\alpha^2 - (1-\alpha)^2 \right)\mI - 2\alpha (1-\alpha) \mU_{\mm} + (1-\alpha^2)\left(\mI - \mU_{\mm^2}\right) \\
&= 2\alpha(1-\alpha)(\mI - \mU_{\mm}) + (1-\alpha)^2 (\mI - \mU_{\mm^2})\\
 & =2\alpha\mL_{1}+(1-\alpha)^2(\mI - \mU_{\mm^2})\,{.}
\end{align*}
This yields the first part of the inequality, since the second term in the last line is positive semidefinite.
 Now by Lemma~\ref{lem:square_asym_upper_bound}
we know that 
\[
\mzero
\preceq\mI-\mU_{\mm^2}
\preceq
4 \left(
\mI-\mU_{\mm}\right)
=\frac{4}{1-\alpha}\mL_{1}\,{,}
\]
Plugging this into our previous identity we obtain
\begin{align*}
\mL_2 \preceq 2\alpha  \mL_1 + (1-\alpha)^2 \cdot \frac{4}{1-\alpha} \mL_1 = (4-2\alpha) \mL_1\,{,}
\end{align*}
thus yielding the result.
\end{proof}

\begin{lemma}[Condition Number Improvement]\label{lem:kappa-improvement}
Let a nonzero matrix $\mm \in \R^{n \times n}$ be such that $\ker(\mm)=\ker(\mm^{\top})$, and $\norm{\mm}_2 \leq 1$. For $\alpha \in (0, 1/4]$ let $\mn \defeq \alpha \mI + (1 - \alpha) \mm$. Then, for $\lambdanonzero \defeq \lambdanonzero(\mI - \mU_{\mm})$ we have 
\begin{equation}
\label{eq:kappa-improve2}
\lambdanonzero (\mI - \mU_{\mn^2} )
 \geq 
 \min \left\{\alpha, (1+\alpha) \lambdanonzero\right\}\,{.}
\end{equation}
\end{lemma}

\begin{proof}

Note that we can write
\[
\mU_{\mn^2} =
\frac{1}{2} \left(\mn^2 +(\mn^\top)^2\right)
= 
\left(\frac{1}{2} (\mn +  \mn^\top)\right)^2
-
\left(\frac{1}{2} (\mn -  \mn^\top) \cdot \frac{1}{2} (\mn -  \mn^\top)^{\top} \right)
\preceq \mU_{\mn}^2
\,{.}
\]
Therefore
$\mI - \mU_{\mn^2} \succeq \mI - \mU_{\mn}^2$,
so it is sufficient to lower bound the smallest nonzero eigenvalue of the latter. By expanding, we obtain:
\begin{align*}
\mI - \mU_{\mn}^2 &=\mI - ( \alpha \mI + (1-\alpha) \mU_{\mm} )^2 = \mI -  (  \mI - (1-\alpha)(\mI - \mU_{\mm}) )^2  
{.} 
\end{align*}
Since $\norm{\mm}_2 \leq 1$, we also have $\norm{\mU_{\mm}}_2\leq 1$ by triangle inequality, and thus $\lambdanonzero \mI_{\im{\mm}} \preceq \mI - \mU_{\mm} \preceq 2\mI_{\im{\mm}}$. Therefore, 
\[
\mI-(1-\alpha) 2 \mI_{\im{\mm}}\preceq 
\mI - (1-\alpha)(\mI - \mU_{\mm}) \preceq \mI - (1-\alpha) \lambdanonzero \mI_{\im{\mm}}\,{,}
\]
and equivalently
\[
\mI_{\perp \im{\mm}}+ (2\alpha-1) \mI_{\im{\mm}}\preceq 
\mI - (1-\alpha)(\mI - \mU_{\mm}) \preceq \mI_{\perp \im{\mm}}+(1 - (1-\alpha) \lambdanonzero) \mI_{\im{\mm}}\,{.}
\]

Hence, after squaring, each eigenvalue of the middle term will become upper bounded by the maximum of the squares of those in the lower and the upper bound. This can be seen as a matrix version of the inequality $b^2 \leq \max\{a^2, c^2\}$, if $a\leq b \leq c$. Hence,
\[
 (\mI - (1-\alpha)(\mI - \mU_{\mm}) )^2
\preceq \mI_{\perp \im{\mm}}+
\max\{(2\alpha-1)^2, (1 - (1-\alpha) \lambdanonzero)^2 \}\mI_{\im{\mm}}\,{,}
\]
so after subtracting both sides from $\mI$ we obtain:
\[
\mI-\mU_{\mn}^2 \succeq 1-\max\{(2\alpha-1)^2, (1 - (1-\alpha) \lambdanonzero)^2 \}\mI_{\im{\mm}}\,{.}
\]
Therefore,
\begin{align*}
\lambdanonzero(\mI-\mU_{\mn^2}) &\geq \min\{1-(2\alpha-1)^2, 1-(1-(1-\alpha)\lambdanonzero)^2\} \\
&= \min\{1-(2\alpha-1)^2, 2(1-\alpha)\lambdanonzero - (1-\alpha)^2\lambdanonzero^2\}
\,{.}
\end{align*}
Observe that if $\lambdanonzero \leq (1-3\alpha)(1-\alpha)^{-2}$, then the second part of the lower bound is at least
$$2(1-\alpha)\lambdanonzero - (1-\alpha)^2 (1-3\alpha)(1-\alpha)^{-2}\lambdanonzero = (1+\alpha)\lambdanonzero\,{.}$$
Otherwise, it can be lower bounded simply by
$$2(1-\alpha)\lambdanonzero - (1-\alpha)^2 \lambdanonzero = (1-\alpha^2)\lambdanonzero \geq 1-3\alpha{.}$$
Finally, since for $\alpha \leq 1/4$, both $1-3\alpha \geq \alpha$ and $1-(1-2\alpha)^2 \geq \alpha$ are true, the result follows.
\end{proof}

\begin{lem}\label{lem:sym-hsm}
If $\LL$ is a matrix with $\ker(\LL)=\ker(\LL^\intercal)=\ker(\mU_{\mLL})$, and $\mU_{\mLL}$ is positive semidefinite, then
\[
\mU_{\mLL} \preceq \mLL^\intercal \mU_{\mLL}^\dagger \mLL\,{.}
\]
Furthermore, for any matrix $\AA$ with the same left and right kernels as $\LL$, one has that
\[
\normFull{ \AA }_{\mU_{\LL}\rightarrow \mU_{\LL}} \leq \normFull{\mU_\LL^{\dag/2} \LL \AA  \mU_\LL^{\dag/2}}_2\,{.}
\]
\end{lem}

\begin{proof}
We decompose $\mLL$ and $\mLL^{\top}$
as the sum/difference of a symmetric matrix
$\mU$ and a skew symmetric matrix $\mv$.
Specifically, write $\mLL=\mU+\mv$ for
\begin{align*}
\mU &:= \mU_{\mLL}  =(\mLL+\mLL^{\intercal})/2 \text{\ and\ }
 \mv :=(\mLL-\mLL^{\intercal})/2.
\end{align*}
This gives
\begin{align*}
\mLL^\intercal \mU^\dagger \mLL 
&=
\mU^\intercal \mU^\dagger \mU
  +\mU^\intercal \mU^\dagger \mv
  +\mv^\intercal \mU^\dagger \mU 
  +\mv^\intercal \mU^\dagger \mv.
\end{align*}
As $\mU \mU^\dagger \mv
   =\mv$ and $ \mv^{\top} \mU^\dagger \mU= \mv^{\top}$ by our kernel assumptions, and $\mv=-\mv^\intercal$, this simplifies to
 \begin{align*}
	\mLL^\intercal \mU^\dagger \mLL 
&=
	 \mU+\mv^\intercal \mU^\dagger \mv 
	 \succeq \mU,
\end{align*}
where we used the assumption that $\mU\succeq \mzero$ to guarantee that  $\mv^\intercal \mU^\dagger \mv\succeq \mzero$ for the final inequality.

The second part of the lemma follows by writing
\[
\normFull{\AA}_{\mU_\LL\rightarrow \mU_\LL} = \normFull{\mU_\LL^{1/2}\AA\mU_\LL^{\dag/2}}_2 = \normFull{\mU_\LL^{1/2} \LL^\dag \left(\LL \AA\mU_\LL^{\dag/2}\right)}_2\,{,}
\]
then applying the equivalent form of the previous inequality $\LL^{\top \dag} \mU_\LL \LL^{\dag} \preceq \mU_\LL^\dag$ in order to obtain
\[
\normFull{\mU_\LL^{1/2} \LL^\dag \left(\LL \AA\mU_\LL^{\dag/2}\right)}_2 \leq \normFull{\mU_\LL^{\dag/2} \LL \AA \mU_\LL^{\dag/2}}_2\,{.}
\]
\end{proof}

\section{Decomposition \label{sec:decomposition}}

Here we discuss the proof of Theorem~\ref{thm:decomposition_thm},
our main result on decomposing directed graphs. The theorem follows directly from standard results involving decomposing undirected graphs into expanders~\cite{SpielmanTengSolver:journal, SpielmanT11, KLOS14}. 

Before stating the decomposition result, we first define the \textit{conductance} of a graph:

\begin{definition} Given an undirected graph $G(V,E,w)$, we define the conductance of a $S \subseteq V$ by
$$\Phi(S) \defeq \frac{ \sum_{(u,v)\in E : u \in S, v \notin S} w_{uv} }{\min\{\textnormal{vol}(S), \textnormal{vol}(V\setminus S)\}}\,{,} $$
where $\textnormal{vol}(S) \defeq \sum_{u \in S} \sum_{v : (u,v) \in E} w_{uv} $.

The conductance of $G$ is then defined as 
$$\Phi(G) = \min_{S \subset V, S\neq\emptyset} \Phi(S)$$
\end{definition}

Relating the conductance of an undirected graph to its smallest nontrivial eigenvalue is done via Cheeger's inequality:
\begin{thm}[Cheeger's inequality, rephrased]\label{thm:cheeger}
Given an undirected graph $G(V,E,w)$, with a symmetric Laplacian $\mU = \md - \ma$, one has that $\mU$'s spectral gap satisfies:
$$\lambda_2( \md^{-1/2} \mU \md^{-1/2} ) \geq \frac{\Phi(G)^2}{4}\,{.}$$
\end{thm}

We refer to the following lemma, which is implicit in~\cite{SpielmanTengSolver:journal, SpielmanT11}. 

\begin{lemma}[Lemma 31 from~\cite{KLOS14}]\label{lem:decomp}
For an unweighted graph $G = (V,E)$, in $\tilde{O}(m)$ time, we can produce a partition $V_1, \dots,V_k$ of $V$, and a collection of sets $S_1,\dots,S_k \subseteq V$ with the following properties:
\begin{enumerate}
\item For all $i$, $S_i \subseteq V_i$.
\item For all $i$, there exists a set $T_i$ with $S_i \subseteq T_i \subseteq V_i$, such that $\Phi(G(T_i)) \geq \Omega(1/\log^2 n)$.
\item At least half of the edges are found within the sets $\{S_i\}$, i.e.
$$\sum_{i=1}^k \abs{E(S_i)} = \sum_{i=1}^k \abs{\{ (a,b)\in E : a, b \in S_i \}} \geq \frac{1}{2} \abs{E}\,{.}$$ 
\end{enumerate}
\end{lemma}

We use this decomposition lemma in order to first prove the result for unweighted graphs. The more general weighted version will then follow from a bucketing argument.

The decomposition can be produced by iteratively applying Lemma~\ref{lem:decomp} on the symmetrized input graph $\mU$, choosing $\mlap^{(i)}$ to be the directed Laplacian induced on vertices in $S_i$, and $\mU^{(i)}$ be the undirected Laplacian induced on vertices in $T_i$. Since $E[S_i] \subseteq E[T_i]$, clearly $\mdiag(\ms_{\mlap^{(i)}}) \preceq \mdiag(\mU^{(i)})$.~\footnote{Note that the decomposition lemma gives us a much stronger property than what we are using, since our sparsification routine only requires degree dominance.} The lower bound on the spectral gap for each $\mU_i$ is given by the second property of Lemma~\ref{lem:decomp}, combined with Cheeger's inequality (Theorem~\ref{thm:cheeger}): the spectral gap of $\mU^{(i)}$ is at least $1/\tilde{O}(1)$. Also, note that the sum of supports of these $\mU^{(i)}$ is $O(n)$, since the graphs $T_i$ are vertex-disjoint subgraphs of $G$, we have $\sum_{i=1}^k \mU^{(i)} \preceq \mU$.

Removing the graphs induced by $S_i$ (or, equivalently, $\mlap^{(i)}$), for all $i$,  reduces the number of edges in $G$ by half. Therefore, after $\lceil \log n \rceil$ iterations of applying Lemma~\ref{lem:decomp}, the edges on $G$ will have been exhausted, and we are done. Since each iteration produces undirected Laplacians whose sum is bounded by $\mU$, the sum of all the undirected Laplacians produced during the $\lceil \log n \rceil$ iterations is at most $\lceil \log n \rceil \cdot \mU$. Also, the sum of support sizes $O(n\log n)$. Hence we have a $(\tilde{O}(n), 1/\tilde{O}(1), \tilde{O}(1))$-decomposition of $G$.

In order to obtain a weighted version of the theorem, we initially decompose the graph $G$ into $b = \lceil \log (w_{\max} / w_{\min}) \rceil$ graphs $G_1, \dots, G_b$ , where $w_{\max}$ and $w_{\min}$ represent the maximum, and minimum arc weight in $G$, respectively, and $G_j = (V, E_j)$ with $E_j = \{e \in E :  w_{\min} \cdot 2^{t-1} \leq w_e < w_{\min} \cdot 2^t \}$. For each of these graphs, the cover corresponding to the unweighted $G_j$, scaled by $w_{\min} \cdot 2^t$, becomes a $(\tilde{O}(n), 1/\tilde{O}(1), \tilde{O}(1))$-decomposition of the weighted $G_j$. Therefore, taking the union of all the decompositions from the $b$ graphs, we obtain a $(\tilde{O}(bn), 1/\tilde{O}(1), \tilde{O}(b))$-decomposition of $G$. Since all weights are polynomially bounded, $b = \tilde{O}(1)$, and we obtain the desired result.

In addition, it can be shown that the result from Lemma~\ref{lem:decomp} can be made parallel using~\cite{OrecchiaV11}. Indeed, using their SDP-based balanced partitioning routine, one can in $\tilde{O}(1)$ parallel time and $\tilde{O}(m)$ work find a balanced cut with with polylogarithmic (i.e.  $1/\tilde{O}(1)$) conductance, or certify that none exists. Such a partitioning routine is then called recursively on the pieces of the input graph that are not yet certified to be expanders, yielding a nearly-linear work, polylogarithmic time algorithm which produces the partition from Lemma~\ref{lem:decomp}. We also refer the reader to Section 6 of~\cite{PengS14}, for a discussion concerning the parallelization of the decomposition routine.

\section{The Complete Solver}
\label{sec:complete}

In this section we provide some details of the full algorithm for computing stationary distributions
and solving directed Laplacians.
Various applications of these routines are
presented in~\cite{cohen2016faster} and briefly
enumerated in Section~\ref{sec:intro:results}.
We provide details on these two routines of computing
stationary and solving linear system
because they are most important for completing the picture.

We start by stating the stationary computation algorithm
given as Algorithm 1 in Section 3.2 of~\cite{cohen2016faster}.
One difference in presentation is that the routine as shown
in~\cite{cohen2016faster} relies on solving matrices
whose diagonal entries are strictly bigger than the
total magnitude of off-diagonal entries in the corresponding
row/column.
On the other hand, we have only presented a solver for Eulerian Laplacians. The
 diagonal entries of these matrices are equal to the total magnitude of the off-diagonal entries.
As a result, we also need to incorporate the reduction from such
matrices to Eulerian Laplacians from Section 5 of~\cite{cohen2016faster}.
Pseudocode of this routine is in Figure~\ref{fig:findStationary}.

\begin{figure}[ht]

\begin{algbox}

$\vec{s}=\textsc{ComputeStationary}(\mlap, \alpha)$

\textbf{Input:} $n \times n$ directed Laplacian $\mlap$,
	with diagonal $\DD$.\\
	Restart parameter $\alpha \in [0, \frac{1}{2}]$.

\textbf{Output:} approximate stationary distribution $\vec{s}$.

\begin{enumerate}
\item Set $\vec{x}^{(0)} \leftarrow \DD^{-1} \allones$, $\epsilon \leftarrow \poly\left(\frac{\alpha}{n}\right)$.
\item For $t = 0, \ldots,k= 3 \ln \left(\alpha^{-1}\right)$
\begin{enumerate}
\item Set $\vec{e}^{(t)} \leftarrow
			\max\left\{\allzeros,~\mlap \vec{x}^{(t)},~
				\diag(\vec{x}^{(t)}) \mlap \allones \right\},$
	and let $\mE^{(t)} = \mdiag(\vec{e}^{(t)}),$ $\XX^{(t)} = \mdiag(\vec{x}^{(t)})$ be the
	corresponding diagonal matrices.
\item Create $\mlap^{(t + 1)} \in \R^{(n + 1) \times (n + 1)}$
	from $(\mE^{(t)} + \mlap) \XX^{(t)}$
	by adding a row/column to make all row/column sums $0$.
\item Create a length $n + 1$ vector $\vec{b}^{(t)}$
	with sum $0$ whose first $n$ entries are given by the vector $\frac{1}{\norm{\DD^{-1} \vec{e}^{(t)}}_1}
		\DD^{-1} \vec{e}^{(t)}$.
\item Let $\vec{z}^{(t + 1)}$
	be the first $n$ entries of the vector returned by
	$\textsc{SolveEulerian}(\mlap^{(t)}, \vec{b}^{(t)}, \epsilon )$,
	and $\vec{x}^{(t + 1)} \leftarrow \XX^{(t)} \vec{z}^{(t + 1)}$
\end{enumerate}
\item Return $\vec{s} = \frac{\DD \vec{x}^{(k + 1)}}{\norm{\DD \vec{x}^{(k + 1)}}_1}$.
\end{enumerate}
\end{algbox}

\caption{Stationary Computational Algorithm.
This routine combines the reduction from solving
strictly row/column dominant matrices to Eulerian
Laplacians from Section 5 of~\cite{cohen2016faster}
with the stationary finding algorithm from Section 3.}

\label{fig:findStationary}

\end{figure}

This can then be turned into a solver for a strongly
connected Laplacian by rescaling it by the
stationary, and solve the (approximately) Eulerian
Laplacian that result from this.
The pseudocode in Figure~\ref{fig:fullSolve}
is based on Sections 7.1. and
7.3. of~\cite{cohen2016faster}.

\begin{figure}
\begin{algbox}

$\vec{s}=\textsc{SolveFull}(\mlap, \alpha)$

\textbf{Input:} $n \times n$ directed Laplacian $\mlap$,
	with diagonal $\DD$, desired error $\epsilon$.
	vector $\vec{b}$

\textbf{Output:} approximate solution to $\mlap \vec{x} = \vec{b}$.

\begin{enumerate}
\item Estimate the mixing time of $\mlap$, $T_{mix}$
	by binary searching on the mixing times of $\mlap + \alpha \mI$.
\item Add $\frac{\epsilon}{T_{\min} \poly(n)}$
to $\mlap$ to get a strictly diagonally dominant
matrix $\mlaphat$.
\item Compute approximate stationary distribution
of $\mlaphat$, $\vec{s}$.
\item Let $\vec{x} \leftarrow \textsc{SolveEulerian}(\mlaphat \DD^{-1} \ms, \vec{b}, \epsilon / \poly(n))$.
\item Return $\md^{-1} \ms \vec{x}$.
\end{enumerate}

\end{algbox}

\caption{Full solver algorithm for strongly
connected Laplacians.
It first perturbs $\mlap$ to form a matrix where
a good stationary distribution is easy to compute,
and uses the normalization given by it to reduce
the problem solving on an Eulerian Laplacian.}

\label{fig:fullSolve}

\end{figure}

\section{Approximating the Harmonic Symmetrization}
\label{sec:harmonic_approx}
\newcommand{\maw}{\widetilde{\ma}}
To solve an Eulerian Laplacian system $\mlap \vec{x} = \vec{b},$ \cite{cohen2016faster} instead solved the system $\mlap^\top \mU_\mlap^\dagger \mlap = \mlap^\top \mU_\mlap \vec{b}$. Since $\mU_\mlap$ is a symmetric Laplacian, its pseudoinverse can be applied in nearly linear time via a variety of methods~\cite{SpielmanTengSolver:journal,KoutisMP10,KoutisMP11,KelnerOSZ13,lee2013efficient,PengS14,KyngLPSS16}, and therefore given a linear system solver for $\mlap^\top \mU_\mlap^\dagger \mlap$ one is achieved for $\mlap$ with only a polylogarithmic running time overhead at worst. %

The matrix $\mlap^\top \mU_\mlap^\dagger \mlap$ was referred to in \cite{cohen2016faster} as the \emph{harmonic symmetrization} of $\mlap$ and it was shown that linear systems in this harmonic symmetrization can be solved in $O((nm^{3/4} + n^{2/3} m) \log^{O(1)} (n \kappa))$ time.
Here we show that if $\widetilde{\ma}$ is a strong approximation for $\ma,$ then $\maw^{\top}\mU^{\dagger}_{\maw} \maw$ is a spectral approximation for $\ma^{\top}\mU^{\dagger}_\ma\ma$. Consequently, by simply producing a sparsifier for $\mlap$ in nearly linear time using the results of Section~\ref{sec:sparsification}, solving the harmonic symmetrization using \cite{cohen2016faster}, and then using this solver as a preconditioner in the standard symmetric sense~\cite{Saad03:book}, yields an $O(m + n^{7/3} 
\log^{O(1)} (n \kappa))$ time algorithm for solving directed Laplacian systems.

\begin{lem}
If $\maw$ is an $\epsilon$-strong-approximation of $\ma$ for $\epsilon \in (0, 1/2)$ and $\ker(\mU_\ma) \subseteq \ker(\ma)$, then
\[
\left(1-2\epsilon\right)^{2} \ma^{\top} \mU_{\ma}^\dagger \ma
\preceq \maw^{\top} \mU_{\maw}^\dagger \maw
\preceq\left(1+\epsilon\right)^{3} \ma^{\top} \mU_{\ma}^\dagger \ma \,.
\]
\end{lem}

\begin{proof}
Let $\mm=\mU^{\dagger/2}_{\ma}\maw\mU^{\dagger/2}_\ma$ and $\mn=\mU^{\dagger/2}_{\ma}\ma\mU^{\dagger/2}_{\ma}$. The definition of strong approximation implies that $\norm{\mm-\mn}_{2} \leq \epsilon$. Applying Lemma~\ref{lem:relative-diff}, a general lemma regarding the difference of matrices, yields that
\[
(1 - \epsilon) \mn^{\top}\mn - \epsilon^{-1}\epsilon^{2}\mI 
\preceq \mm^{\top}\mm\preceq(1 + \epsilon) \mn^{\top} \mn+(1 + \epsilon^{-1})\epsilon^{2}\mI\,.
\]
Now let $\vec{x} \in \R^n$ be an arbitrary vector perpendicular to the kernel of $\mU_\ma$. For such a vector $\vx,$ we have
\[
\vec{x}^{\top}\mn^{\top}\mn \vec{x}
= \vec{x}^{\top}\mU^{\dagger/2}_\ma \ma^{\top}\mU^{\dagger}_\ma \ma \mU^{\dagger/2}_\ma\vec{x}
\geq \vec{x}^{\top}\mU^{\dagger/2}_{\ma}\mU_{\ma}\mU^{\dagger/2}_{\ma}\vec{x}
\geq \vec{x}^{\top}\mI \vec{x}\,{,}
\]
where we used that by Lemma~\ref{lem:sym-hsm} the harmonic symmetrization spectrally dominates the symmetric Laplacian, i.e., $\ma^\top \mU_\ma^\dagger \ma \succeq \mU_{\ma}$. Consequently, for $\mc=\ma^{\top}\mU^{\dagger}_\ma\ma$ and $\widetilde{\mc}=\maw^{\top}\mU^{\dagger}_\ma \maw$ we have 
\[
(1-2\epsilon) \vec{x}^{\top}\mU^{\dagger/2}_\ma \mc \mU^{\dagger/2}_\ma\vec{x}
\preceq \vec{x}^{\top}\mU^{\dagger/2}_\ma \widetilde{\mc} \mU^{\dagger/2}_\ma\vec{x}
\preceq(1 + 2\epsilon + \epsilon^{2}) \vec{x}^{\top}\mU^{\dagger/2}_\ma\mc\mU^{\dagger/2}_{\ma}\vec{x}\,{.}
\]
One can easily see that when $x\in\ker(\mU_\ma)$ all of the terms are $0$. Hence, we have that the above holds for all $x$. Since $\ker(\mU_\ma) \subseteq \ker(\ma)$ this in turn implies
$
(1 - 2\epsilon) \mc \preceq \widetilde{\mc} \preceq (1 + \epsilon)^2 \mc
$.  Furthermore, since
$
(1-\epsilon)\mU_{\ma}
\preceq\mU_{\maw}\preceq(1+\epsilon)\mU_{\ma}
$ by Lemma ~\ref{lem:asym_strong_implies_undir}, we also have $(1-\epsilon)\mU_{\maw}^\dagger \preceq \mU_{\ma}^\dagger \preceq (1 + \epsilon) \mU_{\maw}^\dagger$. The result follows from combining and using $(1 - 2 \epsilon) \leq (1 - \epsilon)$.
\end{proof}

\section{Reducing the Condition Number}
\label{sec:reduction}
In this section, we present a reduction from the problem of (approximately) solving Eulerian Laplacians to solving
Eulerian Laplacians that are at most polynomially ill-conditioned, with a logarithmic dependence on condition number.
This reduces the overall dependency on $\kappa$ to logarithmic
instead of sub-polynomial.
The main result of this section is:

\begin{theorem}
\label{thm:reduction}
There exists a procedure $\textsc{CrudeSolveIllConditioned}$ which, when given an $n \times n$ Eulerian Laplacian $\mlap$ with $m$ non-zeros such that
the condition number of $\mlap + \mlap^{\top}$ is bounded by $\kappa$, returns a crude approximate solution $x$ to $\mlap x = b$
in the sense that
\[
\norm{\vec{x}-\mlap^\dagger \vec{b}}_{\mU_{\mlap}} \leq \frac{1}{2} \norm{\mlap^\dagger \vec{b}}_{\mU_{\mlap}}.
\]
This procedure performs only $O(\log(n \kappa))$ calls to an approximate Eulerian Laplacian solver, each
on $O(n)$ vertices with $O(m)$ nonzero entries, with error parameter $\frac{1}{n^{O(1)}}$ and each with condition number $n^{O(1)}$,
plus $O(m \log(n \kappa))$ additional work.

Furthermore, if the Eulerian solver is an implicit polynomial,
the overall procedure is an implicit polynomial as well.
\end{theorem}
Combining this routine with iterative refinement / preconditioned Richardson iteration
as stated in Lemma~\ref{lem:precond_richardson}
implies that one can obtain an $\epsilon$-approximate solution
with $O(\log(n \kappa) \log(1/\epsilon))$ such solver calls and
$O(m \log(n \kappa) \log(1/\epsilon))$ additional work.
We note that this construction can also be used to reduce the $\log{\kappa}$
dependencies in~\cite{PengS14} to $\log{n}$.

We include this construction here because it removes
the $\exp(\sqrt{\log{\kappa}})$ dependencies in our algorithms and
replaces them with similar terms in $n$.
Since problems on directed graphs can be highly ill-conditioned even when the edge weights are all small integers, this can potentially be a significant improvement.

\begin{remark*}
	We believe that developing combinatorial techniques to mitigate the effects of ill-conditioning is an important endeavor in both theory and practice, and that the scheme we include reflects only partial progress. 
	While the technique we describe is sufficient to improve our running time bounds,
	we conjecture that far better reduction schemes are
	possible.  Some ways in which an ideal such result could improve on the one presented here include:
	\begin{enumerate}
		\item The overhead in work and depth could be as low as $O(1)$:
		the different scales would only have minimal overlap, and processing could occur in parallel.
		\item The scheme could only need to manipulate numbers
		with double or quadruple the precision of $n/\epsilon$,
		instead of involving words whose sizes are $O(1)$ bigger.
		Treating these increases
		with similar emphasis as constants in approximation
		algorithms may be a more realistic model of the
		costs of floating point arithmetic.
		\item As highly ill conditioned systems often
		arise under floating point arithmetic (due to the exponent) such reductions should ideally be robust to
		floating instead of fixed point arithmetic.
	\end{enumerate}
	Systematically developing strong versions of such ``condition number reducing reductions'' is potentially a challenging but important direction important direction for future work.
\end{remark*}

The main idea of our reduction is to collar the Laplacian to a fixed ``scale'' of edge weights, contracting edges above the scale while adding a smaller multiple of a clique to bound the lower eigenvalue.  The algorithm then uses the Laplacian at a given scale to route demand between vertices connected at about that scale.
Our algorithm is defined around the following
contraction and projection operators:
\begin{defn}
	\label{def:contract}
	Given a partition of $[n]$ into $k$ sets $S_1$, $S_2$, ..., $S_k$,
	\begin{enumerate}
		\item the contraction operator $\mc$ is the linear map
		$\R^n \to \R^k$ mapping $\indic_i$ to $\indic_j$ if $i \in S_j$.
		\item let $\operatorname{Proj}$ be the orthogonal projection onto the kernel of $\mc$.
	\end{enumerate}
\end{defn}
Pseudocode based on these operators is in
Figure~\ref{fig:reduction}.

\begin{figure}[ht]
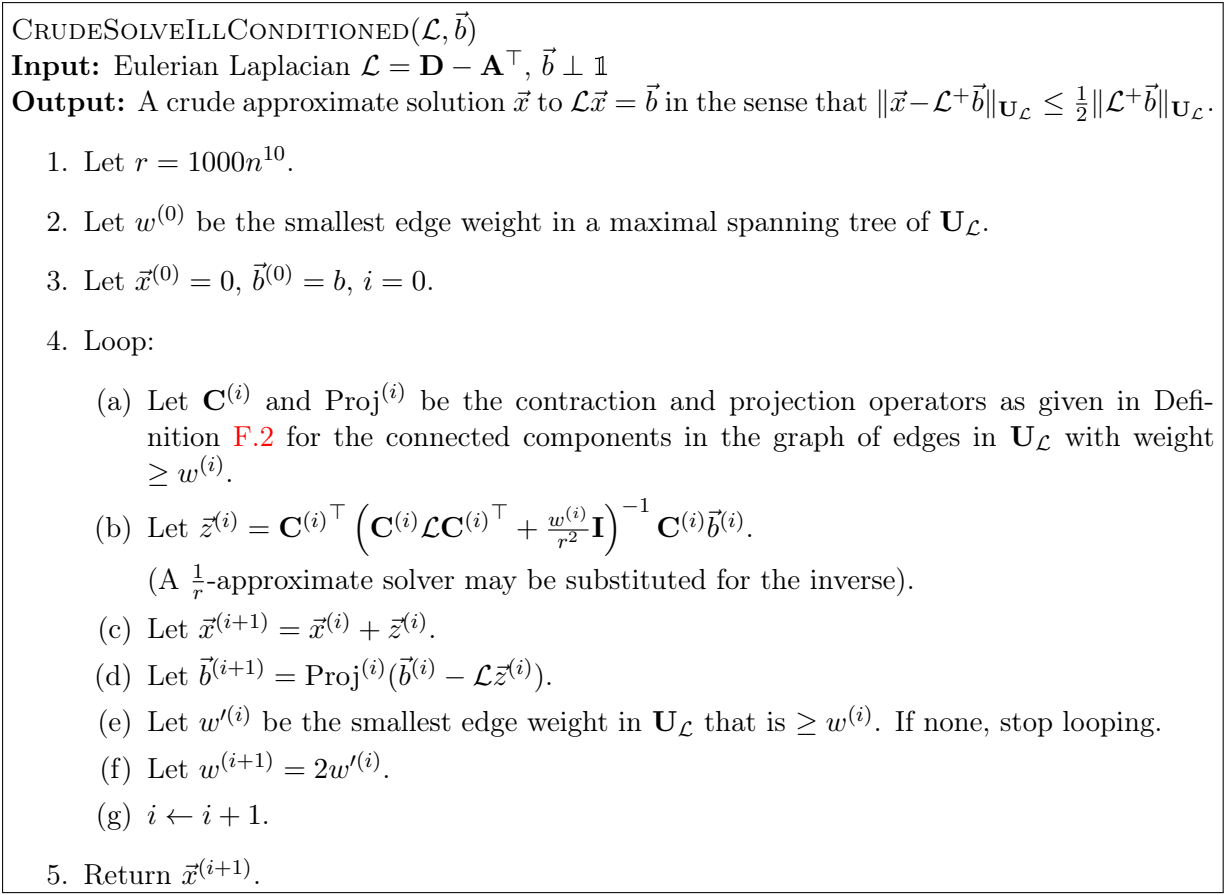

	\begin{algbox}
		$\textsc{CrudeSolveIllConditioned}(\mlap, \vec{b})$ \\
	\textbf{Input:} Eulerian Laplacian $\mlap = \DD - \AA^{\top}$,
		$\vec{b} \perp \vones$
	
	\textbf{Output:} A crude approximate solution $\vec{x}$ to $\mlap \vec{x} = \vec{b}$ in the sense that $\norm{\vec{x}-\mlap^\dagger \vec{b}}_{\mU_{\mlap}} \leq \frac{1}{2} \norm{\mlap^\dagger \vec{b}}_{\mU_{\mlap}}$.
	\begin{enumerate}
		\item
		Let $r = 1000 n^{10}$.
		\item
		Let $w^{(0)}$ be the smallest edge weight in a maximal spanning tree of $\mU_{\mlap}$.
		\item
		Let $\vec{x}^{(0)} = 0$, $\vec{b}^{(0)} = b$, $i = 0$.
		\item
		Loop:
		\begin{enumerate}
			\item
			Let $\mc^{(i)}$ and $\operatorname{Proj}^{(i)}$ be the contraction and projection operators as given in Definition~\ref{def:contract} for the connected components in the graph of edges in $\mU_{\mlap}$ with weight $\geq w^{(i)}$.
			\item
			Let $\vec{z}^{(i)} = {\mc^{(i)}}^\top \left ( \mc^{(i)} \mlap {\mc^{(i)}}^\top + \frac{w^{(i)}}{r^2} \mI \right )^{-1} \mc^{(i)} \vec{b}^{(i)}$.
			
			(A $\frac{1}{r}$-approximate solver may be substituted for the inverse).
\label{step:approxSolveProjected}
			\item
			Let $\vec{x}^{(i+1)} = \vec{x}^{(i)} + \vec{z}^{(i)}$.
			\item \label{step:bNext}
			Let $\vec{b}^{(i+1)} = \operatorname{Proj}^{(i)} (\vec{b}^{(i)} - \mlap \vec{z}^{(i)})$.
			\item
			Let $w'^{(i)}$ be the smallest edge weight in $\mU_{\mlap}$ that is $\geq w^{(i)}$.  If none, stop looping.
			\item
			Let $w^{(i+1)} = 2 w'^{(i)}$.
			\item
			$i \leftarrow i+1$.
		\end{enumerate}
	\item
	Return $\vec{x}^{(i+1)}$.
	\end{enumerate}
	\end{algbox}
	\caption{Reduction to solving well-conditioned Eulerian systems}
	\label{fig:reduction}
\end{figure}

Our proofs crucially rely on the following fact that relates
the ordinary (arithmetic) and harmonic symmetrizations of
Eulerian Laplacians.
It is immediate from Lemma 13 of~\cite{cohen2016faster},
and one side of it is also present as Lemma~\ref{lem:sym-hsm}.
\begin{lemma}
\label{lem:symBound}
Let $\mlap$ be an Eulerian Laplacian and $\mU_\mlap$ its symmetrization $\frac{\mlap+\mlap^\top}{2}$.  Then
\begin{equation*}
\mU_\mlap \preceq \mlap^\top {\mU_\mlap}^\dagger \mlap \preceq 2(n-1)^2 \mU_\mlap.
\end{equation*}
\end{lemma}

We also need a simple technical lemma about the extreme values in
solutions to Eulerian Laplacian systems in terms of the support
of the demand vector.

\begin{lem}
\label{lem:voltages}
Let $\mlap$ be a connected Eulerian Laplacian and $\vec{x}$ and $\vec{b}$ be nonzero vectors such that $\mlap \vec{x} = \vec{b}$.  Then the maximum value of the entries of $\vec{x}$, as well as the minimum value, must be attained on $\sup{\vec{b}}$.
\end{lem}

\begin{proof}
Consider the set $S$ of all entries of $\vec{b}$ attaining the maximum value $v$.  If this set contains every vertex, it automatically overlaps $\sup{\vec{b}}$.  Otherwise, we have:
\[
\sum_{i \in S} \vec{b}_i = \left ( \sum_{i \in S, j \not \in S} w_{ji} \vec{x}_j \right ) - \left ( \sum_{i \in S, j \not \in S} w_{ij} \vec{x}_i \right )
< v \left ( \sum_{i \in S, j \not \in S} w_{ji} \right ) - v \left ( \sum_{i \in S, j \not \in S} w_{ij} \right )
= 0.
\]
Here the last equality is due to $\mlap$ being Eulerian:
the total weight entering and leaving $S$ is equal.
The sum of the entries from $S$ in $\vec{b}$ is strictly negative and thus nonzero, so $S$ must overlap with the support of $\vec{b}$.

The case of the minimum value is analogous.
\end{proof}

Next, we show that if the demands for an Eulerian Laplacian system are supported on a well-connected subset of the graph, perturbing the system by a small multiple of the identity matrix cannot induce too much error.
\begin{lem}
\label{lem:onepiece}
Let $\mlap$ be a connected Eulerian Laplacian with corresponding undirected Laplacian $\mU_\mlap$, and let $S$ be a subset of the vertices such that any two vertices in $S$ are connected by a path in $\mU_\mlap$ containing only edges of weight at least $\beta$.  Let $\vec{b}$ be a vector in the image of $\mlap$ supported on $S$, and let $\alpha > 0$.  Then
\begin{equation*}
\normFull{(\mlap + \alpha \mI)^{-1} \vec{b} - \mlap^\dagger \vec{b}}_{\mU_\mlap} \leq n \sqrt{\frac{\alpha}{\beta}} \normFull{\mlap^\dagger \vec{b}}_{\mU_\mlap}.
\end{equation*}
\end{lem}

\begin{proof}
We define $\vec{x} = \mlap^\dagger \vec{b}$ and $\vec{y} = (\mlap + \alpha \mI)^{-1} \vec{b}$.
Writing $\vec{y}$
as $\vec{x} + (\mlap + \alpha \mI)^{-1} (\vec{b} - (\mlap + \alpha \mI) \vec{x})$
gives:
\[
\vec{y} - \vec{x}
= \left(\mlap + \alpha \mI \right)^{-1}
\left(\mlap \vec{x} - \left(\mlap + \alpha \mI\right) \vec{x}\right)
= -\alpha \left(\mlap + \alpha \mI \right)^{-1} \vec{x}.
\]
which when substituted into the $\mu_{\mlap}$ norm gives:
\[
\normFull{\vec{y} - \vec{x}}_{\mU_\mlap}^2
= \alpha^2 \vec{x}^\top
\left(\mlap^\top + \alpha \mI\right)^{-1}
\mU_\mlap \left(\mlap + \alpha \mI\right)^{-1} \vec{x}.
\]
Since $\mU_\mlap  \preceq \mU_{\mlap} + \alpha \mI$,
we can invoke Lemma~\ref{lem:symBound} with the matrix
$(\mlap + \alpha \mI)$ to get:
\[
\normFull{\vec{y} - \vec{x}}_{\mU_\mlap}^2
\leq \frac{\alpha^2}{\alpha} \vec{x}^\top \vec{x}
= n \alpha \norm{\vec{x}}_\infty^2.
\]

Now, by Lemma~\ref{lem:voltages}, the full range of the
entries of $\vec{x}$ occurs within $S$.
The existence of a path with at most $n$ vertices connecting
the minimum and maximum of these entries implies that
\[
\normFull{\vec{x}}_{\mU_\mlap}^2 \geq \frac{\beta}{n} \normFull{\vec{x}}_\infty^2.
\]
Putting this together we get
\[
\normFull{\vec{y} - \vec{x}}_{\mU_\mlap}^2
\leq n^2 \frac{\alpha}{\beta} \normFull{\vec{x}}_{\mU_\mlap}^2.
\]
\end{proof}

Next, we handle the case of simultaneous demands within multiple well-connected components.
We also switch the error bound to be in
$\norm{\vec{b}}_{{\mU_\mlap}^\dagger}$ to facilitate
our later steps.
\begin{lem}
\label{lem:addclique}
Let $\mlap$ be a connected Eulerian Laplacian with corresponding undirected Laplacian $\mU_\mlap$, and let $S_1, S_2, ... S_k$ be the connected components of the graph consisting of those edges in $\mU_\mlap$ with edges of weight at least $\beta$.  Let $\vec{b}$ be a vector such that $\sum_{i \in S_j} \vec{b}_i = 0$ for all $j$, and let $\alpha > 0$.  Then
\begin{equation*}
\normFull{(\mlap + \alpha \mI)^{-1} \vec{b} - \mlap^\dagger \vec{b}}_{\mU_\mlap}
\leq 2 n^{7/2} \sqrt{\frac{\alpha}{\beta}} \normFull{\vec{b}}_{{\mU_\mlap}^\dagger}.
\end{equation*}
\end{lem}

\begin{proof}
We decompose $\vec{b}$ as
\begin{equation*}
\vec{b} = \sum_j \vec{b_j}
\end{equation*}
where $\vec{b_j}$ is supported on $S_j$, and $0$ everywhere else.

Now we aim to bound $\norm{\vec{b_j}}_{{\mU_\mlap}^\dagger}$.
Note that there is an electrical flow $\vec{y}$ on $\mU_\mlap$
with energy $\norm{\vec{b}}_{{\mU_\mlap}^\dagger}^2$ routing
overall demands $\vec{b}$.
We define $\vec{y}'_j$ as the restriction of the flow to the
internal edges of $S_j$, and let is residuals be $\vec{b'_j}$.
Since this flow is on a subset of the edges, its total
energy is almost most $\norm{\vec{b}}_{{\mU_\mlap}^\dagger}^2$
so this certifies that $\norm{\vec{b'_j}}_{{\mU_\mlap}^\dagger} \leq \norm{\vec{b}}_{{\mU_\mlap}^\dagger}$.
Then
\[
\normFull{\vec{b'_j}-\vec{b_j}}_1 \leq \frac{n}{\sqrt{\beta}} \normFull{\vec{b}}_{{\mU_\mlap}^\dagger}
\]
since it is at most the $\ell_1$ norm of the flows on the edges incident to but not contained in $S_j$ in $y$,
and each of these $\leq n^2$ edges has weight at most $\beta$
by the assumption of $S_j$ being the connected components
on edges with weights $\geq \beta$.
$\vec{b'_j}-\vec{b_j}$ is also supported on $S_j$;
since $S_j$ is connected by edges of weight $\geq \beta$, we have
\[
\normFull{\vec{b'_j}-\vec{b_j}}_{{\mU_\mlap}^\dagger}
\leq \sqrt{n \beta} \normFull{\vec{b'_j}-\vec{b_j}}_1
\leq n^{3/2} \normFull{\vec{b}}_{{\mU_\mlap}^\dagger}.
\]
Then by the triangle inequality
\[
\normFull{\vec{b_j}}_{{\mU_\mlap}^\dagger}
\leq \normFull{\vec{b'_j}}_{{\mU_\mlap}^\dagger} + \normFull{\vec{b'_j}-\vec{b_j}}_{{\mU_\mlap}^\dagger}
\leq 2 n^{3/2} \normFull{\vec{b}}_{{\mU_\mlap}^\dagger}.
\]
Lemma~\ref{lem:sym-hsm} then implies that $\norm{\mlap^\dagger \vec{b_j}}_{\mU} \leq \norm{\vec{b_j}}_{{\mU_\mlap}^\dagger}$.  Finally, we apply Lemma~\ref{lem:onepiece} to each $\vec{b_j}$, yielding that
\begin{equation*}
\normFull{(\mlap + \alpha \mI)^{-1} \vec{b_j} - \mlap^\dagger \vec{b_j}}_{\mU_\mlap} \leq 2 n^{5/2} \sqrt{\frac{\alpha}{\beta}} \normFull{\vec{b}}_{{\mU_\mlap}^\dagger}.
\end{equation*}
Summing over the up to $n$ different $\vec{b_j}$ and applying the triangle inequality over this sum gives
\begin{equation*}
\normFull{(\mlap + \alpha \mI)^{-1} \vec{b} - \mlap^\dagger \vec{b}}_{\mU_\mlap}
\leq 2 n^{7/2} \sqrt{\frac{\alpha}{\beta}} \normFull{\vec{b}}_{{\mU_\mlap}^\dagger}
\end{equation*}
as desired.
\end{proof}

We now have the tools to analyze $\textsc{CrudeSolveIllConditioned}$. 
Our analyses rely on the following key intermediate variables:
\begin{enumerate}
\item
\[
\vec{q}^{(i)} \defeq {\mc^{(i)}}^\top \left ( \mc^{(i)} \mlap {\mc^{(i)}}^\top \right )^\dagger \mc^{(i)} \vec{b}^{(i)}.
\]
\item
\[
\vec{e}^{(i)} \defeq \vec{z}^{(i)}-\vec{q}^{(i)},
\]
where $\vec{z}^{(i)}$ is the `shifted' solution obtained
on Step~\ref{step:approxSolveProjected}.
\item
\[
\vec{f}^{(i)} \defeq  \operatorname{Proj}^{(i)}\left(\mlap \vec{e}^{(i)}\right).
\]
\item
\[
\vec{b}^{*(i)} \defeq \vec{b} - \sum_{j<i} \vec{f}^{(j)}.
\]
\end{enumerate}

We first show that the right hand side in the iterations
can be expressed as a close form involving the errors.
\begin{lemma}
\label{lem:asproj}
For all $i$,
\begin{equation*}
\vec{b}^{(i)} = \left ( \mI - \mlap {\mc^{(i-1)}}^\top \left ( \mc^{(i-1)} \mlap {\mc^{(i-1)}}^\top \right )^\dagger \mc^{(i-1)} \right ) \vec{b}^{*(i)}.
\end{equation*}
\end{lemma}
\begin{proof}
First we note two equations about the projection operator $\left ( \mI - \mlap {\mc^{(i-1)}}^\top \left ( \mc^{(i-1)} \mlap {\mc^{(i-1)}}^\top \right )^\dagger \mc^{(i-1)} \right )$.

Now we prove this by induction.
The base case of $i = 0$ follows from the two sides being identical.
For the inductive case,
substituting in the construction of $\vec{b}^{(i + 1)}$
on Step~\ref{step:bNext} gives:
\begin{align*}
\vec{b}^{(i+1)} &= \operatorname{Proj}^{(i)} (\vec{b}^{(i)} - \mlap \vec{z}^{(i)}) \\
&= \operatorname{Proj}^{(i)} \left ( \vec{b}^{(i)} - \mlap {\mc^{(i)}}^\top \left ( \mc^{(i)} \mlap {\mc^{(i)}}^\top \right )^\dagger \mc^{(i)} \vec{b}^{(i)} - \mlap \vec{e}^{(i)} \right ) \\
&= \operatorname{Proj}^{(i)} \left ( \left ( \mI - \mlap {\mc^{(i)}}^\top \left ( \mc^{(i)} \mlap {\mc^{(i)}}^\top \right )^\dagger \mc^{(i)} \right ) \vec{b}^{(i)} \right ) - \operatorname{Proj}^{(i)} ( \mlap \vec{e}^{(i)} ) \\
&= \left ( \mI - \mlap {\mc^{(i)}}^\top \left ( \mc^{(i)} \mlap {\mc^{(i)}}^\top \right )^\dagger \mc^{(i)} \right ) \vec{b}^{(i)} - \vec{f}^{(i)}.
\end{align*}
Here, the last line follows from the fact that $\left ( \mI - \mlap {\mc^{(i)}}^\top \left ( \mc^{(i)} \mlap {\mc^{(i)}}^\top \right )^\dagger \mc^{(i)} \right ) \vec{b}^{(i)}$ is already in the kernel of $\mc^{(i)}$, as
\begin{align*}
\mc^{(i)} \left ( \mI - \mlap {\mc^{(i)}}^\top \left ( \mc^{(i)} \mlap {\mc^{(i)}}^\top \right )^\dagger \mc^{(i)} \right ) \vec{b}^{(i)} &= \mc^{(i)} \vec{b}^{(i)} - \left ( \mc^{(i)} \mlap {\mc^{(i)}}^\top \right ) \left ( \mc^{(i)} \mlap {\mc^{(i)}}^\top \right )^\dagger \mc^{(i)} \vec{b}^{(i)} \\
&= \mc^{(i)} \vec{b}^{(i)} - \mc^{(i)} \vec{b}^{(i)} \\
&= 0
\end{align*}
We similarly have
\begin{align*}
\left ( \mI - \mlap {\mc^{(i)}}^\top \left ( \mc^{(i)} \mlap {\mc^{(i)}}^\top \right )^\dagger \mc^{(i)} \right ) \mlap {\mc^{(i)}}^\top &= \mlap {\mc^{(i)}}^\top - \mlap {\mc^{(i)}}^\top \left ( \mc^{(i)} \mlap {\mc^{(i)}}^\top \right )^\dagger \left ( \mc^{(i)} \mlap {\mc^{(i)}}^\top \right ) \\
&= \mlap {\mc^{(i)}}^\top - \mlap {\mc^{(i)}}^\top \\
&= 0
\end{align*}
Since the image of ${\mc^{(i-1)}}^\top$ is contained in the image of ${\mc^{(i)}}^\top$, this also implies that
\begin{equation*}
\left ( \mI - \mlap {\mc^{(i)}}^\top \left ( \mc^{(i)} \mlap {\mc^{(i)}}^\top \right )^\dagger \mc^{(i)} \right ) \mlap {\mc^{(i-1)}}^\top = 0
\end{equation*}
and hence that
\begin{align*}
\hspace{8em}&\hspace{-8em} \left ( \mI - \mlap {\mc^{(i)}}^\top \left ( \mc^{(i)} \mlap {\mc^{(i)}}^\top \right )^\dagger \mc^{(i)} \right ) \left ( \mI - \mlap {\mc^{(i-1)}}^\top \left ( \mc^{(i-1)} \mlap {\mc^{(i-1)}}^\top \right )^\dagger \mc^{(i-1)} \right ) \\
&= \left ( \mI - \mlap {\mc^{(i)}}^\top \left ( \mc^{(i)} \mlap {\mc^{(i)}}^\top \right )^\dagger \mc^{(i)} \right ).
\end{align*}

We also have $\left ( \mI - \mlap {\mc^{(i)}}^\top \left ( \mc^{(i)} \mlap {\mc^{(i)}}^\top \right )^\dagger \mc^{(i)} \right ) \vec{f}^{(i)} = \vec{f}^{(i)}$ as $\vec{f}^{(i)}$, output by $\operatorname{Proj}^{(i)}$, is in the kernel of $\mc^{(i)}$.  Putting these together and substituting in
the induction hypothesis on $i$ gives
\begin{align*}
\vec{b}^{(i+1)} &= \left ( \mI - \mlap {\mc^{(i)}}^\top \left ( \mc^{(i)} \mlap {\mc^{(i)}}^\top \right )^\dagger \mc^{(i)} \right ) \left ( \vec{b}^{*(i)} - \vec{f}^{(i)} \right ) \\
&= \left ( \mI - \mlap {\mc^{(i)}}^\top \left ( \mc^{(i)} \mlap {\mc^{(i)}}^\top \right )^\dagger \mc^{(i)} \right ) \vec{b}^{*(i+1)}
\end{align*}
which shows that the identity holds for $i + 1$ as well.
\end{proof}

\begin{lemma}
\label{lem:answerbound}
For all $i$,
\begin{equation*}
\vec{x}^{(i)} + \mlap^\dagger \vec{b}^{(i)} =
\sum_{j<i} \vec{e}^{(j)} +
\mlap^\dagger \vec{b}^{*(i)}.
\end{equation*}
\end{lemma}
\begin{proof}
Again, we proceed by induction.
\begin{align*}
\vec{x}^{(i+1)} + \mlap^\dagger \vec{b}^{(i+1)} &= \vec{x}^{(i)} + \vec{z}^{(i)} + \mlap^\dagger \operatorname{Proj}^{(i)} \left ( \vec{b}^{(i)} - \mlap \vec{z}^{(i)} \right ) \\
&= \vec{x}^{(i)} + \vec{q}^{(i)} + \vec{e}^{(i)} + \mlap^\dagger \operatorname{Proj}^{(i)} \left ( \vec{b}^{(i)} - \mlap \vec{q}^{(i)} - \mlap \vec{e}^{(i)} \right ) \\
&= \vec{x}^{(i)} + \vec{q}^{(i)} + \vec{e}^{(i)} + \mlap^\dagger \operatorname{Proj}^{(i)} \left ( \vec{b}^{(i)} - \mlap \vec{q}^{(i)} \right ) - \mlap^\dagger \operatorname{Proj}^{(i)}\left(\mlap \vec{e}^{(i)}\right) \\
&= \vec{x}^{(i)} + \vec{q}^{(i)} + \vec{e}^{(i)} + \mlap^\dagger \vec{b}^{(i)} - \vec{q}^{(i)} - \mlap^\dagger \vec{f}^{(i)}.
\end{align*}
The last line here follows from the fact that $\vec{b}^{(i)} - \mlap \vec{q}^{(i)}$ is already inside the kernel of $\mc^{(i)}$.  Cancelling the $\vec{q}^{(i)}$ terms and applying the induction hypothesis then gives
\begin{align*}
\vec{x}^{(i+1)} + \mlap^\dagger \vec{b}^{(i+1)} &= \vec{x}^{(i)} + \mlap^\dagger \vec{b}^{(i)} + \vec{e}^{(i)} - \mlap^\dagger \vec{f}^{(i)} \\
&= \mlap^\dagger \vec{b}^{*(i)} + \sum_{j<i} \vec{e}^{(j)} + \vec{e}^{(i)} - \mlap^\dagger \vec{f}^{(i)} \\
&= \mlap^\dagger \vec{b}^{*(i+1)} + \sum_{j<i+1} \vec{e}^{(j)}.
\end{align*}
\end{proof}

These relations then allows us to bound the global
error via the guarantees of the separate approximate solves.
As these solves produce solutions with relative error,
we first need to bound the norm of $\vec{b}^{(i)}$:
\begin{lemma}
\label{lem:bstable}
For all $i$,
\[
\normFull{\vec{b}^{(i)}}_{{\mU_\mlap}^\dagger} \leq
3 n \normFull{\vec{b}^{*(i)}}_{{\mU_\mlap}^\dagger}.
\]
\end{lemma}
\begin{proof}
First, note that
\begin{equation*}
\normFull{\mc^{(i-1)} \vec{b}^{*(i)}}
_{\left ( \mc^{(i-1)} \mU_\mlap {\mc^{(i-1)}}^\top \right )^\dagger} 
\leq \normFull{\vec{b}^{*(i)}}_{{\mU_\mlap}^\dagger}.
\end{equation*}
Lemma~\ref{lem:symBound} then gives
\begin{equation*}
\normFull{\left ( \mc^{(i-1)} \mlap {\mc^{(i-1)}}^\top \right )^\dagger \mc^{(i-1)} \vec{b}^{*(i)}}_{\left ( \mc^{(i-1)} \mU_\mlap {\mc^{(i-1)}}^\top \right )} \leq \normFull{\vec{b}^{*(i)}}_{{\mU_\mlap}^\dagger}
\end{equation*}
or equivalently
\begin{equation*}
\normFull{{\mc^{(i-1)}}^\top \left ( \mc^{(i-1)} \mlap {\mc^{(i-1)}}^\top \right )^\dagger \mc^{(i-1)} \vec{b}^{*(i)}}_{\mU_\mlap}
\leq \normFull{\vec{b}^{*(i)}}_{{\mU_\mlap}^\dagger}.
\end{equation*}
Now applying Lemma~\ref{lem:symBound}, we get
\begin{equation*}
\normFull{\mlap {\mc^{(i-1)}}^\top \left ( \mc^{(i-1)} \mlap {\mc^{(i-1)}}^\top \right )^\dagger \mc^{(i-1)} \vec{b}^{*(i)}}_{{\mU_\mlap}^\dagger} \leq 2 n \normFull{\vec{b}^{*(i)}}_{{\mU_\mlap}^\dagger}.
\end{equation*}
The result then follows from the triangle inequality with $\vec{b}^{*(i)}$.
\end{proof}

For the next step we will use the following fact about matrices:
\begin{fact}
\label{fact:schur}
For any symmetric positive semidefinite matrix $\mm$ and arbitrary matrix $\mc$, and for any vector $\vec{v}$ in the image of $\mm$,
\[
\normFull{\mc \vec{v}}_{\left ( \mc \mm \mc^\top \right )^\dagger} \leq \normFull{\vec{v}}_{\mm^\dagger}.
\]
\end{fact}
Notably, this fact is equivalent to the standard result that the Schur complement of a positive semidefinite matrix is spectrally dominated by the matrix.
\begin{proof}
We can prove this using duality of norms:
\begin{align*}
\normFull{\mc \vec{v}}_{\left ( \mc \mm \mc^\top \right )^\dagger} &= \max_{\normFull{\vec{u}}_{\left ( \mc \mm \mc^\top \right )} \leq 1} \left \langle u, \mc v \right \rangle \\
&= \max_{\normFull{\mc^\top \vec{u}}_{\mm} \leq 1} \left \langle \mc^\top u, v \right \rangle \\
&\leq \max_{\norm{\vec{u}'}_{\mm} \leq 1} \left \langle u', v \right \rangle \\
&= \normFull{\vec{v}}_{\mm^\dagger}.
\end{align*}
\end{proof}

We begin to bound the error terms:
\begin{lemma}
\label{lem:errore}
\[
\normFull{\vec{e}^{(i)}}_{\mU_\mlap} \leq \frac{12 n^{9/2}}{r} \normFull{\vec{b}^{*(i)}}_{{\mU_\mlap}^\dagger}.
\]
\end{lemma}

\begin{proof}
First, we note that by Fact~\ref{fact:schur}
\begin{equation*}
\normFull{\mc^{(i)} \vec{b}^{(i)}}_{\left ( \mc^{(i)} \mU_\mlap {\mc^{(i)}}^\top \right )^\dagger} \leq \normFull{\vec{b}^{(i)}}_{{\mU_\mlap}^\dagger}
\leq 3 n \normFull{\vec{b}^{*(i)}}_{{\mU_\mlap}^\dagger},
\end{equation*}
where the last inequality is by Lemma~\ref{lem:bstable}.
Now, define intermediate variables:
\begin{enumerate}
\item
\[
\vec{q'}^{(i)}  \defeq \left ( \mc^{(i)} \mlap {\mc^{(i)}}^\top \right )^\dagger \mc^{(i)} \vec{b}^{(i)},
\]
\item
\[
\vec{z''}^{(i)}  \defeq \left ( \mc^{(i)} \mlap {\mc^{(i)}}^\top + \frac{w^{(i)}}{r^2} \mI \right )^{-1} \mc^{(i)} \vec{b}^{(i)},
\]
and $\vec{z'}^{(i)}$ analogously, but as the output of a $\frac{1}{r}$-approximate solver
for the system
$(\mc^{(i)} \mlap {\mc^{(i)}}^\top + \frac{w^{(i)}}{r^2} \mI ) \vec{x}
= \mc^{(i)} \vec{b}^{(i)}$.
\end{enumerate}

By rearranging and applying the triangle inequality, we have
\begin{equation*}
\normFull{\vec{e}^{(i)}}_{\mU_\mlap} \leq \normFull{\vec{z'}^{(i)}-\vec{z''}^{(i)}}_{\left ( \mc^{(i)} \mU_\mlap {\mc^{(i)}}^\top \right )}+\normFull{\vec{z''}^{(i)}-\vec{q'}^{(i)}}_{\left ( \mc^{(i)} \mU_\mlap {\mc^{(i)}}^\top \right )}
\end{equation*}
The first term captures the error induced by using the approximate rather than exact solver, while the second captures the error induced from adding the multiple of the identity (the more serious issue).

For the first term, we have
\begin{align*}
\normFull{\vec{z'}^{(i)}-\vec{z''}^{(i)}}_{\left ( \mc^{(i)} \mU_\mlap {\mc^{(i)}}^\top \right )} &\leq \normFull{\vec{z'}^{(i)}-\vec{z''}^{(i)}}_{\left ( \mc^{(i)} \mU_\mlap {\mc^{(i)}}^\top + \frac{w^{(i)}}{r^2} \mI \right )} \\
&\leq \frac{1}{r} \normFull{\vec{z''}^{(i)}}_{\left ( \mc^{(i)} \mU_\mlap {\mc^{(i)}}^\top + \frac{w^{(i)}}{r^2} \mI \right )} \textrm{ (by definition of approximate solver)} \\
&\leq \frac{1}{r} \normFull{\mc^{(i)} \vec{b}^{(i)}}_{\left ( \mc^{(i)} \mU_\mlap {\mc^{(i)}}^\top + \frac{w^{(i)}}{r^2} \mI \right )^{-1}} \textrm{ (by Lemma~\ref{lem:symBound})} \\
&\leq \frac{1}{r} \normFull{\mc^{(i)} \vec{b}^{(i)}}_{\left ( \mc^{(i)} \mU_\mlap {\mc^{(i)}}^\top \right )^{\dagger}} \\
&\leq \frac{3 n}{r} \normFull{\vec{b}^{*(i)}}_{{\mU_\mlap}^\dagger}.
\end{align*}

For the second term, we will apply Lemma~\ref{lem:addclique}.  We will use the fact that $\vec{b}^{(i)}$ is in the image of $\mc^{(i-1)}$--or equivalently, that its entries on any connected component of the edges in $\mU_\mlap$ with weight $\geq w^{(i-1)}$ sum to 0.  By the definition of $w^{(i)}$, these are the same as the edges with weight $\geq \frac{w^{(i)}}{2}$.  Furthermore, contracting can only increase the connectivity of a component, so $\mc^{(i)} \vec{b}^{(i)}$ satisfies the same property relative to $\left ( \mc^{(i)} \mU_\mlap {\mc^{(i)}}^\top \right )$.  Then we can apply Lemma~\ref{lem:addclique} with $\alpha = \frac{w^{(i)}}{r^2}$ and $\beta = \frac{w^{(i)}}{2}$:
\[
\normFull{\vec{z''}^{(i)}-\vec{q'}^{(i)}}_{\left ( \mc^{(i)} \mU_\mlap {\mc^{(i)}}^\top \right )}
\leq \frac{3 n^{7/2}}{r} \normFull{\mc^{(i)} \vec{b}^{(i)}}_{\left ( \mc^{(i)} \mU_\mlap {\mc^{(i)}}^\top \right )^\dagger}
\leq \frac{9 n^{9/2}}{r} \normFull{\vec{b}^{*(i)}}_{{\mU_\mlap}^\dagger}.
\]

Summing these two bounds gives the desired result.
\end{proof}

It remains to bound the norms of the other error
vectors $\vec{f}^{(i)}$.

\begin{lemma}
\label{lem:errorf}
For all $i$,
\[
\normFull{\vec{f}^{(i)}}_{{\mU_\mlap}^\dagger}
\leq 6 n^{5/2} \normFull{\vec{e}^{(i)}}_{\mU_\mlap}.
\]
\end{lemma}
\begin{proof}
First, we apply Lemma~\ref{lem:symBound}, showing that $\norm{\mlap \vec{e}^{(i)}}_{{\mU_\mlap}^\dagger} \leq 2 n \norm{\vec{e}^{(i)}}_{\mU_\mlap}$.

Now we will show that the $\operatorname{Proj}^{(i)}$ operator cannot increase the ${\mU_\mlap}^\dagger$ norm by more than a factor of $2 n^{3/2}$.

The proof is similar to that of Lemma~\ref{lem:addclique}:
we consider the electrical flow $\mU_\mlap$ that meets
the demands $\mlap \vec{e}^{(i)}$.
We denote this flow with $\vec{y}^{(i)}$, and define $\vec{y}^{(i)'}$
as the restriction of $\vec{y}^{(i)}$ to the edges of weight
$\geq w^{(i)}$.
We then let the residue of this flow be $\vec{b}^{(i)'}$, and write:
\begin{equation*}
\operatorname{Proj}^{\left(i\right)}\left(\mlap \vec{e}^{(i)}\right) = \operatorname{Proj}^{\left(i\right)}\left(\vec{b}^{(i)'}\right) +
	\operatorname{Proj}^{\left(i\right)}
		\left(\mlap \vec{e}^{(i)}-\vec{b}^{(i)'}\right).
\end{equation*}
Since $\vec{b}^{(i)'}$ is induced by a flow $\vec{y}^{(i)'}$
wholly within the components with weights $\geq w^{(i)}$,
\[
\operatorname{Proj}^{(i)}(\vec{b}^{(i)'}) = \vec{b}^{(i)'}.
\]
Furthermore, since $\vec{y}^{(i)'} $ rounds the demand
$\vec{b}^{(i)'}$, and is a restriction of $\vec{y}$, we have
\[
\normFull{\vec{b}^{(i)'}}_{{\mU_\mlap}^\dagger}^2
\leq \mathcal{E}_{\mU_{\mlap}}\left(\vec{y}^{(i)'}\right)
\leq \mathcal{E}_{\mU_{\mlap}}\left(\vec{y}^{(i)}\right)
\leq \normFull{\mlap \vec{e}^{(i)}}_{{\mU_\mlap}^\dagger}^2,
\]
where $\mathcal{E}_{\mU_{\mlap}}(\vec{y})$ denotes the electrical
energy of the flow $\vec{y}$ on $\mu_{\mlap}$.

On the other hand, $\mlap \vec{e}^{(i)}-\vec{b}^{(i)'}$ is the residual
of the flow $\vec{y}^{(i)}-\vec{y}^{(i)'}$,
which is supported on edges with weight $< w^{(i)}$
and also has energy at most
$\norm{\mlap \vec{e}^{(i)}}_{{\mU_\mlap}^\dagger}^2$.
Since each edge can contribute to the residual of at
most two vertices, we have
\[
\normFull{\mlap \vec{e}^{(i)}-\vec{b}^{(i)'}}_1
\leq 2 \normFull{\vec{y}^{(i)}-\vec{y}^{(i)'}}_1
\leq 2 n \normFull{\vec{y}^{(i)}-\vec{y}^{(i)'}}_2
\leq \frac{2n}{\sqrt{w^{(i)}}} \normFull{\mlap \vec{e}^{(i)}}_{{\mU_\mlap}^\dagger}.
\]
Finally, using the fact that $\operatorname{Proj}^{(i)}$ can at most double the $\ell_1$ norm of its input and that the ${\mU_\mlap}^\dagger$ norm of demands connected by edges of weight at least $w^{(i)}$ is at most $\sqrt{n w^{(i)}}$ times the $\ell_1$ norm of those demands, we have
\begin{equation*}
\normFull{\operatorname{Proj}^{\left(i\right)}
	\left(\mlap \vec{e}^{\left(i\right)}-\vec{b}^{(i)'}\right)}_{\mU_\mlap}^\dagger
\leq 4n^{3/2} \normFull{\mlap \vec{e}^{\left(i\right)}}_{{\mU_\mlap}^\dagger}.
\end{equation*}
Applying the triangle inequality to $\vec{b}^{(i)'}$ and
$\mlap \vec{e}^{(i)}-\vec{b}^{(i)'}$ then gives the desired result.
\end{proof}

Note that combining the previous two lemmas shows that 
\[
\normFull{\vec{f}^{(i)}}_{{\mU_\mlap}^\dagger}
\leq \frac{72 n^7}{r} \normFull{\vec{b}^{*(i)}}_{{\mU_\mlap}^\dagger}.
\]
Putting these together with the breakdown of errors
then gives the overall guarantees.

\begin{proof}[Proof of Theorem~\ref{thm:reduction}]
First, note that the number of rounds is bounded by
$\min\{n^2, O(\log(n \kappa))\}$.
The former is from the number of edges, while the latter follows
from the fact that the largest and smallest eigenvalues of
$\mU_\mlap$ are within $\poly(n)$ factors of the smallest
and largest weighted vertex degrees.

This then implies by induction that $\norm{\vec{b}^{*(i)}}_{{\mU_\mlap}^\dagger} \leq 2 \norm{\vec{b}}_{{\mU_\mlap}^\dagger}$ for all $i$ -- since, assuming it held for all previous $i$, each of the at most $n^2$ error terms $\vec{f}^{(j)}$ had norm at most $\frac{1}{5 n^3} \norm{\vec{b}^{*(i)}}_{{\mU_\mlap}^\dagger}$ (by Lemmas~\ref{lem:errore} and \ref{lem:errorf}).  By Lemma~\ref{lem:errore} each $\vec{e}^{(i)}$ had norm at most $\frac{1}{40 n^{11/2}}$.

Then applying Lemma~\ref{lem:answerbound} on the final configuration (where $\vec{b}^{(i+1)} = 0$) with these bounds (and again using the fact that there are most $n^2$ iterations) implies that
\begin{align*}
\normFull{\vec{x}-\mlap^\dagger \vec{b}}_{\mU_{\mlap}}
&\leq \sum_i \left(\normFull{\vec{e}^{\left(i\right)}}_{\mU_\mlap} +
	 \normFull{\mlap^\dagger \vec{f}^{\left(i\right)}}_{\mU_\mlap}\right) \\
&\leq \sum_i \left(\normFull{\vec{e}^{\left(i\right)}}_{\mU_\mlap} + \normFull{\vec{f}^{\left(i\right)}}_{\mU^\dagger}\right) \\
&\leq \frac{1}{4n} \normFull{\vec{b}}_{{\mU_\mlap}^\dagger} \\
&\leq \frac{1}{2} \normFull{\mlap^\dagger \vec{b}}_{\mU_{\mlap}}.
\end{align*}
Here the second and last inequalities follow
from Lemma~\ref{lem:symBound}.
This is the desired bound on the final error of the solver.

Finally, we need to show that the procedure can be implemented
in the desired runtime.
We note that the contracted matrices
$( \mc^{(i)} \mlap {\mc^{(i)}}^\top )$
are still Eulerian Laplacians,
and the contractions cannot increase the number of vertices or edges.
 The actual systems solved are in
$ ( \mc^{(i)} \mlap {\mc^{(i)}}^\top + \frac{w^{(i)}}{r^2} \mI )$,
or an Eulerian Laplacian plus positive diagonal.
This matrix can be reduced, with the reduction in Section 5 of~\cite{cohen2016faster},
to solving an Eulerian Laplacian with asymptotically the same sparsity and condition number.
The symmetrized matrix for each system has min eigenvalue at
least $\frac{w^{(i)}}{r^2}$ and max eigenvalue at most
$O(n w^{(i)})$, so all condition numbers of their symmetrizations
are polynomially bounded, as desired.
\end{proof}

\end{document}